\documentclass[12pt]{article}
\usepackage{amsmath, amsthm, amssymb, amsfonts,bigints}
\usepackage{natbib,comment}
\usepackage{setspace}
\usepackage{graphicx}
\usepackage{caption,subcaption}
\usepackage{color}
\usepackage{xcolor,mathtools}
\usepackage{afterpage,array,tabularx,makecell,multirow}
\usepackage{appendix}
\usepackage{float}
\usepackage{epstopdf}
\usepackage{placeins,enumitem}
\usepackage{units}
\usepackage[margin=1.25in]{geometry}
\usepackage{booktabs}
\usepackage{threeparttable}
\usepackage{float}
\usepackage{appendix}
\floatstyle{plaintop}
\restylefloat{table}
\restylefloat{figure}
\usepackage{titlesec}
\usepackage[bottom]{footmisc}
\usepackage{tikz}

\newif\ifkeepremark

\newcounter{subassumption}[asu]

\makeatletter
\renewcommand{\p@subassumption}{\theasu}
\makeatother

\newif \ifshowexplanations
\showexplanationstrue

\newtheorem{condition}{Condition}

\newtheorem{example}{Example}

\newtheorem{lemma}{Lemma}

\newtheorem{proposition}[condition]{Proposition}

\theoremstyle{theorem}
\newtheorem{remark}{Remark}
\theoremstyle{definition}

\theoremstyle{definition}
\newtheorem{assumption}{Assumption}

\usepackage{chngcntr}
\usepackage{apptools}

\theoremstyle{plain} 
\newtheorem*{continuancex}{Example \continuanceref~(continued)}
\newenvironment{continuance}[1]
{\newcommand{\continuanceref}{\ref{#1}}\begin{continuancex}}
	{\end{continuancex}}

\newcommand{\E}{\mathbb{E}}

\definecolor{JuliaOrange}{HTML}{E36F47}
\definecolor{JuliaGreen}{HTML}{3EA44E}
\definecolor{dark_green}{rgb}{0,0.75,0}
\newif \ifshowexplanationsGiovanni
\showexplanationsGiovannitrue

\definecolor{dark_blue}{rgb}{0.0, 0.0, 0.55}
\newif \ifshowexplanationsTim
\showexplanationsTimtrue

\usepackage[colorlinks=true,
linkcolor=blue,
citecolor=blue,
urlcolor=blue]{hyperref}

\defcitealias{compiani2025demand}{CMS}

\usepackage{bibunits}
\defaultbibliographystyle{ecta} 
\defaultbibliography{references}

\begin{document}

\begin{bibunit}

\onehalfspacing 

\date{June 20, 2026}
\title{From Unstructured Data to Demand Counterfactuals: Theory and Practice\thanks{%
We are grateful to Steve Berry, Phil Haile, JF Houde, Jonathon McClure, Francesca Molinari, Takeaki Sunada and seminar and conference participants at Amazon, Arizona, ASSA 2026, BU, Chicago Booth, Cornell, Duke, EIEF, Montreal, Pittsburgh, Sciences Po, Seoul National, Toronto, Toulouse, UCLA, USC, Wisconsin, and Yale. This material is based upon work supported by the National Science Foundation under Award No.~2521471 (Christensen).
}}
\author{Timothy Christensen\thanks{%
Department of Economics, Yale University. \texttt{timothy.christensen@yale.edu}} 
\quad \quad
Giovanni Compiani\thanks{%
Booth School of Business, University of Chicago. \texttt{giovanni.compiani@chicagobooth.edu}}
}

\maketitle
\thispagestyle{empty}

\begin{abstract}  
\singlespacing
\noindent 
Empirical models of multi-product demand rely on low-dimensional product representations to capture substitution patterns, increasingly using proxies built from unstructured data. When proxies are imperfect, standard workflows yield biased counterfactuals and invalid inference. We develop a practical toolkit to address these issues. Our methods apply to market-level and/or individual data, require minimal additional computation, provide simple standard-error formulas, and accommodate proxies from fine-tuned models. Further, we propose diagnostics to assess proxy quality. Our methods yield meaningful improvements in predicting substitution in empirically calibrated simulations and in an application where we assess counterfactual prediction performance against a ground truth.

\vspace{15pt}
\end{abstract}

\clearpage
\setcounter{page}{1}

\pagenumbering{arabic}

\newpage

\section{Introduction}

Estimating demand for differentiated products is central to many fields of economics and marketing. A common strategy is to specify the utility of a product as a function of the product's price and other observable attributes, often allowing for rich types of consumer heterogeneity (\cite{BLP_1995,BLP_micro_2004}, henceforth BLP). Among other applications, this approach has been used to study the impact of horizontal mergers \citep{nevo2000mergers}, new product launches \citep{hausman1994valuation, petrin2002quantifying}, trade policy \citep{goldberg1995product}, school choice \citep{bayer2007unified, neilson2017targeted}, two-sided markets \citep{fan2013ownership, lee2013vertical}, the evolution of markups over time \citep{grieco2024evolution}, and how firms should optimize prices and promotions \citep{vilcassim1995investigating}.

The success of demand models in predicting counterfactual quantities of interest (\emph{counterfactuals} hereafter) hinges on their ability to capture substitution patterns. Doing so requires knowledge of the product attributes that correctly reflect the underlying dimensions of differentiation, which poses a fundamental measurement challenge \citep{berry2021foundations}. First, consumer choices are often driven by hard-to-quantify characteristics, such as visual design, user friendliness, or style. In these cases, a growing literature shows that product images, descriptions, and reviews contain valuable information to capture substitution 
\citep{compiani2025demand, han2025copyright, lee2025generative, caouiEstimatingConsumerPreferences2025}. Consumer surveys and product recommendations on platforms  may also provide useful measures of product differentiation \citep{Magnolfi_et_al_2022, mcclure2025_recommendations}. To use these high-dimensional, unstructured data in demand models, practitioners transform them into lower-dimensional numerical variables, or \emph{embeddings}, using  machine learning (ML) methods. Second, even numeric attributes could be imperfect proxies \citep[e.g.,][]{nevo2001measuring, allcott2014gasoline}, or could be high-dimensional and collinear, requiring dimension-reduction \citep[e.g.,][]{Backus_et_al_2021}. In all of these cases, the variables used as inputs in the demand model are \emph{proxies} for the true attributes that drive consumer choices. It is essential that these proxies adequately capture the true dimensions of differentiation: poor proxies can lead to biased estimates of demand model parameters and biased counterfactuals. 

In this paper, we propose a simple, post-estimation \emph{bias correction} for counterfactuals. We take the \emph{naive estimator} that treats the proxies as if they were the true dimensions of differentiation (as is implicitly done in practice) and add a correction term designed to achieve two goals. First, it removes the asymptotic bias arising from using imperfect proxies. Second, it makes the bias-corrected estimator \emph{efficient}, meaning it has the smallest possible asymptotic variance within a broad class of estimators, when proxies may be imperfect. This is particularly important, as it is well known that estimating substitution can be data demanding. The intuition for the bias correction is simple: when proxies are imperfect, the model fit will tend to suffer in a way that correlates with the bias in the counterfactual estimator. Our approach leverages this correlation, mapping the discrepancy between the model and the data into a correction for the counterfactual.
We also provide closed-form standard errors, making it easy to perform valid inference on counterfactuals. 
 In addition, we show how two simple \emph{diagnostics} can be used to assess the adequacy of different proxies in capturing substitution. These diagnostics help guide the choice of \emph{how many} and  \emph{which} proxies/attributes should be included---questions that practitioners need to answer in any instance. 

We develop these methods for two widely-used empirical frameworks. The first follows \cite{BLP_1995,BLP_micro_2004}: prices vary across markets, instrumental variables are used to address price endogeneity, and market-level data may be supplemented with individual choice data. Many papers in industrial organization (IO) and fields using IO tools fit in this category. The second consists of models estimated on individual choice data with product-level fixed effects, which are common in marketing applications (see \cite{dube2019handbook} for a review). While we illustrate our approach for workhorse specifications (e.g., mixed logit with normal random coefficients), the methods apply to a broad class of parametric functional forms. 

The bias corrections and diagnostics are computationally light and integrate easily into the standard demand estimation workflow. The bias corrections take the naive parameter estimates produced by standard packages (e.g., PyBLP or xlogit) as inputs, and require neither bootstrapping nor additional optimization. Similarly, the diagnostics are Lagrange Multiplier (LM) statistics evaluated at the naive estimates. 

Our methods are based on two key insights. The first is that the demand model is misspecified when proxies are imperfect. Importantly, \emph{this is not a measurement error problem:} the proxies are at the product level whereas the unit of observation is the market and/or individual. By contrast, mismeasurement is at the observation level in a standard measurement error problem. The second insight is to reparameterize the model with a composite parameter that captures how proxies interact with structural parameters to affect utilities. The benefit of this framing is that it allows us to then address the misspecification problem using standard two-step methods. We develop formal theory to justify this approach. Importantly, the usual workflow that treats the proxies as the true attributes rests on assumptions that are stronger than the key conditions needed for our theory. Thus, our approach gains robustness to proxy error at little extra cost. Further, our diagnostics can be used to validate the key assumptions of our theory in empirical settings.

A further advantage of this approach is that it allows us to be agnostic about the nature of the proxy error.
This is particularly valuable because proxies are often obtained via black-box ML models, making it difficult to justify specific assumptions on the error.
This approach also accommodates data-dependent proxies, such as those obtained by \emph{fine-tuning}. For instance, practitioners may update the embeddings produced by off-the-shelf LLMs or neural networks to better fit the choice data. Moreover, this approach does not require taking a stand on the units of the proxies and/or true dimensions of differentiation. This is especially important for hard-to-quantify characteristics like visual design or user friendliness that lack natural units.

Two simulations calibrated to a standard dataset \citep{nevo2001measuring} and to the empirical application below  confirm that the bias correction improves performance for a range of levels of proxy error. Specifically, when bias in the naive estimator is material, the corrected estimator has lower bias and lower RMSE. Our formulas for standard errors also deliver confidence intervals that exhibit good coverage of the true counterfactual even when the naive estimator is substantially biased.

Finally, using the experimental data from \cite{compiani2025demand} (\citetalias{compiani2025demand} hereafter), we show that the bias correction meaningfully improves the model's ability to predict counterfactual choices following product removals. To this end, we leverage the fact that the data features both consumers' first and second choices. We estimate model parameters on the first choice data, then compare how well the naive and bias-corrected estimators predict second choices. This provides a ground truth to validate the performance of our approach in predicting counterfactuals. The bias correction improves the model's ability to predict a product's closest substitute from 40\% (with the naive estimator) to 70\%. Further, our diagnostics correctly identify the set of proxies that perform best at the counterfactual prediction task, indicating that they can be valuable tools for practitioners.

The mechanism behind the improvement is informative. The predictions are based on a mixed logit model with random coefficients on the embeddings. As documented in our simulations, when embeddings are poor proxies they are effectively discarded and the naive predictions revert to those of a plain logit model. This is what we see in the application: the naive predictions are biased in the direction of a plain logit model, whereas the bias correction recovers the true second-choice behavior.

Our approach is also helpful for practitioners using standard numeric attributes, as proxy error may be present even in this case, especially when dimension-reduction methods (e.g., PCA) are used to shrink the attribute set. Further, the choice of which numeric attributes to include is generally ad hoc, and our diagnostics can help practitioners in making these decisions. Whether or not proxy error is a concern, an additional contribution of this paper is to provide easy-to-compute, efficient estimators and standard errors for counterfactuals across many empirical settings, including combined market-level and microdata \citep[e.g.,][and many subsequent works]{petrin2002quantifying, BLP_micro_2004}. To the best of our knowledge, these contributions are new.\footnote{\cite{grieco2025optimal} study efficient estimation of model parameters in mixed logit models with combined market-level and microdata. Our focus is instead on efficient estimation of counterfactuals. For counterfactuals that depend on data moments in addition to model parameters (e.g., average diversions and average welfare),  efficient estimators of model parameters do not necessarily lead to efficient estimators of counterfactuals \citep{brown1998efficient,ai2012semiparametric}.}

Our approach is related to double/debiased ML (DML).\footnote{A recent application of DML in single-product demand estimation is \cite{bach2024adventures}. Unlike our approach, they assume their embeddings perfectly capture the true product attributes and use DML to correct for the first-stage estimation of nuisance functions, not to correct proxy error.} Both approaches aim to estimate a target parameter in the presence of nuisance parameters. In our setting, the target is the counterfactual and the nuisances are both the latent dimensions of differentiation, which are ``estimated'' using proxies, and the demand model parameters, which are themselves estimated using the proxies as inputs. This multi-step dependence on the proxies combined with the fact that the proxies are often outputs of black-box ML models trained on data to which one might have no access makes our setting quite different. Nevertheless, both approaches rely on standard orthogonalization ideas from the semiparametrics literature.

A recent literature recognizes that naively treating ML-generated variables as data leads to measurement-error bias and develops corrections for it. But as noted above, the problem we study is one of model misspecification rather than measurement error. Moreover, almost all strategies in this literature rely on validation data linking ML-generated variables and their ground-truth values.\footnote{See, e.g., \cite{fong2021machine,allon2023machine,angelopoulos2023prediction,egami2023using,zhang2023debiasing,carlson2025unifying} and references therein. These works build on an earlier literature on auxiliary data \citep{chen2005measurement,chen2008semiparametric}. One exception is \cite{battaglia2024inference} who develop analytical bias corrections for linear regression without validation data.} In our setting, however, the true dimensions of differentiation are latent, rendering these methods inapplicable.

The remainder of the paper is structured as follows. Section~\ref{sec:blp} presents our bias corrections and diagnostics for BLP-type models, while Section~\ref{sec: exog prices} does the same for models with individual-level choice data and product fixed effects. Both sections conclude with a practitioner's guide detailing the steps involved and giving practical recommendations. Simulations and the empirical application are presented in Sections~\ref{sec:simulations} and~\ref{sec:application}, respectively. Section~\ref{sec:theory} presents all theoretical results with all proofs deferred to Appendix~\ref{app:proofs}. Section~\ref{sec:conclusion} concludes.

\section{Case 1: Endogenous Prices}
\label{sec:blp}

We first consider a setting where prices vary at the market level and identification is achieved through instruments.

\subsection{Model and Data}

Following an established literature \citep{berry2014identification, freyberger2015asymptotic}, we assume that one has data on a large number $T$ of markets in which (subsets of) $J$ goods are sold.\footnote{For simplicity, we assume that $J$ is fixed. It is straightforward to extend our approach to asymptotic thought experiments where $J$ grows slowly with the number of markets $T$.} 
In addition to the outside option (denoted by 0), each market $t$ features products $\mathcal{J}_t \subseteq \left\{1, \ldots, J \right\} $, for which prices $ p_t =  (p_{jt})_{j \in \mathcal{J}_t}$, exogenous product attributes $ x_t =   (x_{jt})_{j \in \mathcal{J}_t}$, and market shares $s_t = (s_{jt})_{j \in \mathcal{J}_t}$ are observed. 
Consumer choices are also driven by unobservables $\xi_t = (\xi_{jt})_{j \in \mathcal{J}_t}$. The model predicts market shares as a function of $p_t$, $\xi_t$, $x_t$, and a parameter vector $\theta$:
\begin{equation}
	\label{eq: choice probab blp}
	s_{jt} = \sigma_j(p_t, \xi_t, x_t; \theta), \hspace{10pt} j \in \mathcal{J}_t.
\end{equation}
Prices are endogenous and  may be correlated with the unobservables. To address this, we rely on a vector of instrumental variables $z_t = (z_{jt})_{j \in \mathcal J_t}$ that satisfy\footnote{For ease of exposition, we focus only on demand-side moments but our approach extends naturally to supply-side moments as well.}
\begin{equation}
	\label{eq:instrument}
	\E[\xi_{jt} | z_{jt}] = 0, \hspace{10pt} j\in \mathcal{J}_t,  \hspace{10pt} t=1, \ldots, T.
\end{equation}
We also accommodate ``micro BLP'' settings \citep{BLP_micro_2004,berry2024nonparametric} where, in addition, individual-level data are available in a subset of markets. The microdata consists of choice indicators $d_{it} = (d_{ijt})_{j \in \mathcal J_t}$ taking the value $1$ if $i$ chose $j$ in market $t$ and $0$ otherwise, and demographics $y_{it}$ that vary at the consumer level, such as income, and/or $\bar y_{ijt}$ that vary at the product-consumer level, such as distance between a household's home and a school or a hospital. We assume the microdata are available for a fixed set of markets, labeled $t=1, \ldots, \tau$ without loss.\footnote{The case with microdata for all $T$ markets is simpler. Derivations are available upon request.} We treat the microdata as repeated cross sections of size $N_1,\ldots,N_\tau$.

This model subsumes many empirical specifications used in the literature. We provide two simple examples below to fix ideas.

\begin{example}[BLP]
 	\label{ex:blp}
 	The utility that individual $i$ derives from good $j$ in market $t$ is
 	\begin{equation}
 		\label{eq: utility BLP}
 	\begin{aligned}
 		u_{ijt} & =  \beta'_{i} x_{jt} -\alpha_i p_{jt}  + \xi_{jt} + \varepsilon_{ijt} , & & &  j \in \mathcal{J}_t , \\
 		u_{i0t} & = \varepsilon_{i0t} ,
 	\end{aligned}
 	\end{equation}
 	where the $\varepsilon_{ijt}$ are iid type 1 extreme value random variables. Market shares are 
 	\begin{equation}
 	\sigma_j(p_t, \xi_t, x_t; \theta)
	=  \int \frac{e^{ \beta' x_{jt} -\alpha p_{jt}  + \xi_{jt} }}{1+\sum_{k \in \mathcal J_t} e^{ \beta' x_{kt} -\alpha p_{kt}  + \xi_{kt}}} dF(\alpha, \beta; \theta), \hspace{10pt} j\in \mathcal{J}_t, \label{eq: sigma blp parametric}
 	\end{equation}
	for some parametric distribution $F$.
 \end{example}
 
 \begin{example}[Micro BLP]
 \label{ex:micro blp}
 	The utility is specified as:
 	 	\begin{equation}
 		\label{eq: utility micro BLP}
 		\begin{aligned}
 			u_{ijt} & =  \beta'_{i} x_{jt} -\alpha_i p_{jt} + y_{it}' \Pi (x_{jt},  p_{jt}) + \pi ' \bar y_{ijt} + \xi_{jt} + \varepsilon_{ijt} , & & & j \in \mathcal{J}_t, \\
 			u_{i0t} & = \varepsilon_{i0t} ,
 		\end{aligned}
 	\end{equation}
 Micro moments can be computed using the following expression for choice probabilities:
 	\begin{multline}
	\label{eq: choice probabilities micro BLP}
	Pr(d_{ijt}=1|y_{it},\bar y_{it}, p_t, \xi_t, x_t; \theta) \\
	= \int \frac{e^{ \beta' x_{jt} -\alpha p_{jt} + y_{it}' \Pi (x_{jt},  p_{jt}) +  \pi' \bar y_{ijt} + \xi_{jt} }}{1+\sum_{k \in \mathcal J_t} e^{ \beta' x_{kt} -\alpha p_{kt} + y_{it}' \Pi (x_{kt},  p_{kt}) +  \pi' \bar y_{ikt}  + \xi_{kt}}} dF(\alpha, \beta; \theta).
 	\end{multline}
     Market shares are obtained by integrating \eqref{eq: choice probabilities micro BLP} over the (known) distribution of $(y_{it},\bar y_{it})$ for $\bar y_{it} = (\bar y_{ijt})_{j \in \mathcal J_t}$.
 \end{example}

 \bigskip
 
So far the model is standard. Our point of departure is to partition 
\[
 x_{jt} \equiv (\bar x_{jt},e_j),
\]
where $\bar x_{jt}$ is a vector of conventional observed product attributes, such as product size, and $e_j$ is an $r$-vector of product attributes, such as visual design or user friendliness, that are difficult to capture using standard numeric data. Accordingly, we treat $e = (e_j')_{j=1}^J$ as known to the consumer but not observed in the data. Instead, the practitioner observes proxies $\tilde e = (\tilde e_j')_{j=1}^J$ for the true underlying $e$. Note that $e$ does not vary across markets, consistent with the fact that difficult-to-quantify characteristics such as visual design or user friendliness are often fixed product characteristics.

Where might $\tilde e$ come from? We consider two leading cases:

\begin{enumerate}
    \item \textbf{Unstructured data.}
    A growing literature computes low-dimensional representations (or \emph{embeddings}) $\tilde e_j$ of unstructured data via ML methods. To maximize generality, we stay agnostic on the ML method and the type of unstructured data. It could be text (e.g., product descriptions and reviews), images, audio/video, a combination thereof \citep{compiani2025demand,han2025copyright}, or data from consumer surveys \citep{Magnolfi_et_al_2022}. 
    \item \textbf{Standard quantitative attributes.}
    Our approach may equally be used to correct bias when included product attributes that do not vary across markets fail to correctly capture the dimensions of differentiation (e.g., the ``mushiness'' of cereal hand-coded by \cite{nevo2001measuring}). In these scenarios, the $\tilde e_j$ are proxies for the true latent attributes $e_j$, and $\bar x_{jt}$ represents the remaining attributes that are assumed to be perfectly observed.
\end{enumerate}

Unlike prior work which implicitly treats the $\tilde e_j$ as the ground truth, we account for the fact that proxies $\tilde e_j$ are only approximations to the true latent $e_j$. When the $\tilde e_j$ are poor proxies for $e_j$, standard workflows can deliver biased estimates of counterfactuals and invalid inference. Our first main goal is to develop estimators that are immune to this bias. Another goal is to shed light on what a ``good'' proxy might look like, which can guide the choice of unstructured data source and ML model. This objective is fundamentally different from the standard problem of choosing proxies for a prediction problem, since in our case the counterfactual is not observed in the data.

\subsection{Bias Correction for Counterfactuals}
\label{sec: counterfactual case 1}

We consider a broad class of counterfactuals that can be written as
\begin{equation}
	\label{eq: cf blp}
\kappa = \E[ k(p_t, \xi_t, \bar x_t, e; \theta)], 
\end{equation}
where the expectation is over the distribution of $(p_t, \xi_t, \bar x_t)$ across markets $t$.  For instance, $\kappa$ might represent an average price elasticity, average equilibrium price or consumer welfare measure (possibly after a counterfactual change on the supply side). Counterfactuals for a specific market, such as the price elasticity at a given $(p, \xi, \bar x)$, are also subsumed. In this case, $k$  is a deterministic function of $\theta$ and the expectation becomes redundant. As we discuss below, $\kappa$ could also represent certain elements of $\theta$, such as the average price coefficient.

In the standard workflow, given the observed aggregate data and a candidate set of proxies $\tilde e$, the model parameters $\theta$ are first estimated by GMM based on the moment\footnote{With slight abuse of notation, we now write $\sigma_j$ as functions of $\bar x_t$ and $e$, with the understanding that $\sigma_j$ depends on $e$ only through $(e_j)_{j \in \mathcal J_t}$, and similarly for $\hat \xi_t$.}
\[
\frac 1T \sum_{t=1}^T Z_t \hat \xi_{t}(s_t,p_t, \bar x_t, \tilde e; \theta),
\]
where $\hat \xi_{t}(s_t,p_t, \bar x_t, e; \theta) = (\hat \xi_{j}(s_t,p_t, \bar x_t, e; \theta))_{j \in \mathcal J_t}$ is defined implicitly via
\begin{equation} \label{eq:xi}
	s_{jt} = \sigma_j( p_t, \hat \xi_t, \bar x_t, e; \theta), \quad j \in \mathcal{J}_t,
\end{equation}
and $Z_t$ is a $\dim(z) \times |\mathcal J_t|$ matrix whose columns contain the original instruments $z_{jt}$ in~(\ref{eq:instrument}) or transformations thereof, such as optimal instruments (see Remark~\ref{rmk:optimal} below). 
With microdata, the GMM criterion is combined with a minimum-distance criterion based on the micro moments \citep[as in, e.g.,][]{conlon2025incorporating}.

Given an estimate $\hat \theta$ of $\theta$, the counterfactual $\kappa$ is usually estimated as
\begin{equation}
	\label{eq: kappa naive blp}
	\hat \kappa = \frac 1T \sum_{t=1}^T k(p_t, \hat \xi_t(s_t, p_t, \bar x_t, \tilde e; \hat \theta), \bar x_t, \tilde e; \hat \theta). 
\end{equation}
We call this the \emph{naive estimator} of $\kappa$ since it does not account for the fact that the proxies $\tilde e$ might differ from the true latent attributes $e$. This \emph{proxy error} has the potential to affect the estimator via two channels: (i) directly, since $\tilde e$ is an argument of $k$, and (ii) indirectly through both $\hat \theta$ and $\hat \xi_t$. 

We now introduce our bias correction procedure. We restrict attention to models in which the attributes $e$ and model parameters $\theta$ enter choice probabilities~(\ref{eq: sigma blp parametric}) via a \emph{composite parameter} 
\[ \gamma \equiv \gamma(\theta, e).
\]
Many models feature this property. We illustrate it in the BLP example.

\medskip

\begin{continuance}{ex:blp}
	We partition $\beta_i = (\beta_{\bar x, i}, \beta_{e,i})$ and write the utilities  as:
	\[
	u_{ijt}  =  \beta'_{\bar x, i} \bar x_{jt}  + \beta'_{e, i} e_{j} -\alpha_i p_{jt} + \xi_{jt} + \varepsilon_{ijt} ,\hspace{10pt} j\in \mathcal{J}_t.
	\]
	Suppose $\alpha_i \sim N(\bar \alpha, \sigma^2_\alpha)$, $\beta_{\bar x,i} \sim N(\bar \beta_{\bar x}, \Sigma_{\bar x})$, $\beta_{e,i} \sim N(\bar \beta_e, \Sigma_{e})$, and $\alpha_i$, $\beta_{\bar x,i}$ and $\beta_{e,i}$ are independent. Then, 
	\[
	\theta = \left(\bar \alpha, \sigma_\alpha, \bar \beta_{\bar x}, \bar \beta_e, l(\Sigma_{\bar x}), l(\Sigma_e) \right),
\] where $l(\Sigma)$ stacks the lower-triangular entries of the Cholesky factor of $\Sigma$ into a vector.\footnote{If $\Sigma_{\bar x}$ and/or $\Sigma_e$ are diagonal, then let $l(\Sigma_{\bar x})$ and/or $l(\Sigma_e)$ denote their diagonal entries.} Note that  $e_j$ only enters via $\beta'_{e,i} e_j$. 
Collecting $\beta'_{e,i} e_j$ across products, we have $e \beta_{e, i} \sim N(e \bar \beta_e, e \Sigma_e e')$, where $e \Sigma_e e'$ has rank $r \leq J$ because $e$ is $J \times r$. Hence,
	\[
	\gamma(\theta,e) = \left(\bar \alpha, \sigma_\alpha, \bar \beta_{\bar x}, e \bar \beta_e, l(\Sigma_{\bar x}), l_r(e \Sigma_{e}e')\right),
	\]
	where $l_r$ stacks the lower-triangular entries of the rank-$r$ Cholesky factor of $e \Sigma_e e'$.\footnote{As $e \Sigma_e e'$ is $J \times J$ with rank $r$, $l_r(e \Sigma_{e}e')$ is the unique $J \times r$ matrix $L$ whose above-diagonal entries are all zeros and whose diagonal entries are all positive, such that $LL' = e \Sigma_e e'$.}
    For instance, when both $\bar x_j$ and $e_j$ are scalars and $J=2$, we have
    \[\gamma(\theta, e) = \left(\bar \alpha, \sigma_\alpha, \bar \beta_{\bar x}, e_1 \bar \beta_e, e_2 \bar \beta_e, \sigma_{\bar x}, \sigma_e |e_1|, \sigma_e e_2 \mathrm{sign}(e_1) \right),\]
    where $\sigma_{\bar x}$ and $\sigma_e$ are the standard deviations of the random coefficients on $\bar x_j$ and $e_j$.
\end{continuance}

\bigskip

As can be seen from this example, parameters that do not interact with $e$, e.g., the average price coefficient $\bar \alpha$, are left unchanged in $\gamma$. The remaining components of $\gamma$  capture how $e$ and $\theta$ interact to jointly affect utilities. A similar reparameterization for Example~\ref{ex:micro blp} is provided in Appendix~\ref{app:micro BLP}.

This reparameterization allows us to simplify notation as follows. First, we note that the right-hand side of (\ref{eq:xi}) depends on $(\theta, e)$ only via $\gamma(\theta,e)$. Thus, we write $\hat \xi_{jt}(\gamma(\theta, e)) = \hat \xi_{j}(s_t,p_t, \bar x_t, e; \theta)$ for $j \in \mathcal J_t$ (suppressing dependence on $s_t,p_t, \bar x_t$) and let $\hat \xi_t(\gamma(\theta, e)) = (\hat \xi_{jt}(\gamma(\theta, e)))_{j \in \mathcal J_t}$. 
We similarly restrict attention to counterfactuals that depend on $(\theta, e)$ via $\gamma(\theta, e)$ and write $k_t(\gamma(\theta, e)) = 	k(p_t, \hat \xi_t(\gamma(\theta, e)), \bar x_t, e; \theta)$. This includes many counterfactuals, such as elasticities with respect to prices or $\bar x$, equilibrium prices, and welfare changes associated with changes in prices or $\bar x$. It precludes quantifying the effect of changes in $e$ or measuring heterogeneity in preferences for  $e$, but these have little meaning when $e$ has no natural scale or interpretation.

Without microdata, the \emph{bias-corrected estimator} is
\begin{equation}
\label{eq: kappa_bc blp linear no micro}
 \hat \kappa_{bc} = \frac 1T \sum_{t=1}^T \big( k_t(\hat \gamma) - \hat c' Z_t \hat \xi_t(\hat \gamma) \big) ,
\end{equation}
where $\hat \gamma = \gamma(\hat \theta, \tilde e)$ denotes the value of $\gamma$ at the estimated structural parameters $\hat \theta$ and the candidate proxies $\tilde e$, and $\hat c$ is a $\dim(z) \times 1$ vector of weights. We give a closed-form expression for $\hat c$ below. This bias correction is easy to implement: it simply takes the naive estimator $\frac 1T \sum_{t=1}^T k_t(\hat \gamma)$ and adds a weighted average of the estimation moments. Thus, it requires minimal computation beyond what is needed to estimate model parameters $\theta$ in the first place.

With microdata, the bias correction also depends on micro moments. Let $\bar m_t = \frac{1}{N_t} \sum_{i=1}^{N_t}  m_{it}$, where $m_{it} = m(y_{it},\bar y_{it}, d_{it})$ is a known function of demographics and choice data for individual $i$ in market $t$. Let $m(p_t, \xi_t, \bar x_t, e; \theta)$ denote the model-implied expectation of $m_{it}$ conditional on the market-level data:
\[
 m(p_t, \xi_t, \bar x_t, e; \theta) = \E[m_{it} | p_t, \xi_t, \bar x_t, e].
\]
As the choice probabilities~(\ref{eq: choice probabilities micro BLP}) depend on $(\theta, e)$ only via $\gamma(\theta, e)$, we write $m_t(\gamma(\theta, e)) = m(p_t, \hat \xi_t(\gamma(\theta, e)), \bar x_t, e; \theta)$. With this notation, the bias-corrected estimator is
\begin{equation}
\label{eq: kappa_bc blp linear}
\hat \kappa_{bc} = \frac 1T \sum_{t=1}^T \big( k_t(\hat \gamma) - \hat c' Z_t \hat \xi_t(\hat \gamma) \big) + \sum_{t = 1}^\tau \hat d_t' \big( \bar m_t - m_t(\hat \gamma) \big),
\end{equation}
where $\hat c$ and $\hat \gamma$ are as above, and $\hat d_1,\ldots,\hat d_{\tau}$ are $\dim(m) \times 1$ vectors of weights for the micro moments, given below. 

The idea behind (\ref{eq: kappa_bc blp linear no micro}) and (\ref{eq: kappa_bc blp linear}) is to choose the weights so that $\hat \kappa_{bc}$ does not depend on $\hat \gamma$ to first order.\footnote{There is a long tradition of using corrections such as these in two-step estimation. See, e.g., \cite{andrews1994asymptotics} and \cite{newey1994asymptotic}. Of course, similar debiasing ideas underlie the DML literature.} This means that, to first order, $\hat \kappa_{bc}$ behaves like the right-hand side of \eqref{eq: kappa_bc blp linear no micro} or \eqref{eq: kappa_bc blp linear} with $\hat \gamma$ replaced by the true value $\gamma_0 = \gamma(\theta_0, e_0)$, where $\theta_0$ are the true structural parameters and $e_0$ are the true latent attributes. Section~\ref{sec:theory blp} formally shows that  $\hat \kappa_{bc}$ is asymptotically unbiased and its distribution does not depend on $\hat \gamma$. This has two important implications. First, $\hat \kappa_{bc}$ is immune to bias due to proxying $e$ with $\tilde e$. Second, standard errors do not need to be corrected when $\tilde e$ is chosen in a data-dependent way, e.g., by fine-tuning an ML model on choice data. 

For the intuition, consider the case without microdata. A Taylor expansion of the naive estimator yields 
\[
 \hat \kappa \approx \frac{1}{T} \sum_{t=1}^T \left( k_t(\gamma_0) + \frac{\partial k_t(\hat \gamma)}{\partial \gamma'}(\hat \gamma - \gamma_0)\right).
\]
We wish to eliminate the second problematic term depending on $\hat \gamma$. We replicate this term using the estimation moments, by choosing $\hat c$ so that
\[
 \frac{1}{T} \sum_{t=1}^T \frac{\partial k_t(\hat \gamma)}{\partial \gamma'} = \hat c' \left( \frac{1}{T} \sum_{t=1}^T  Z_t \frac{\partial \hat \xi_t(\hat \gamma)}{\partial \gamma'} \right).
\]
Substituting in the previous display and ``undoing'' the Taylor expansion, we get
\[
 \begin{aligned}
 \hat \kappa & \approx \frac{1}{T} \sum_{t=1}^T \left( k_t(\gamma_0) + \hat c' Z_t \frac{\partial \hat \xi_t(\hat \gamma)}{\partial \gamma'}(\hat \gamma - \gamma_0) \right) \\
 & \approx \frac{1}{T} \sum_{t=1}^T \left( k_t(\gamma_0) + \hat c' Z_t \hat \xi_t(\hat \gamma) - \hat c' Z_t \hat \xi_t(\gamma_0) \right).
 \end{aligned}
\]
This suggests that to correct bias we want to adjust the naive estimator by subtracting $\frac{1}{T}\sum_{t=1}^T ( \hat c' Z_t \hat \xi_t(\hat \gamma) - \hat c' Z_t \hat \xi_t(\gamma_0) )$. The final (infeasible) term depending on $\gamma_0$ has mean zero by virtue of~\eqref{eq:instrument}, so we drop it, leading to the corrected estimator~\eqref{eq: kappa_bc blp linear no micro}. This correction therefore ensures that, to first order, $\hat \kappa_{bc}$ depends only on $\gamma_0$.

What assumptions are needed for this result? Besides standard regularity conditions, we require that the discrepancy between $\hat \gamma$ and $\gamma_0$ not be too large relative to sampling error (see Section~\ref{sec:theory blp}), which is a standard condition in two-step estimation (e.g., \cite{newey1994asymptotic}). A few points on this condition. First, given the nonlinear nature of multi-product demand models, a condition on the discrepancy between $\hat \gamma$ and $\gamma_0$ seems unavoidable. Second, the usual demand estimation workflow implicitly assumes that $\tilde e$ is the true $e$. This, in turn, already implies that the discrepancy between $\hat \gamma$ and $\gamma_0$ is much smaller than required by our theory (see Remark~\ref{rmk:vicinity}). Thus, our approach gains robustness to proxy error and its theoretical guarantees hold under weaker conditions than are already assumed by the usual workflow. Third, we provide diagnostics to validate this condition in empirical settings. 

While this condition is made to justify our asymptotic theory, asymptotics are only useful insofar as they deliver accurate approximations to the finite-sample distribution of $\hat \kappa_{bc}$ encountered in practice. In any finite sample, this condition requires that $\hat \gamma$ be within a vicinity of $\gamma_0$ that is roughly double the order of sampling uncertainty. This seems plausible, because $\hat \gamma = \gamma(\hat \theta, \tilde e)$, where $\hat \theta$ (and possibly $\tilde e$, when fine-tuning) is estimated on the choice data. In simulations calibrated to the \cite{nevo2001measuring} data, Figure~\ref{fig:bias rmse sim case 1} shows that the bias correction is effective even when the proxy error is sufficiently large that the naive estimator is badly biased.

We also require that the $e_j$ and $\tilde e_j$ have the same dimension. This condition is again implicitly assumed in the usual workflow and can be validated with our diagnostics.

To introduce the expressions for the weights $\hat c$ and $\hat d_1,\ldots,\hat d_\tau$, let 
\[
\begin{aligned}
	\hat V & = \frac 1T \sum_{t=1}^T (Z_t \hat \xi_t(\hat \gamma))(Z_t \hat \xi_t(\hat \gamma))' - \bar g \bar g' , \\
	\hat V_t & = \frac{T}{N_t} \left( \frac{1}{N_t} \sum_{i=1}^{N_t} m_{it} m_{it}' - \bar m_t \bar m_t' \right), \quad t = 1,\ldots,\tau,
\end{aligned}
\]
denote the sample variance of the estimation moments, where $\bar g = \frac 1T \sum_{t=1}^T Z_t \hat \xi_t(\hat \gamma)$. Define 
\begin{equation}\label{eq:h blp}
\begin{aligned}
	\hat h & = \hat k - \hat G' \hat V^{-1} \hat K,  & & & 
	\hat H & =  \hat G' \hat V^{-1}  \hat G  + \sum_{t=1}^\tau \hat M_t' \hat V_t^{-1} \hat M_t , 
\end{aligned}
\end{equation}
where $\hat k = \frac 1T \sum_{t=1}^T \dot k_t(\hat \gamma)$ is $\dim(\gamma) \times 1$, $\hat K = \frac 1T \sum_{t=1}^T k_t(\hat \gamma) Z_t \hat \xi_t(\hat \gamma) - \hat \kappa \bar g$ is $\dim(z) \times 1$, $\hat G = \frac 1T \sum_{t=1}^T Z_t \dot \xi_t(\hat \gamma)$ is $\dim(z) \times \dim(\gamma)$, $\hat M_t = \dot m_t(\hat \gamma)$ is $\dim(m) \times \dim(\gamma)$, 
and
$\dot k_t(\gamma) = \frac{\partial k_t(\gamma)}{\partial \gamma}$, $\dot \xi_{t}(\gamma)' = \frac{\partial \hat \xi_{t}(\gamma)'}{\partial \gamma}$, and $\dot m_{t}(\gamma)' = \frac{\partial m_{t}(\gamma)'}{\partial \gamma}$ are $\dim(\gamma) \times 1$, $\dim(\gamma) \times J_t$, and $\dim(\gamma) \times \dim(m)$, respectively.
The weights to plug into \eqref{eq: kappa_bc blp linear} are
\begin{equation}
\label{eq: c^hat blp}
	\hat c  = \hat V^{-1} ( \hat K  +\hat G \hat H^{-1} \hat h ) , 
\end{equation}
and
\begin{equation}
\label{eq: d^hat blp}
	\hat d_t  = \hat V_t^{-1} \hat M_t \hat H^{-1} \hat h, \quad t = 1,\ldots,\tau.
\end{equation}
Without microdata, $\hat d_1 = \ldots = \hat d_\tau = 0$ and $\hat H = \hat G' \hat V^{-1}  \hat G$.

We defer formal statements of the results sketched out above to Section~\ref{sec:theory blp}, and instead highlight a few key properties of the corrected estimator.
	
\begin{remark}[Easy to compute standard errors]
\label{rmk:std errors}
The asymptotic variance of the bias-corrected estimator can be estimated as follows:
\begin{equation}
\label{eq:standard errors blp}
\hat V_{bc} = \hat s_k^2 + \hat c' \hat V \hat c - 2 \hat c' \hat K + \sum_{t=1}^\tau \hat d_t' \hat V_t \hat d_t,
\end{equation}
where $\hat s_k^2 = \frac 1T \sum_{t=1}^T k_t(\hat \gamma)^2 - \hat \kappa^2$ is the sample variance of $k_t(\hat \gamma)$. Without microdata, the expression simplifies to 
\begin{equation}
\label{eq:standard errors blp no micro data}
\hat V_{bc} = \hat s_k^2 + \hat c' \hat V \hat c - 2 \hat c' \hat K. 
\end{equation}
In either case, standard errors for $\hat \kappa_{bc}$ are $\sqrt{\hat V_{bc}/T}$.
Again, these are closed-form expressions involving quantities that are easy to compute given $\hat \theta$.
\end{remark}
	
\begin{remark}[Efficiency]
While there are many different choices of weights $\hat c$ and $\hat d_1,\ldots,\hat d_\tau$ that correct bias, the weights given above are chosen so that the asymptotic variance of $\hat \kappa_{bc}$ is as small as possible given the instruments. See Proposition~\ref{prop:efficiency} for a formal statement. Importantly, $\hat \kappa_{bc}$ remains efficient even if $\hat \theta$ is inefficient. 
\end{remark}
	
\begin{remark}[Fine-Tuning]
The bias correction and standard error formulas allow the proxies $\tilde e$ to be sample-dependent. This accommodates scenarios where an off-the-shelf algorithm has been fine-tuned to obtain embeddings that better fit the choice data. Thus, there is no need to correct the standard errors for fine-tuning. 
\end{remark}

\begin{remark}[Optimal Instruments]\label{rmk:optimal}
We recommend transforming the instruments $z_t$ in~(\ref{eq:instrument}) into ``optimal instruments'' for $\gamma$. An advantage of this approach is that the dimension of $Z_t$ is always compatible with that of $\gamma$. That is, no additional instruments are needed to implement the bias correction: one just takes different transformations of the original instruments $z_t$ in~(\ref{eq:instrument}). 

To construct optimal instruments for $\gamma$, we adapt the approach of \cite{conlon2020best} with a demand side only. We regress each $p_{jt}$ on the full set of instruments $z_t$ and let $\hat p_{jt}$ denote the fitted values. We then compute market shares $\hat s_{jt} = \sigma_j(\hat p_t, \xi_t, \bar x_t, \tilde e; \hat \theta)$ at the predicted prices $\hat p_t = (\hat p_{jt})_{j \in \mathcal J_t}$ and $\xi_t = 0$. Writing $\tilde \xi_{jt}(\gamma(\theta, e)) = \hat \xi_j(\hat s_t, \hat p_t, \bar x_t, e; \theta)$ with $\hat s_t = (\hat s_{jt})_{j \in \mathcal{J}_t}$, we then compute 
\[
 \tilde z_{jt} = \frac{\partial \tilde{\xi}_{jt}(\hat \gamma)}{\partial \gamma}, \quad j \in \mathcal J_t.
\]
Finally, we concatenate the $\tilde z_{jt}$ across products to form the $\dim(\gamma) \times |\mathcal{J}_t|$ matrix $Z_t$. 
\end{remark}

\begin{remark}[Product or Market Fixed Effects]\label{rmk:fes}
Our method accommodates additive product and/or market fixed effects  in utilities by de-meaning $Z_t$ and $\hat \xi_t$ at the product and/or market level before forming moments (see, e.g., \cite{somainiAlgorithmEstimateTwoWay2016} and \cite{conlon2020best}). All methods can be implemented as described in this section based on moments formed from the de-meaned variables. The theory developed in Section~\ref{sec:theory blp} carries over, as discussed further in Appendix~\ref{app:fes}.
\end{remark}

\subsection{Diagnostics}
\label{sec:blp diagnostics}

Next, we propose two diagnostics that practitioners can use to assess the suitability of a candidate set of proxies $\tilde e$. The first speaks to whether the $\hat \gamma \equiv \gamma(\hat \theta, \tilde e)$ induced by $\tilde e$ is sufficiently close to the truth $\gamma_0 \equiv \gamma(\theta_0, e_0)$. The second addresses the question of whether the dimension of $\tilde e$ matches that of $e_0$. Both diagnostics are based on LM statistics evaluated at $\hat \gamma$, so they require minimal additional computation. We view these diagnostics as useful tools to make the standard workflow more rigorous and transparent. Developing an inferential theory that fully integrates the diagnostics as pre-tests might be an interesting topic for future research.

\subsubsection{Diagnostic 1: Is $\hat \gamma$ Close To $\gamma_0$?}
\label{sec:goodness of fit blp}

The bias correction is based on linearization, and its theory developed in Proposition~\ref{prop:kappa.endogenous} below requires the unavoidable assumption that the discrepancy between $\hat \gamma$ and $\gamma_0$ not be too large relative to sampling error, as discussed above.
Here we show that a simple LM statistic can be used to validate this condition. This diagnostic also helps guide the choice of proxies (e.g., embeddings from various data sources or ML methods).

The diagnostic is
\begin{equation}
\label{eq: LM1 blp}
LM_1 = \|\sqrt T \hat H^{-1/2} \hat S \|^2,
\end{equation}
where $\hat S = \hat G' \hat V^{-1} ( \frac 1T \sum_{t=1}^T Z_t \hat \xi_t(\hat \gamma)) + \sum_{t=1}^\tau \hat M_t' \hat V_t^{-1}  (m_t(\hat \gamma) - \bar m_t)$ represents the ``score'' at $\hat \gamma$ and $\hat H$ is given in (\ref{eq:h blp}). This diagnostic can be interpreted as an LM statistic for testing the null hypothesis that $\gamma_0$ is in the set of composite parameters spanned by the candidate proxies $\tilde e$, ignoring the fact that $\tilde e$ is possibly stochastic.\footnote{Formally, a test of $\mathbb{H}_0: \gamma_0 \in \Gamma(\tilde e) := \{\gamma(\theta,\tilde e) : \theta \in \Theta\}$ against the alternative $\mathbb{H}_1:  \gamma_0 \in \Gamma \setminus \Gamma(\tilde e)$, where $\Theta$ and $\Gamma$ are the parameter spaces for $\theta$ and $\gamma$, respectively.}

Proposition~\ref{prop:lm.fit.gamma.2} below shows that $LM_1$ behaves asymptotically like $\|\sqrt T (\hat \gamma - \gamma_0)\|^2$. This allows one to validate the assumption that the discrepancy between $\hat \gamma$ and $\gamma_0$ is sufficiently small, as needed in Proposition~\ref{prop:kappa.endogenous}, by checking whether $LM_1$ is below a 
threshold. See the practitioner's guide at the end of this section for a concrete recommendation, and Section~\ref{sec:theory blp} for a formal discussion.
It can also be shown under a slight strengthening of the conditions of Proposition~\ref{prop:lm.fit.gamma.2} that $LM_1$ can be used to bound the distance between $\tilde e$ and the true latent attributes driving consumer choices (see Proposition~\ref{prop:lm.fit.e} for case 2). This is especially useful as the true attributes can never be observed. As a result, $LM_1$ also serves as a criterion to target when choosing among proxies and/or fine-tuning.

\subsubsection{Diagnostic 2: Is The Dimension of $\tilde e$ Correct?}
\label{sec:rank blp}

Practitioners also have to choose the dimension of $\tilde e$, i.e., how many proxies to include per product. This is particularly delicate when the proxies lack a natural economic interpretation, e.g., when they are obtained from black-box algorithms or by applying PCA to many numeric attributes. For example, in the application of Section~\ref{sec:application}, we use PCA to reduce the dimension of proxies obtained from pre-trained algorithms and need to choose how many principal components to include.

We guide this choice by providing a second diagnostic, denoted $LM_2$, again based on an LM statistic. The idea is to augment $\tilde e$ with a vector $\eta \in \mathbb R^J$ representing some excluded but potentially important product attributes and augment $\theta$ with an additional component $\psi \in \Psi$ representing coefficients on $\eta$. The diagnostic $LM_2$ is then an LM statistic for the null that $\psi = 0$. As with $LM_1$, it depends on $\hat \theta$ only and therefore requires minimal additional computation. Full details are deferred to Appendix~\ref{app: LM2 case 1}.

\subsection{Practitioner's Guide}

Given a counterfactual of interest $\kappa$ and a candidate set of proxies $\tilde e$:
\begin{enumerate}
    \item Compute the model parameter estimates $\hat \theta$ as usual, treating $\tilde e$ as the truth.
    \item Compute weights $\hat c$ in \eqref{eq: c^hat blp} and, if microdata is available, weights $\hat d_t$ in \eqref{eq: d^hat blp}.
    \item Plug the weights into \eqref{eq: kappa_bc blp linear} to compute the bias-corrected estimator $\hat \kappa_{bc}$. If only aggregate data is available, use \eqref{eq: kappa_bc blp linear no micro} instead.
    \item Compute standard errors for $\hat \kappa_{bc}$ using the formulas in Remark \ref{rmk:std errors}.
    \item To check whether a given set of proxies $\tilde e$ is adequate:
    \begin{enumerate}
        \item Compute the $LM_1$ statistic in \eqref{eq: LM1 blp}. If it is below a threshold $C_T^2$, conclude that the $\hat \gamma$ induced by $\tilde e$ is sufficiently close to $\gamma_0$. In Section~\ref{sec:goodness of fit blp.theory}, we motivate a threshold of $C_T^2 = \chi^2_{\dim(\gamma), 0.95} \log T$.
        \item Compute the $LM_2$ statistic in \eqref{eq: LM2 blp}. If it is below a threshold $\hat \xi^*_{0.95}$, conclude that $\tilde e$ is of adequate dimension. The bootstrap method in Appendix~\ref{app: LM2 case 1} can be used to compute $\hat \xi^*_{0.95}$.
    \end{enumerate}
\end{enumerate}

\subsection{What About Models with Standard Numeric Attributes?}

Our approach also provides a way to robustly and efficiently estimate counterfactuals and validate some of the assumptions implicitly made in the demand estimation literature in contexts where only standard numeric attributes are available.
The typical workflow implicitly assumes that the product attributes included in the model perfectly capture the dimensions of differentiation. Our bias correction allows practitioners to relax this assumption. To do so, the bias-corrected estimator can be implemented as above, where now $\tilde e$ simply represents the numeric attributes that may be imperfect proxies. 
Similarly, our diagnostics may be used to choose which and how many attributes to include. 

Even in the special case that proxy error is not a concern, our bias-corrected estimator and standard error formulas can be used as above with $\gamma \equiv \theta$ to easily perform efficient inference on counterfactuals. This approach offers a few advantages: (i) it yields efficient estimates of counterfactuals (given the instruments) even when $\hat \theta$ is inefficient, (ii) standard errors for the counterfactual are available in closed form; and (iii) it allows for combined market-level and microdata.

\section{Case 2: Individual-Level Price Variation}
\label{sec: exog prices}

Next, we consider settings with individual-level price variation and individual choice data. Following an established literature \citep[as reviewed in][]{dube2019handbook}, we include product fixed effects to account for systematic differences across products and rule out any remaining price endogeneity. 

\subsection{Model and Data}

We assume that one has data on a large number $n$ of consumers in a single market\footnote{For ease of exposition, we present results for a single market, though our approach extends easily to settings with multiple markets.} in which $J$ goods are sold,\footnote{As before, we assume that $J$ is fixed but it is straightforward to extend our analysis to asymptotic thought experiments where $J$ grows slowly with $n$.} and identify the outside option with $j = 0$. For each consumer $i$, we observe individual choices $d_i = (d_{ij})_{j=1}^J$ where $d_{ij} = 1$ if $i$ chooses good $j$ and $0$ otherwise, $p_i = (p_{ij})_{j=1}^J$ which collects prices and other variables (e.g., rankings on the results page) that vary across consumers, and a vector of demographics $y_i$. For each product $j$, the data  may also contain attributes $x_j$ that are common across all consumers. The model predicts choice probabilities as a function of $p_i$, $y_i$, $x = (x_j)_{j=1}^J$, and a parameter vector $\theta$:
\begin{equation}\label{eq:choice probs exogenous}
 Pr(d_{ij} = 1|p_i, y_i, x; \theta) = \sigma_j(p_i, y_i, x; \theta), \quad j =1,\ldots,J.
\end{equation}
This setting covers many empirical examples. Here we give just one standard model.

\begin{example}[Mixed Logit with Fixed Effects]
\label{ex:mixed logit}
    The utility that individual $i$ derives from good $j$ is of mixed logit form with microdata:
    \[
    \begin{aligned}
    	u_{ij} & = \alpha_i' p_{ij} + \beta_i' x_j + y_i' \Pi x_j + \xi_j + \varepsilon_{ij} , & & & j = 1,\ldots,J , \\
    	u_{i0} & = \varepsilon_{i0} ,
    \end{aligned}
    \]
    where $\xi_j$ is a product fixed effect and the $\varepsilon_{ij}$ are iid type 1 extreme value random variables. The vector $x_j$ collects characteristics with random coefficients and/or interactions with $y_i$; remaining characteristics are absorbed into the product fixed effects. Similarly, the means of the random coefficients $\beta_i$ can be normalized to zero due to the presence of the fixed effects. Choice probabilities are 
    \[
    \sigma_j(p_i, y_i, x; \theta) = \int \frac{e^{ \alpha' p_{ij} + \beta' x_j + y_i' \Pi x_j + \xi_j }}{1+\sum_{k =1}^J e^{ \alpha' p_{ik} + \beta' x_k + y_i' \Pi x_k + \xi_k }} dF(\alpha, \beta; \theta), \hspace{10pt} j = 1,\ldots,J, 
    \]
	where $F$ is a parametric distribution and $\theta$ contains $(\xi_j)_{j=1}^J$ and other parameters.
\end{example}

\bigskip

Our point of departure is again to partition $x_j \equiv (\bar x_j,e_j)$, where $\bar x_j$ collects quantitative product attributes observed by the practitioner, such as product size, and $e_j$ is an $r$-vector of product attributes that are not observed  in the data. Instead, one has proxies $\tilde e = (\tilde e_j')_{j=1}^J$ for the true underlying $e = (e_j')_{j=1}^J$. We stay agnostic on the form that $\tilde e$ takes. A leading case is when $\tilde e$ are embeddings computed to represent unstructured data, though our approach may equally be used when $\tilde e$ represents any product attributes that fail to correctly capture the true dimensions of differentiation.

\subsection{Bias Correction for Counterfactuals}
\label{sec:bc}

We again consider a broad class of counterfactuals that can be written as 
\[
\kappa = \E[ k(p_i, y_i, \bar x, e; \theta )],
\]
where the expectation is over the distribution of $(p_i,y_i)$, and $\bar x = (\bar x_j)_{j=1}^J$. Section~\ref{sec: counterfactual case 1} lists several counterfactuals of interest that are covered by this class. 

In the usual workflow, model parameters are estimated using the observed data and a candidate set of proxies $\tilde e$. For instance, one could use maximum likelihood  
estimation based on
\[
\sum_{i=1}^n \sum_{j = 0}^J d_{ij} \log \sigma_{j}(p_i,y_i,\bar x, \tilde e; \theta),
\]
where $\sigma_0 = 1-\sum_{j=1}^J \sigma_j$ is the probability of the outside option and $d_{i0} = 1 - \sum_{j=1}^J d_{ij}$.
Given an estimate $\hat \theta$ of $\theta$, the counterfactual $\kappa$ is usually estimated as 
\begin{equation}
\label{eq:kappa.naive}
	\hat \kappa = \frac 1n \sum_{i=1}^n k(p_i, y_i, \bar x, \tilde e; \hat \theta ).
\end{equation}
As before, we refer to this as the \emph{naive estimator} of $\kappa$ since it does not account for the fact that the proxies $\tilde e$ might differ from the true latent attributes. This proxy error affects the naive estimator directly, as $\tilde e$ is an argument of $k$, and indirectly through the first-stage estimate $\hat \theta$, as $\tilde e$ enters the likelihood. 

We now introduce our bias-correction. Again, we consider models in which $e$ and $\theta$ enter choice probabilities (\ref{eq:choice probs exogenous}) via a composite parameter
\[ 
\gamma \equiv \gamma(\theta, e).
\]
Many common models have this property. We illustrate it in our leading example.

\medskip

\begin{continuance}{ex:mixed logit}
We partition $\beta_i = (\beta_{\bar x, i}, \beta_{e, i})$ and $\Pi = [\Pi_{\bar x} \; \Pi_e]$ and write
\[
u_{ij} = \alpha_i' p_{ij} +  y_i' \Pi_{\bar x} \bar x_j + y_i' \Pi_e e_j + \xi_j + \beta_{\bar x,i}' \bar x_j + \beta_{e,i}' e_j + \varepsilon_{ij}, \quad j = 1,\ldots,J,
\]
Suppose $\alpha_i \sim N(\bar \alpha, \Sigma_{\alpha})$, $\beta_{\bar x, i} \sim N(0, \Sigma_{\bar x})$, and $\beta_{e, i} \sim N(0, \Sigma_e)$, and $\alpha_i$, $\beta_{\bar x, i}$, and $\beta_{e, i}$ are independent. Recall the notation $l$ and $l_r$ from Example~\ref{ex:blp}. Then,
\[
 \theta = (\bar \alpha, \xi, v(\Pi_{\bar x}), v(\Pi_e), l(\Sigma_{\alpha}), l(\Sigma_{\bar x}), l(\Sigma_{e})),
\]
where $\xi = (\xi_j)_{j=1}^J$, and $v(\Pi)$ stacks the entries of $\Pi$ into a vector. Collecting $\beta_{e,i}' e_j$ across products, we have $e \beta_{e, i} \sim N(0, e \Sigma_e e')$, where $e \Sigma_e e'$ has rank $r \leq J$. Hence,
\[
 \gamma(\theta, e) = (\bar \alpha, \xi, v(\Pi_{\bar x}), v(e \Pi_{e}'), l(\Sigma_{\alpha}), l(\Sigma_{\bar x}), l_r(e\Sigma_{e}e')).
\]
As before, parameters that do not interact with $e$, such as $\bar \alpha$ and $\xi$, are left unchanged in $\gamma$. The remaining components of $\gamma$ capture how $e$ and $\theta$ interact to jointly affect utilities.
\end{continuance}

\bigskip

We assume in what follows that choice probabilities $\sigma$ and the counterfactual function $k$ depend on $(\theta, e)$ via $\gamma(\theta, e)$, and write 
\[
\begin{aligned}
	\sigma_{ij}(\gamma(\theta, e)) & = \sigma_j(p_i, y_i, \bar x, e; \theta) , \quad j = 1,\ldots,J, \\
    k_i(\gamma(\theta, e)) & = k(p_i, y_i, \bar x, e; \theta) .
\end{aligned}
\]
We then define the \emph{bias-corrected estimator}
\begin{equation}\label{eq:kappa.bc}
	\hat \kappa_{bc} = \frac 1n \sum_{i=1}^n \left( k_i(\hat \gamma) + \hat c_i'(d_i - \sigma_i(\hat \gamma)) \right),
\end{equation}
where $\hat \gamma = \gamma(\hat \theta, \tilde e)$ is the value of $\gamma$ pinned down by $\tilde e$ and the first-stage estimate $\hat \theta$, and $\hat c_i$ is a $J$-vector of weights given by 
\begin{equation}\label{eq:c.hat}
	\hat c_i = V_i (\hat \gamma)^{-1}  \dot \sigma_i(\hat \gamma) \hat H^{-1} \hat k,
\end{equation}
$d_i = (d_{ij})_{j=1}^J$ and $\sigma_i(\gamma) = (\sigma_{ij}(\gamma))_{j=1}^J$ are $J$-vectors, $V_i(\gamma) = \mathrm{diag} (\sigma_i(\gamma)) - \sigma_i(\gamma) \sigma_i(\gamma)'$ is  $J \times J$, $\hat H= \frac 1n \sum_{i=1}^n \dot \sigma_i(\hat \gamma)' V_i (\hat \gamma)^{-1} \dot \sigma_i(\hat \gamma)$ is  $\dim(\gamma) \times \dim(\gamma)$, $\hat k = \frac 1n \sum_{i=1}^n \dot k_i(\hat \gamma)$ and $\dot k_i(\gamma) = \frac{\partial k_i(\gamma)}{\partial \gamma}$ are $\dim(\gamma) \times 1$, and $\dot \sigma_i(\gamma)' = \frac{\partial \sigma_i(\gamma)'}{\partial \gamma}$ is $\dim(\gamma) \times J$. Importantly, $\hat \kappa_{bc}$ requires minimal additional computation, as it involves closed-form expressions of objects that can be easily calculated given $\hat \theta$.

As before, $\hat \kappa_{bc}$ takes the naive estimator and adds an adjustment term that purges the effect of $\hat \gamma$ on $\hat \kappa_{bc}$ to first order. 
Proposition \ref{prop:kappa.exogenous} below shows that $\hat \kappa_{bc}$ is asymptotically unbiased and its distribution does not depend on $\hat \gamma$. We highlight a few key implications of this result.

\begin{remark}[Easy to compute standard errors] 
\label{rmk:standard errors MLE model}
	The asymptotic variance of $\hat \kappa_{bc}$ can be easily estimated using
	\begin{equation}
    \label{eq:standard errors}
	\hat V_{bc} = \hat s_k^2 + \frac 1n \sum_{i=1}^n \hat c_i ' V_i(\hat \gamma) \hat c_i, 
	\end{equation}
	where $\hat s_k^2$ is the sample variance of $k_i(\hat \gamma)$. Standard errors for $\hat \kappa_{bc}$ are then $\sqrt{\hat V_{bc}/n}$.
\end{remark}

\begin{remark}[Fine-Tuning]
The bias correction in \eqref{eq:kappa.bc} allows the proxies $\tilde e$ to be sample-dependent, for instance because an ML algorithm has been fine-tuned on the choice data to provide a better fit. As before, there is no need to correct the standard errors: one can use $\sqrt{\hat V_{bc}/n}$ with $\hat V_{bc}$ as above whether or not $\tilde e$ is data-dependent.
\end{remark}

\begin{remark}[Semiparametric Efficiency]
We may view model (\ref{eq:choice probs exogenous}) as a conditional moment model based on the moment condition
\begin{equation} \label{eq:moment}
    \sigma_i(\gamma_0) = \E[d_i|p_i,y_i,\bar x].
\end{equation}
In Section~\ref{sec:theory model 2}, we show that the asymptotic variance of $\hat \kappa_{bc}$ is the semiparametric efficiency bound for estimating $\kappa$ in model (\ref{eq:moment}) \citep{brown1998efficient,ai2012semiparametric}. Thus, there do not exist regular estimators of $\kappa$ in model (\ref{eq:moment}) with smaller asymptotic variance than the bias-corrected estimator $\hat \kappa_{bc}$.
\end{remark}
	
\subsection{Diagnostics}

Here we provide variants of the two diagnostics introduced in Section~\ref{sec:blp diagnostics} that can be used to assess the suitability of a candidate set of proxies $\tilde e$. We illustrate the use and effectiveness of these diagnostics in the empirical application.
	
\subsubsection{Diagnostic 1: Is $\hat \gamma$ Close To $\gamma_0$?}
\label{sec:goodness of fit}
	
As before, a simple LM statistic can be used to validate whether $\hat \gamma$ is sufficiently close to $\gamma_0$, as required by the theory for $\hat \kappa_{bc}$ in Proposition~\ref{prop:kappa.exogenous}. This diagnostic also helps guide the choice of proxies (e.g., embeddings from various data sources or ML models). 
The first diagnostic is
\begin{equation}
\label{eq:LM1}
LM_1 = \|\sqrt n \hat H^{-1/2} \hat S \|^2,
\end{equation}
where 
$\hat S = \frac 1n \sum_{i=1}^n \sum_{j=0}^J \frac{d_{ij}}{\sigma_{ij}(\hat \gamma)} \dot \sigma_{ij}(\hat \gamma)$ is the score and $\hat H$ is the (expected) Hessian, defined below (\ref{eq:c.hat}). Proposition~\ref{prop:lm.fit.gamma} below shows that $LM_1$ behaves like $\|\sqrt n (\hat \gamma - \gamma_0)\|^2$ as $n$ grows large. This allows researchers to validate the assumption that the discrepancy between $\hat \gamma$ and $\gamma_0$ is sufficiently small, as needed in Proposition~\ref{prop:kappa.exogenous}, by checking whether $LM_1$ is below a threshold. See the practitioner's guide below for a concrete recommendation and Section~\ref{sec:theory model 2} for a formal discussion. Further, Proposition~\ref{prop:lm.fit.e} shows that $LM_1$ can be used to bound the distance between $\tilde e$ and the true latent attributes. Thus, $LM_1$ also serves as a criterion to target when choosing among proxies and/or fine-tuning.
	
\subsubsection{Diagnostic 2: Is The Dimension of $\tilde e$ Correct?}
\label{sec:rank}

We again provide guidance on how many proxies to include per product with a second diagnostic, denoted $LM_2$. The diagnostic is based on an LM statistic for testing whether $\tilde e$ is of sufficient dimension, and its construction follows similar ideas to Section~\ref{sec:rank blp}. As with $LM_1$, it depends only on $\hat \theta$ and therefore requires minimal additional computation. Full details are deferred to Appendix~\ref{app: LM2 case 2}.

\subsection{Practitioner's Guide}

Given a counterfactual of interest $\kappa$ and a candidate set of proxies $\tilde e$:
\begin{enumerate}
    \item Compute the model parameter estimates $\hat \theta$ as usual, treating $\tilde e$ as the truth.
    \item Compute weights $\hat c$ in \eqref{eq:c.hat}.
    \item Plug the weights into \eqref{eq:kappa.bc} to compute the bias-corrected estimator $\hat \kappa_{bc}$.
    \item Compute standard errors for $\hat \kappa_{bc}$ using the formula in Remark \ref{rmk:standard errors MLE model}.
    \item To check whether a given set of proxies $\tilde e$ is adequate:
    \begin{enumerate}
        \item Compute the $LM_1$ statistic in \eqref{eq:LM1}. If it is below a threshold $C_n^2$, conclude that the $\hat \gamma$ induced by $\tilde e$ is sufficiently close to $\gamma_0$. In Section~\ref{sec:goodness.of.fit.theory}, we motivate a threshold of $C_n^2 = \chi^2_{\dim(\gamma), 0.95}\log n$.
        \item Compute the $LM_2$ statistic in \eqref{eq:lm.rank}. If it is below a threshold $\hat \xi^*_{0.95}$, conclude that $\tilde e$ is of adequate dimension. The bootstrap method in Appendix~\ref{app: LM2 case 2} can be used to compute $\hat \xi^*_{0.95}$.
    \end{enumerate}
\end{enumerate}

\section{Simulations}
\label{sec:simulations}

We first illustrate our approach in simulations, considering each of the two models discussed in Sections~\ref{sec:blp} and~\ref{sec: exog prices} in turn.

\subsection{Case 1: Endogenous Prices}
\label{sec:simulations case 1}

We start by considering the BLP-type model of Section~\ref{sec:blp}. As proof of concept, we use a DGP calibrated to the standard \cite{nevo2001measuring} ready-to-eat cereal market data, based on the replication data in the PyBLP tutorial \citep{conlon2020best} (see also \cite{nevoPractitionersGuideEstimation2000}). For each market, the data contain the shares of $J=24$ products, their prices, their sugar content and a dummy variable capturing whether a product is ``mushy'' or not. First, we fit a mixed logit model to the data using PyBLP.\footnote{We use the third model estimated in the PyBLP tutorial, which has random coefficients on price, sugar content, the mushy dummy and the constant for the inside products, with a diagonal covariance matrix. Product fixed effects are included.} For $1{,}000$ simulations, we draw prices and exogenous variables for $T=1{,}200$ markets\footnote{We choose $T=1{,}200$ since that is roughly the sample size ($T=1{,}124$) in \cite{nevo2001measuring}.} from a distribution calibrated to match their mean and variance in the data, and generate market shares using parameters similar to those estimated in the first step (full details of the DGP are in Appendix~\ref{app:simulations design}).\footnote{As shown in the PyBLP tutorial, the instruments provided with the dataset are not enough to identify the standard deviations on the random coefficients well. To address this, we form additional differentiation IVs following the approach of \cite{gandhiMeasuringSubstitutionPatterns2019}; see Appendix~\ref{app:simulations}.} We consider the possibility that the mushy variable is an imperfect proxy, which seems plausible given that the variable was hand-coded by the author based on his own judgment. For each simulated dataset and each product, we generate data treating the binary mushy variable in the original data as the truth ($e_j$), and estimate model parameters and counterfactuals based on a corrupted version given by $\tilde e_j \sim \mathrm{Beta}(\rho, 2-\rho)$ if $e_j = 0$, and $\tilde e_j \sim \mathrm{Beta}(2-\rho, \rho)$ if $e_j = 1$ for $\rho = 0.1, 0.2, \ldots, 0.6$. The case $\rho = 0$ corresponds to $e_j = \tilde e_j$ (no proxy error) while larger values of $\rho$ correspond to more severe error.\footnote{The $\tilde e_j$ have mean $\rho/2$ and $1-\rho/2$ if $e_j = 0$ and $e_j = 1$, respectively. Their distributions approach a uniform distribution on $[0,1]$ as $\rho \to 1$.} This design matches the scale of one of the canonical applications of the BLP model, allowing us to illustrate the performance of our approach on standard datasets.

We focus on estimation of diversion, i.e., the fraction of consumers that switch from one group of products to another one when the former is removed from the choice set. Diversions are a key metric to understand the nature of competition in a market \citep{conlon2021empirical}. Specifically, we consider two counterfactuals: (i) the diversion from \emph{all} products of firm 1 to \emph{all} products of firm 2; (ii) the diversion from all \emph{mushy} products of firm 1 to all \emph{mushy} products of firm 2.\footnote{We focus on firms 1 and 2 as they have the largest assortments of products in the dataset.} 

\begin{figure}[t]
	\centering
	\begin{subfigure}[t]{0.48\textwidth}
		\centering
		\includegraphics[width=\linewidth]{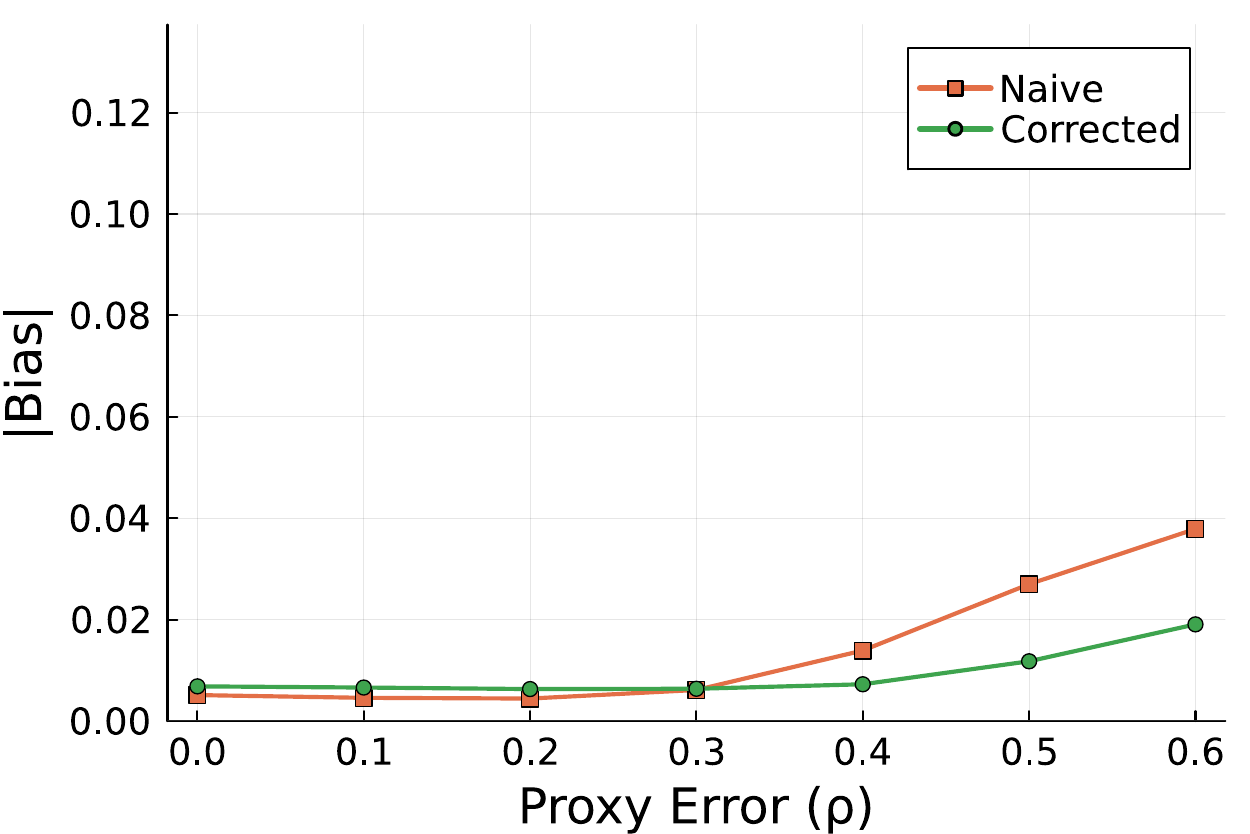}
		\caption{Bias for the ``all'' counterfactual}
	\end{subfigure}
	\hfill 
	\begin{subfigure}[t]{0.48\textwidth}
	\centering
	\includegraphics[width=\linewidth]{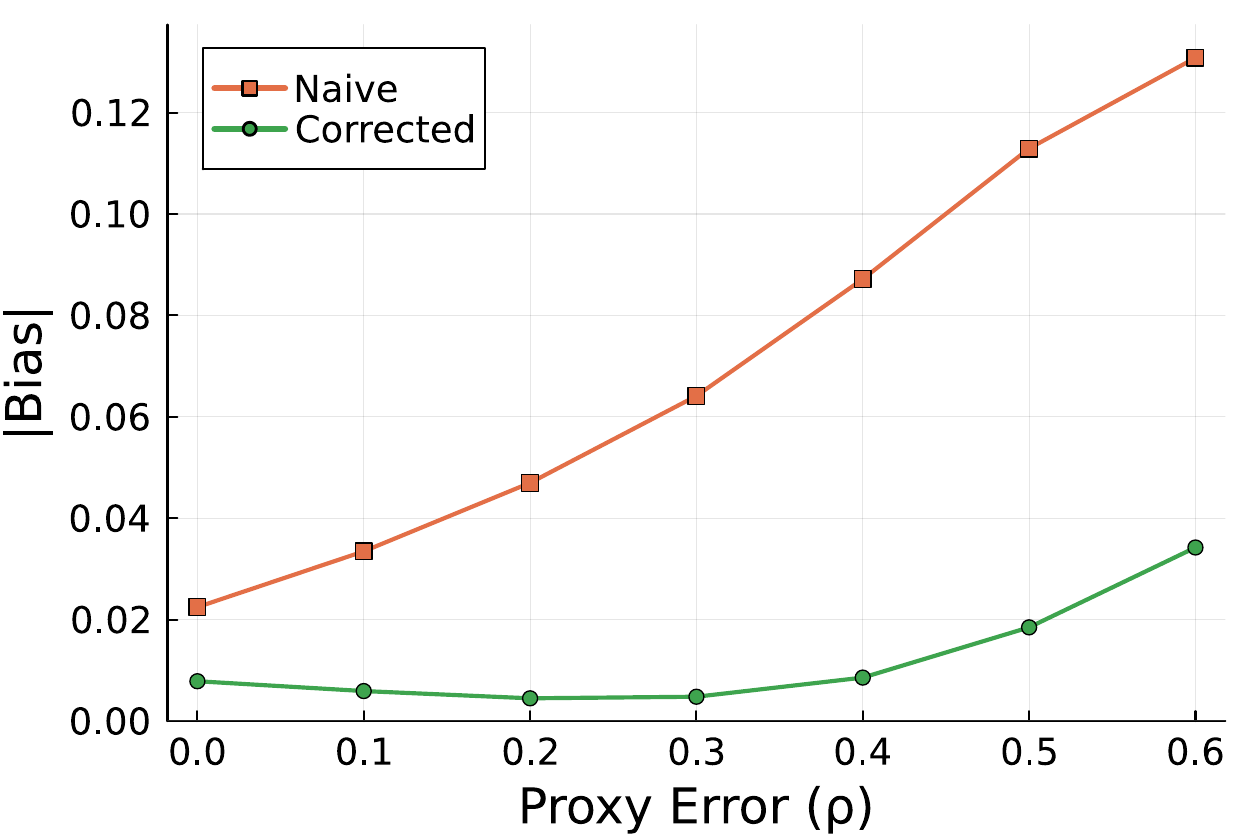}
	\caption{Bias for the ``mushy'' counterfactual}
\end{subfigure}
	
	\par\vspace{1em} 
	
		\begin{subfigure}[t]{0.48\textwidth}
		\centering
		\includegraphics[width=\linewidth]{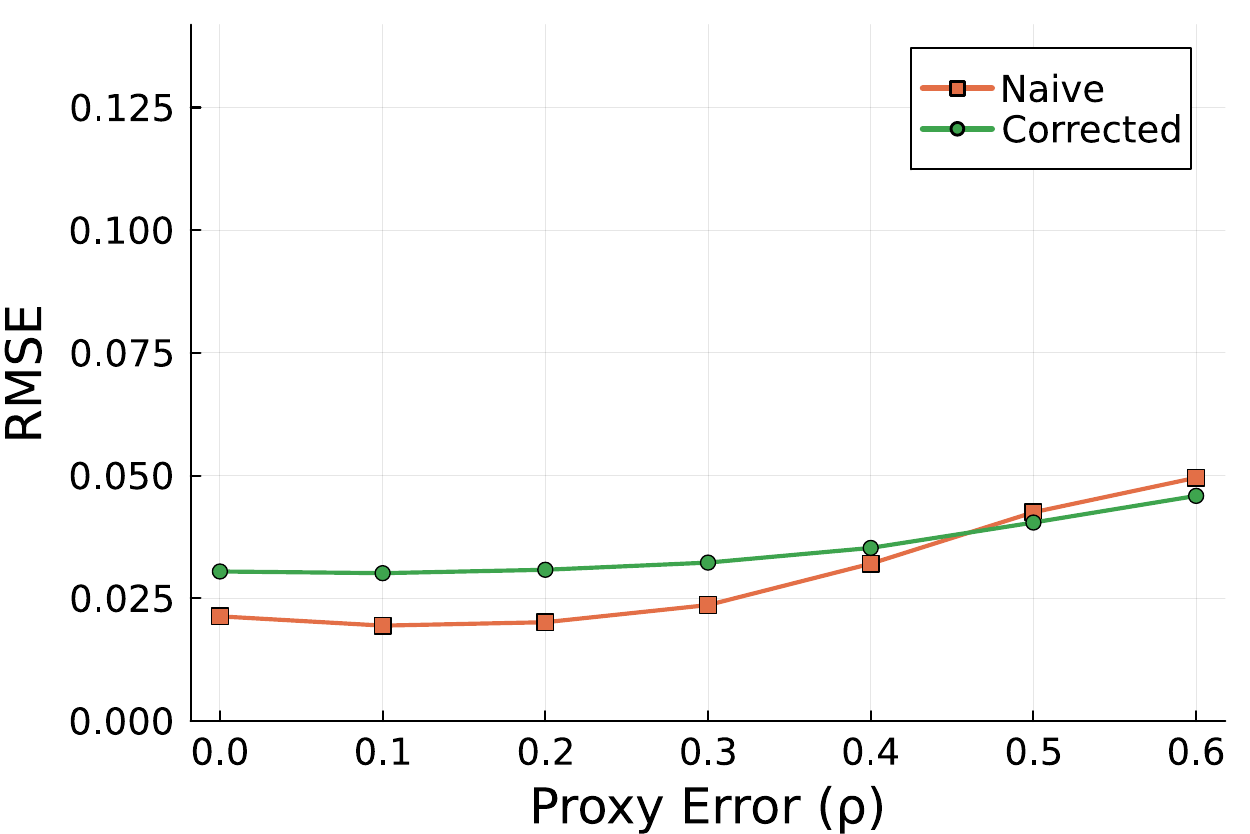}
		\caption{RMSE for the ``all'' counterfactual}
	\end{subfigure}
	\hfill 
	\begin{subfigure}[t]{0.48\textwidth}
	\centering
	\includegraphics[width=\linewidth]{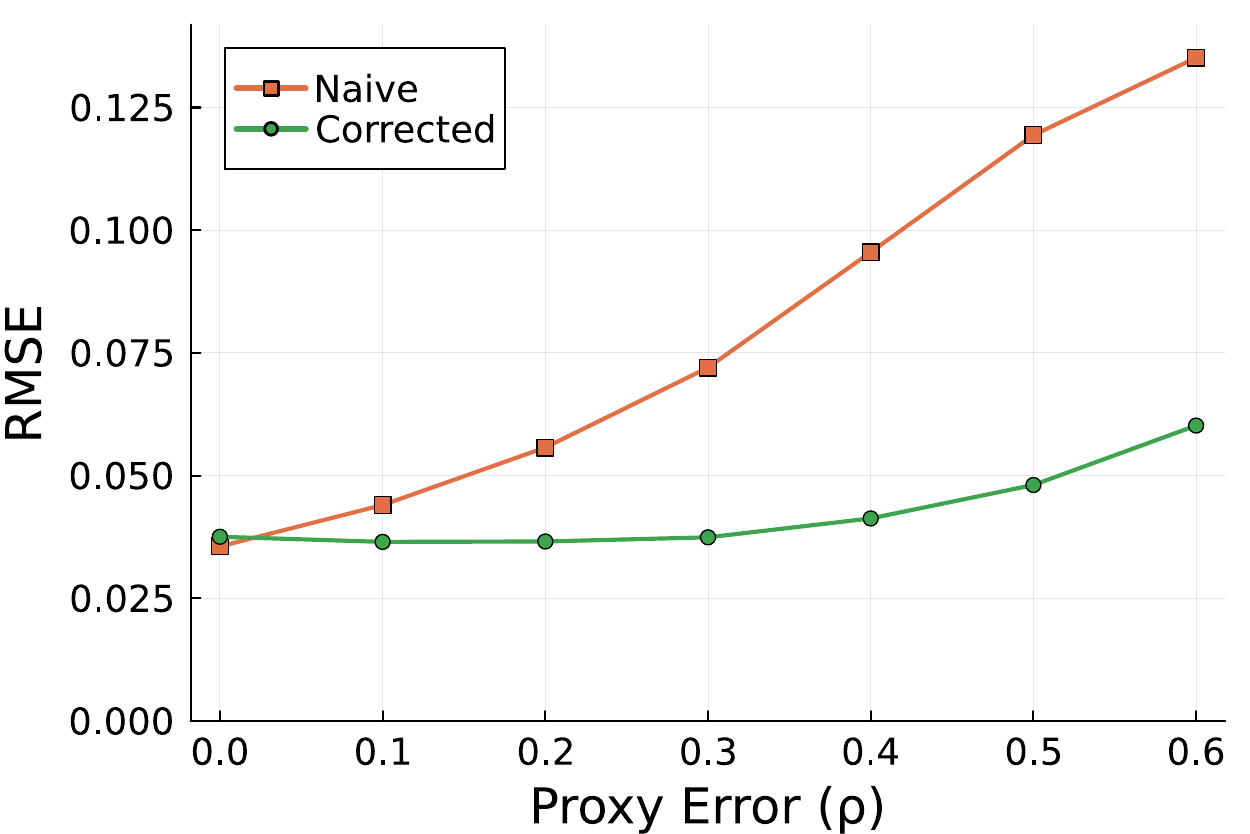}
	\caption{RMSE for the ``mushy'' counterfactual}
\end{subfigure}

	\caption{Bias and RMSE as a function of proxy error, \cite{nevo2001measuring} design}
	\label{fig:bias rmse sim case 1}

\vskip 0.5em

\begin{minipage}{\textwidth}
    \small \textit{Note:} Figures show the absolute bias and RMSE across simulations of the naive estimator $\hat \kappa$ in~\eqref{eq: kappa naive blp} and bias-corrected estimator $\hat \kappa_{bc}$ in~\eqref{eq: kappa_bc blp linear no micro} for the ``all'' and ``mushy'' counterfactuals.
\end{minipage}
  
\end{figure}

\begin{figure}[t]
  \centering
  \caption{Distribution of the naive and bias-corrected estimators, \cite{nevo2001measuring} design}
  \label{fig:est distribution sim case 1} 
  
  \vspace{5pt}
  \begin{subfigure}[t]{0.48\textwidth}
    \centering
    \includegraphics[width=\linewidth]{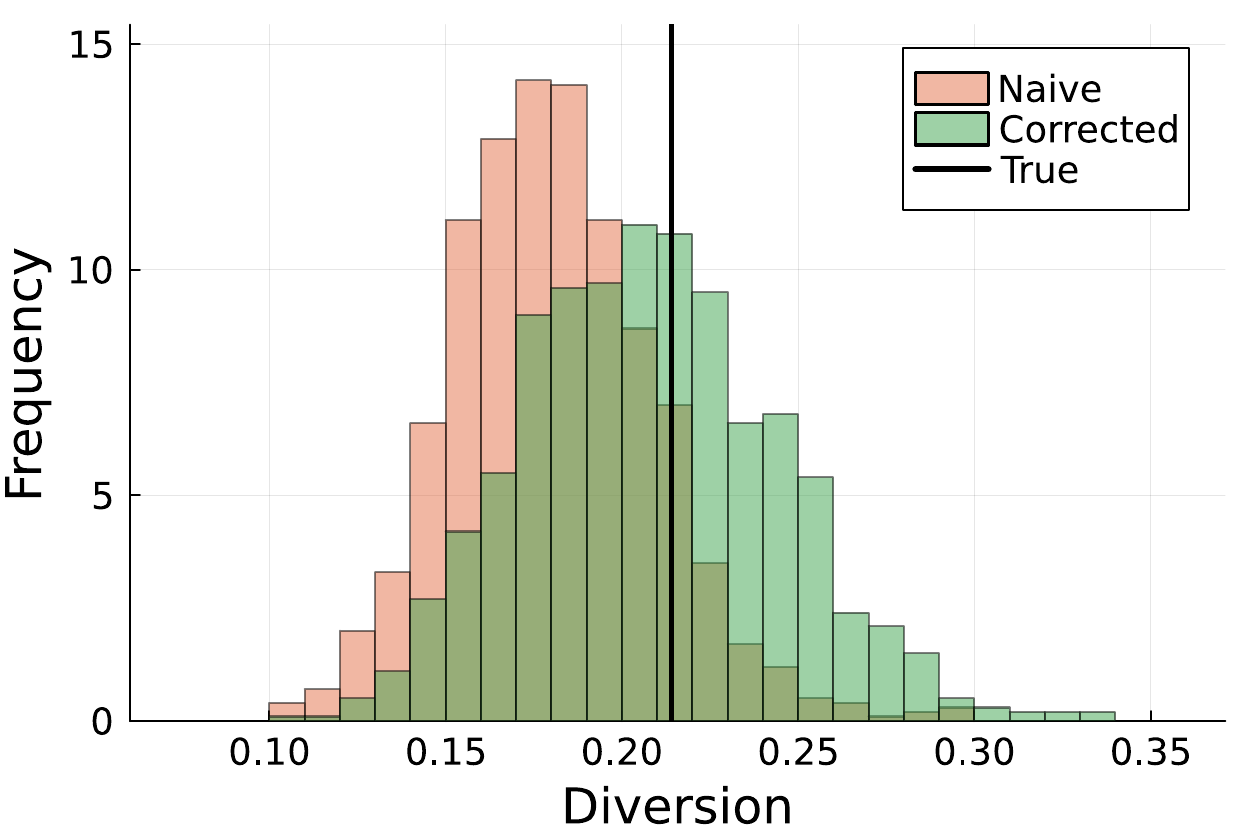}
    \caption{Proxy error $\rho=0.1$}
  \end{subfigure}
  \hfill 
  \begin{subfigure}[t]{0.48\textwidth}
    \centering
    \includegraphics[width=\linewidth]{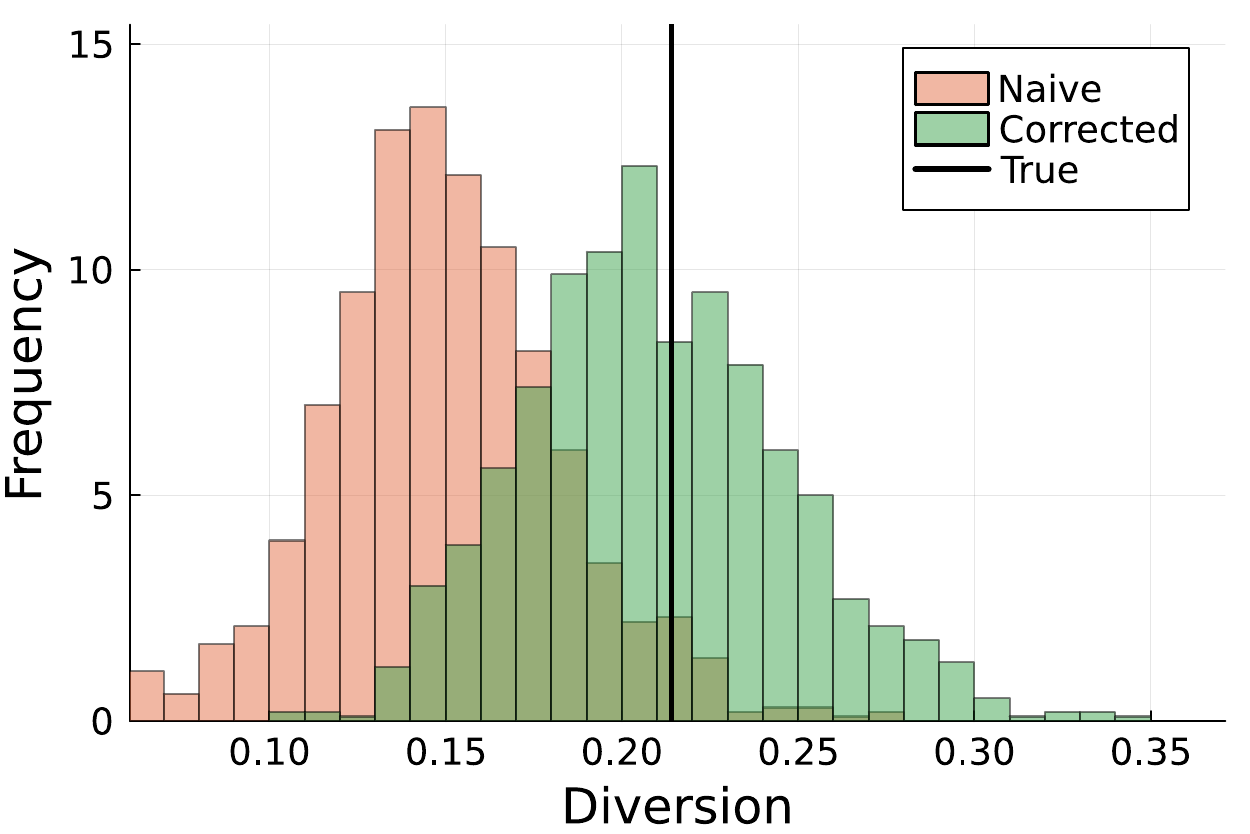}
    \caption{Proxy error $\rho=0.3$}
  \end{subfigure}
    
  \begin{minipage}{\textwidth}
    \small \textit{Note:} Histograms show the distribution across simulations of the naive estimator $\hat \kappa$ in~\eqref{eq: kappa naive blp} and bias-corrected estimator $\hat \kappa_{bc}$ in~\eqref{eq: kappa_bc blp linear no micro} for the ``mushy'' counterfactual.
  \end{minipage}
\end{figure}

We estimate the model parameters $\theta$ using PyBLP then compute the naive estimator $\hat \kappa$ in~\eqref{eq: kappa naive blp} and bias-corrected estimator $\hat \kappa_{bc}$ in~\eqref{eq: kappa_bc blp linear no micro}  of the diversions.
The left two panels of Figure~\ref{fig:bias rmse sim case 1} show the absolute bias and RMSE of the naive and bias-corrected estimators as a function of the level of proxy error $\rho$ for the ``all'' counterfactual. As firms~1 and~2 sell both mushy and non-mushy cereal, this counterfactual involves aggregating across the two types of products, resulting in small bias except for large values of $\rho$, where the correction reduces bias by around 50\%. 

When $\tilde e$ is a perfect proxy, both the naive and bias-corrected estimators have very low bias. Our estimator has a marginally higher RMSE, indicating that its variance is slightly higher than that of the naive estimator. This is intuitive: since the naive estimator leverages the assumption that the proxies $\tilde e$ are correct and ours does not, we obtain slightly less precise estimates when that assumption happens to be correct. However, for larger values of $\rho$ this variance increase is more than offset by the bias reduction, and the bias-corrected estimator has lower RMSE than the naive estimator. 

The right two panels of Figure~\ref{fig:bias rmse sim case 1} show the results for the ``mushy'' counterfactual. Here, the naive estimator has much larger bias than the bias-corrected estimator across all levels of proxy error. Indeed,  the bias-corrected estimator is essentially unbiased for $\rho = 0, 0.1, \ldots, 0.4$, while for larger $\rho$ the bias is only a fraction of that of the naive estimator. To see this in more detail, histograms of the naive and bias-corrected estimators for $\rho = 0.1$ and $\rho = 0.3$ are shown in Figure~\ref{fig:est distribution sim case 1}. The distribution of the bias-corrected estimator is stable and centered around the truth, while the distribution of the naive estimator drifts to the left for larger $\rho$. The RMSE of the corrected estimator is uniformly smaller, except in the knife-edge case of no proxy error where it is essentially on par with the naive estimator. 

Finally, Table~\ref{tab:coverage case 1} displays the coverage of 95\% confidence intervals for the counterfactuals based on the bias-corrected estimator and standard errors computed using the formula~\eqref{eq:standard errors blp no micro data}. Coverage is close to the 95\% target for $\rho = 0, 0.1, \ldots, 0.4$, while for larger values of $\rho$ the bias eventually distorts coverage. This is still a substantial improvement over the naive estimator. For instance, Figure~\ref{fig:est distribution sim case 1} shows that the coverage of confidence intervals based on the naive estimator would be well below 50\% for $\rho$ as low as $0.1$, and near zero for $\rho = 0.3$. Overall, our approach delivers easy-to-implement inference on counterfactuals and is robust to degrees of proxy error that severely distort inference in the standard workflow.

\begin{table}[t]
    \centering
    \begin{tabular}{ccc} \hline \hline 
    Proxy error $\rho$ & \phantom{111}Coverage: ``all''\phantom{111} & Coverage: ``mushy'' \\  \hline \hline 
    0.00 & 0.956 & 0.967 \\
    0.10 & 0.963 & 0.972 \\
    0.20 & 0.967 & 0.977 \\
    0.30 & 0.962 & 0.971 \\
    0.40 & 0.954 & 0.943 \\
    0.50 & 0.924 & 0.880 \\
    0.60 & 0.897 & 0.773 \\  \hline \hline
    \end{tabular}
    \caption{Coverage of 95\% confidence intervals for counterfactuals, \cite{nevo2001measuring} design}
    \label{tab:coverage case 1}

    \vskip 0.5em

  \begin{minipage}{\textwidth}
    \small \textit{Note:} Fraction of simulations in which confidence intervals for counterfactuals contain the true value. Confidence intervals are computed as $\hat \kappa_{bc} \pm 1.96 (\hat V_{bc}/T)^{1/2}$ with $\hat \kappa_{bc}$ in~\eqref{eq: kappa_bc blp linear no micro} and $\hat V_{bc}$ as in Remark~\ref{rmk:std errors}.
  \end{minipage}
\end{table}

Interestingly, we obtain very similar results when we include additional differentiation IVs based on the distance between products in the space of (possibly mismeasured) $\tilde e_j$. See Appendix~\ref{app: sims e^tilde IVs} for additional details and results. While our theory does not explicitly cover this case, this finding suggests that researchers  might be able to continue using standard BLP-type instruments even when the non-price attributes are imperfect proxies.

\subsection{Case 2: Individual-Level Price Variation}
\label{sec:simulations case 2}

Next, we consider the model in Section \ref{sec: exog prices}. In the simulation design, we set $J = 10$ products and $n = 10{,}000$ consumers, in line with the empirical application in Section~\ref{sec:application}. We model utility as a function of price and two-dimensional latent attributes $e$ (in addition to idiosyncratic shocks). The latent attributes $e_j$ are drawn iid $N(0,1)$ across products. Individual-level prices are drawn iid $N(5,1)$ and vary across simulations. The random coefficient on price is $N(-1,0.3^2)$ and the coefficients on the latent attributes are independent $N(0,0.75^2)$. As before, we hold $e$ fixed in the data generating process and vary the amount of error in $\tilde e$. Specifically, for every $j$, we let
\begin{equation}
	\label{eq:e tilde sim}
	\tilde e_j = (1-\rho/2) e_j + \sqrt{1-(1-\rho/2)^2} \eta_j,
\end{equation}
where $\eta_j$ are iid $N(0,1)$ across simulations. The parameter $\rho$ determines the amount of error in $\tilde e$. When $\rho=0$, the  proxies exactly match the latent attributes; as $\rho$ increases, the proxies are increasingly uninformative. We note that this simulation design imposes very few restrictions on the form of proxy error: each element of $\tilde e_j$ could be smaller or larger than the corresponding element of $e_j$, and this varies freely across goods $j$.\footnote{Note that \eqref{eq:e tilde sim} is such  that $\tilde e_j$ has roughly the same amount of variation across goods $j$ as $e_j$ does. This allows us to isolate the effect of proxy error in a way that is not confounded by changes in the scale of the proxies $\tilde e_j$ used in estimation.}$^,$\footnote{Similar results were obtained in a design where true attributes $e_j$ are computed by PCA on a large set of variables, and proxies $\tilde e_j$ are computed by PCA on a corrupted version of the variables.}
Finally, we include product fixed effects in estimation but set these to zero when generating data. Because of the fixed effects, the only channel through which the proxies enter the model is via the covariance of the random component of utilities. Thus, the model will discard the proxies if the variances of their random coefficients collapse to zero.

We focus on diversion from product 3 to 4 as this pair has the smallest diversion at the true parameters and latent attributes, so correcting bias is (potentially) most critical. Figure \ref{fig:est distribution sim} plots histograms of the naive estimator~\eqref{eq:kappa.naive} and our bias-corrected estimator~\eqref{eq:kappa.bc} across $1{,}000$ simulations. As the level of proxy error increases, the distribution of the naive estimator moves away from the true value of the counterfactual (roughly 0.05), whereas the corrected estimator remains centered around the true value. Interestingly, for larger levels of proxy error, the distribution of the naive estimator ends up centered at the counterfactual prediction of a logit model with no random coefficients on proxies. This is intuitive: as the proxies $\tilde e$ become increasingly noisy, they capture less of the substitution patterns in the data, and the estimated variance of their random coefficients shrinks toward zero. As a result, the model collapses to one without random coefficients. This finding also serves as a warning that proxy error can defeat the purpose of estimating a random coefficients model in the first place: if proxy error is not properly accounted for, the model  may revert  to the restrictive substitution patterns that it was specifically intended to relax.

\begin{figure}[t]
	\centering
    	\caption{Distribution of the naive and bias-corrected estimators, \citetalias{compiani2025demand} design}
        	\label{fig:est distribution sim}
	\begin{subfigure}[t]{0.48\textwidth}
		\centering
		\includegraphics[width=\linewidth]{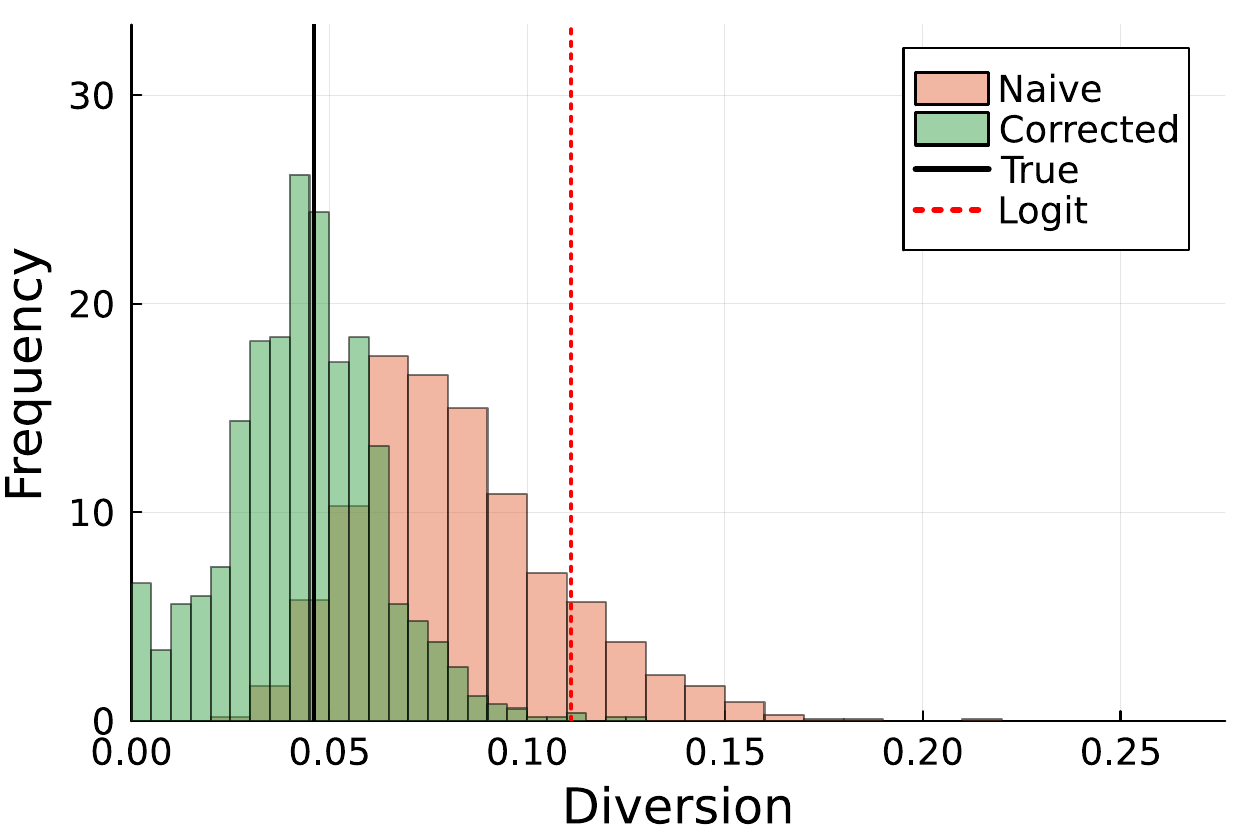}
		\caption{Proxy error $\rho = 0.3$}
	\end{subfigure}
	\hfill 
	\begin{subfigure}[t]{0.48\textwidth}
	\centering
	\includegraphics[width=\linewidth]{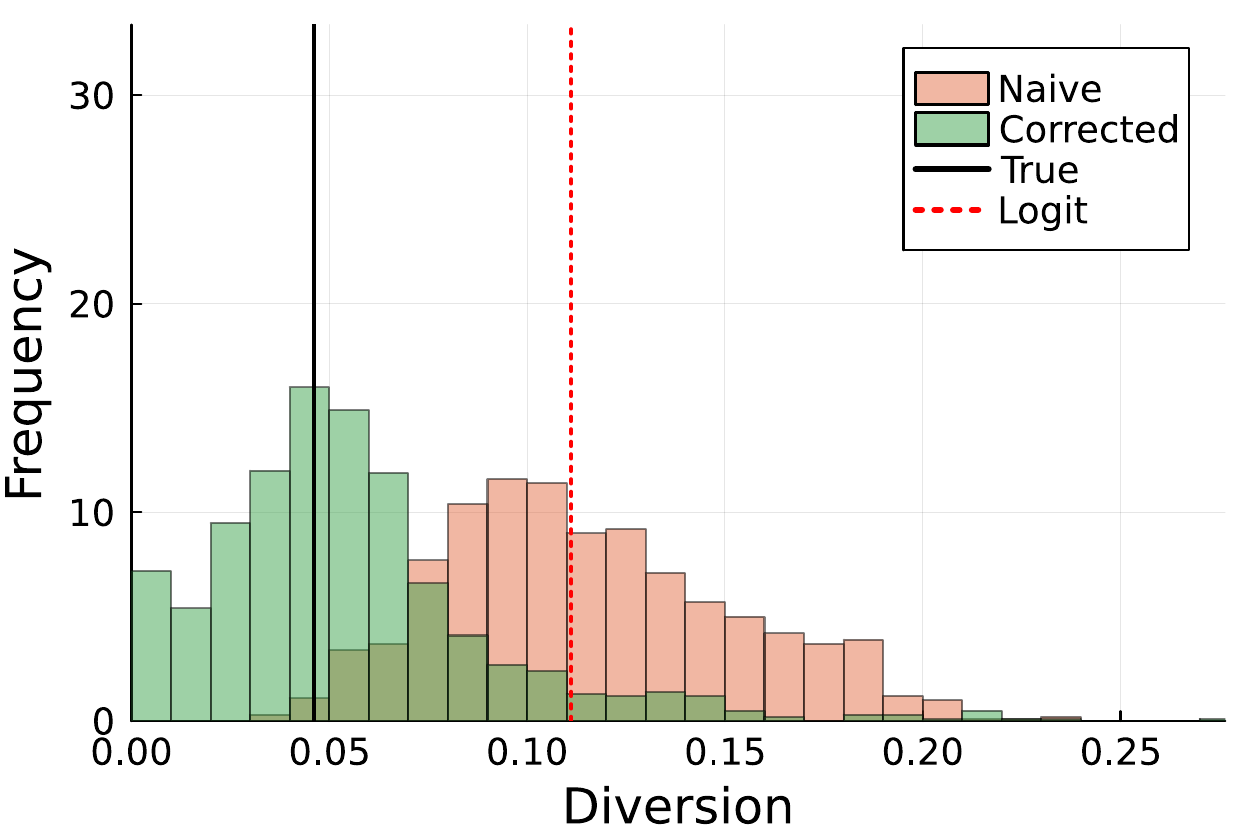}
	\caption{Proxy error $\rho = 0.6$}
\end{subfigure}

\begin{minipage}{\textwidth}
    \small \textit{Note:} Histograms show the distribution across simulations of the naive estimator $\hat \kappa$ in~\eqref{eq:kappa.naive} and bias-corrected estimator $\hat \kappa_{bc}$ in~\eqref{eq:kappa.bc} for the diversion from product 3 to 4.
  \end{minipage}
\end{figure}

To better assess the trade-offs involved in our bias correction, Figure \ref{fig:bias RMSE sim} shows how the bias and RMSE of the two estimators vary with $\rho$, and these are consistent with the findings for Case~1. In particular, in the knife-edge case in which $\tilde e$ is a perfect proxy, bias is not a concern and the correction marginally increases variance. Outside this knife-edge case,  the corrected estimator consistently achieves lower bias and RMSE than the naive estimator. 

Finally, Table~\ref{tab:coverage case 2} presents coverage of 95\% confidence intervals formed from our bias-corrected estimator~\eqref{eq:kappa.bc} with standard errors computed using the formula~\eqref{eq:standard errors}. These have coverage close to the 95\% target for $\rho$ up to $0.3$, while coverage inevitably breaks down when the proxies become very noisy. Again, these results show that our formulas deliver easy-to-implement, valid inference on counterfactuals and are robust to noisy proxies that substantially distort inference in the standard workflow.

\begin{table}[t]
    \centering
    \begin{tabular}{cc} \hline \hline 
    Proxy error  $\rho$ & Coverage \\  \hline \hline 
    0.00 & 0.957 \\
    0.10 & 0.957 \\
    0.20 & 0.929 \\
    0.30 & 0.910 \\
    0.40 & 0.893 \\
    0.50 & 0.863 \\
    0.60 & 0.816 \\ \hline \hline
    \end{tabular}
    \caption{Coverage of 95\% confidence intervals for counterfactuals, \citetalias{compiani2025demand} design}
    \label{tab:coverage case 2}

\vskip 0.5em

  \begin{minipage}{\textwidth}
    \small \textit{Note:} Fraction of simulations in which confidence intervals for the counterfactual contain the true value. Confidence intervals are computed as $\hat \kappa_{bc} \pm 1.96 (\hat V_{bc}/n)^{1/2}$ with $\hat \kappa_{bc}$ in~\eqref{eq:kappa.bc} and $\hat V_{bc}$ as in Remark~\ref{rmk:standard errors MLE model}.
  \end{minipage}
  
\end{table}

\begin{figure}[t]
	\centering
	\begin{subfigure}[t]{0.48\textwidth}
		\centering
		\includegraphics[width=\linewidth]{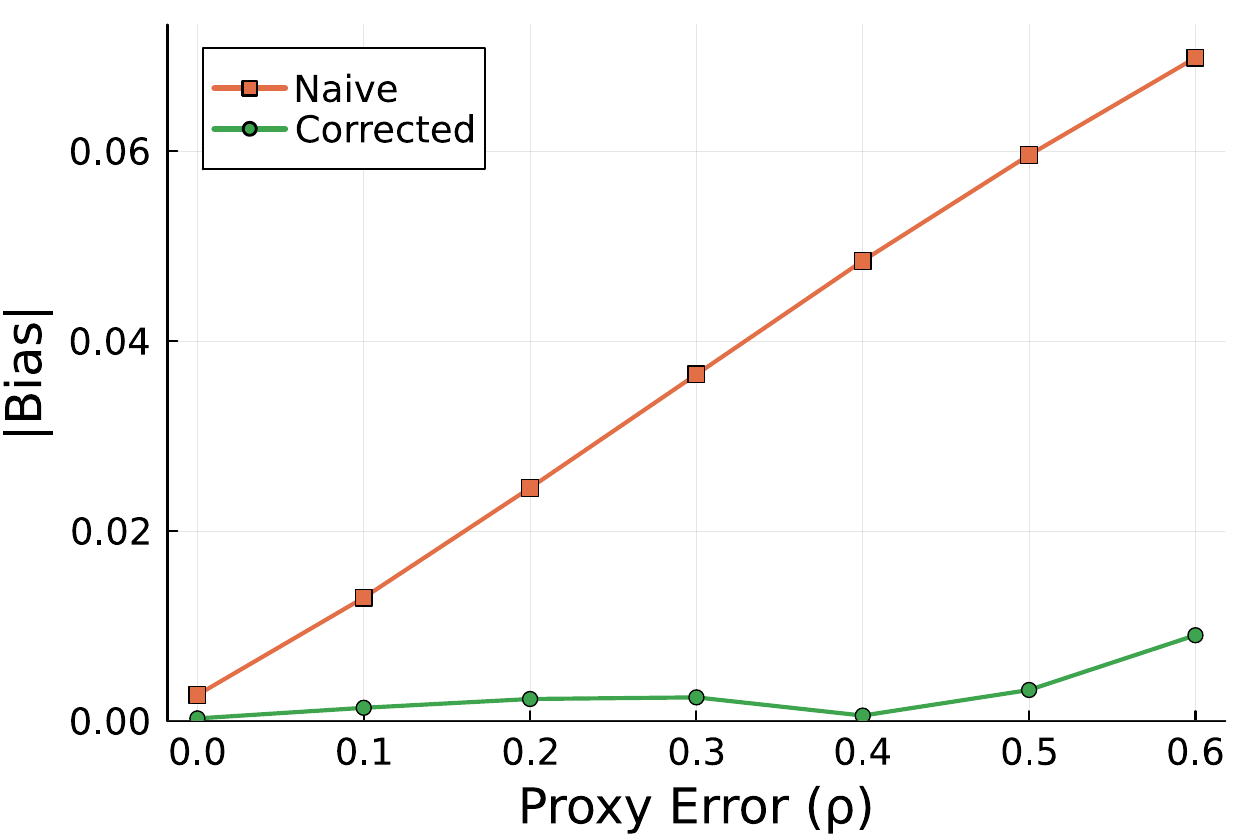}
	\end{subfigure}
	\hfill 
	\begin{subfigure}[t]{0.48\textwidth}
		\centering
		\includegraphics[width=\linewidth]{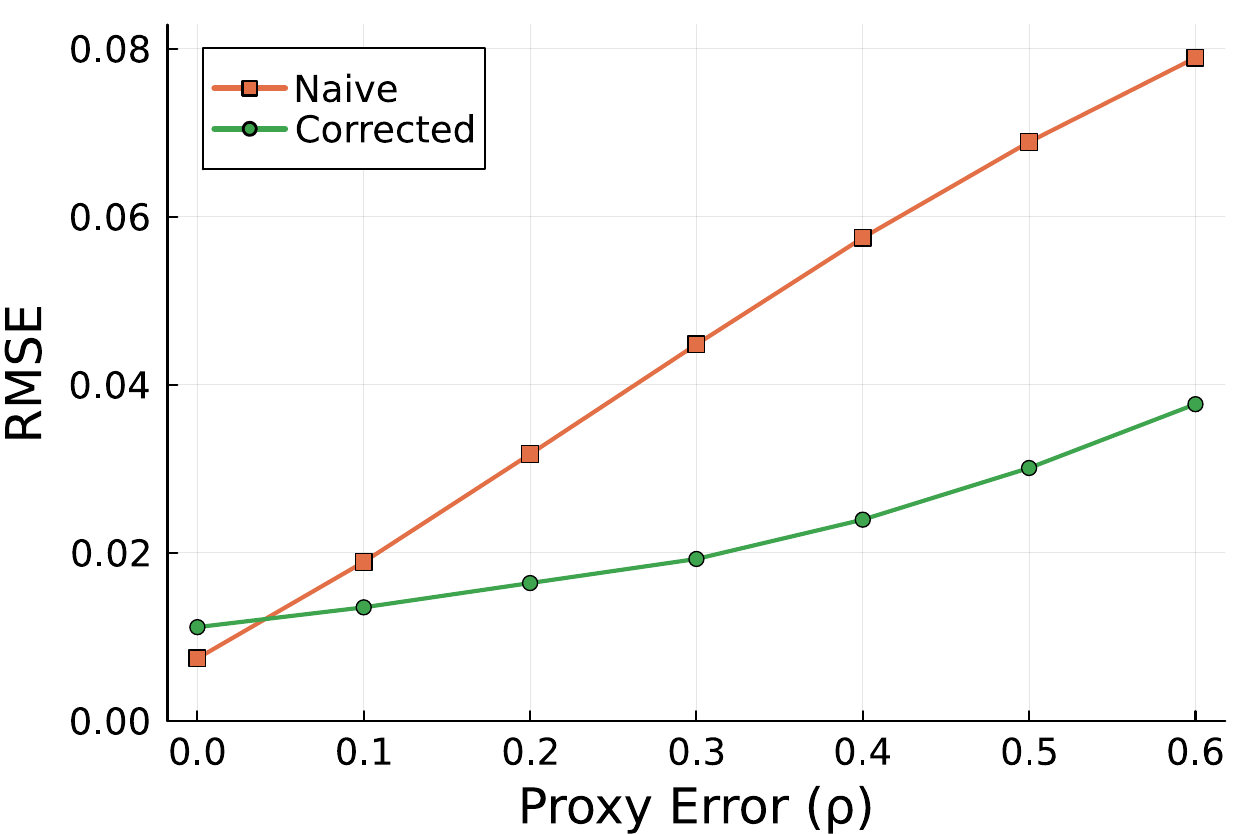}
	\end{subfigure}

	\caption{Bias  and RMSE of the naive and bias-corrected estimators, \citetalias{compiani2025demand} design}
\label{fig:bias RMSE sim}

\vskip 0.5em

\begin{minipage}{\textwidth}
    \small \textit{Note:} Figures show the absolute bias and RMSE across simulations of the naive estimator $\hat \kappa$ in~\eqref{eq:kappa.naive} and bias-corrected estimator $\hat \kappa_{bc}$ in~\eqref{eq:kappa.bc}.
\end{minipage}

\end{figure}

\section{Empirical Application}
\label{sec:application}

We now apply our method to the experimental data from \citetalias{compiani2025demand}. 

\subsection{Setting}

The empirical setting is subsumed by the model in Case 2. The data contain choices of $n = 9{,}265$ participants among $J = 10$ books. Participants performed two choice tasks. In the first task, they were asked to choose their preferred book based on information displayed to them, including (randomized) prices and page rank, standard attributes (author, year of publication, genre and number of pages), and unstructured information (cover images, titles, plot descriptions, and reviews) from Amazon. In the second task, each participant's first choice was removed, and they were asked to choose again from the remaining nine books. Using the first choice data, \citetalias{compiani2025demand} estimate a range of mixed logit models with product fixed effects featuring  proxies computed with different ML methods from the various unstructured data sources. They then evaluate the different models' performance in predicting second choices. This gives a direct measure of how well different models capture counterfactual substitution patterns. The key finding is that unstructured data, particularly book reviews processed with transformer-based text models, can substantially improve predictions of counterfactual substitution.

The results in \citetalias{compiani2025demand} treat the proxies as if they were correctly specified. However, there are good reasons to believe that proxy error might play an important role. First, the unstructured data are processed using pre-trained ML models that are not targeted at predicting substitution patterns and thus may not perfectly capture the underlying attributes that drive consumers' choices.\footnote{Specifically, images are processed via classification models trained to assign each image to one of many classes; text is processed using bag-of-words and transformer models: the former simply capture word frequency, whereas the latter are trained to predict the next word in a text.} Second, the dimension of the raw embeddings is further reduced via PCA before inclusion in the demand model, which may introduce additional proxy error. This PCA step also raises the question of how many principal components should be included in the model.

Here we investigate whether applying our bias correction method and diagnostics  helps better capture substitution patterns. We use the approach from Section \ref{sec: exog prices} since we have individual-level data and prices are randomized, so endogeneity is not a concern. We focus on the ability of different models to correctly predict the closest substitute for any given book. The second choice data give us a direct measure of this: for a given book A, its closest substitute is the book that most people switch to when A is removed from the choice set.

By looking at substitution patterns in response to product removals that are not part of the estimation data, this exercise provides a direct test of the model's ability to predict counterfactuals. Further, unlike in simulations, the demand model might be misspecified even if the proxies perfectly capture the dimensions of differentiation. For instance, the model assumes a normal distribution for the random coefficients but preference heterogeneity may take a different form. As a result, this exercise tests whether our approach works well even in cases where all assumptions needed for the theoretical results might not hold. This sets a high bar for our approach.

\subsection{Results}

Figure \ref{fig:hit rate experiment} shows the results using proxies extracted from book reviews with four ML models.\footnote{We focus on book reviews since \citetalias{compiani2025demand} show that reviews are the most predictive of substitution. Appendix \ref{app:empirical} presents results using proxies extracted from the other data sources.} For each model, we report the fraction of the ten books for which the model correctly identifies the closest substitute (as measured by the second choice data).\footnote{Here, $\kappa$ corresponds to the average (across $p_i$) probability that $B$ is a consumer's second choice conditional on $A$ being their first choice. We compute this across all $(A,B)$ pairs and say that the model predicts product $B$ to be the closest substitute for $A$ if $B$ has the highest second-choice probability when $A$ is the first choice.} The lighter orange bars show the performance of the naive approach based on~\eqref{eq:kappa.naive}, whereas the darker green bars show the performance of the bias-corrected approach. Both use the same $\hat \theta$ from \citetalias{compiani2025demand}. 

The $LM_2$ diagnostic rules out two sets of proxies, indicating that they fail to capture all dimensions of differentiation. Interestingly, these proxies are extracted using first-generation language models: \emph{Count} is a bag-of-words model that represents text as fixed-length vectors based
on simple word counts, whereas \emph{TF-IDF} (Term Frequency-Inverse Document Frequency) is a variation of the bag-of-words approach that assigns more weight to words that are frequent in a specific text but rare in general. By contrast, the proxies that survive the $LM_2$ diagnostic (\emph{ST} for BERT Sentence Transformer, and \emph{USE} for Universal Sentence Encoder) are extracted using more sophisticated models based on the transformer architecture that underpins LLMs. Both \emph{ST} and \emph{USE} have $LM_1$ diagnostics that are below the threshold in the practitioner's guide, with \emph{USE} achieving the smallest value. 

Bias correction substantially improves the model's ability to identify closest substitutes: for both \emph{ST} and \emph{USE}, the fraction of correctly predicted substitutes goes from 40\% without the bias correction to 70\% with the bias correction. For comparison, a coin flip would achieve a hit rate of 11\%.
Overall, these results confirm that our approach is able to guide researchers toward the most informative proxies and meaningfully improve counterfactual predictions.

\begin{figure}[t]
		\centering
		\includegraphics[width=0.7\linewidth]{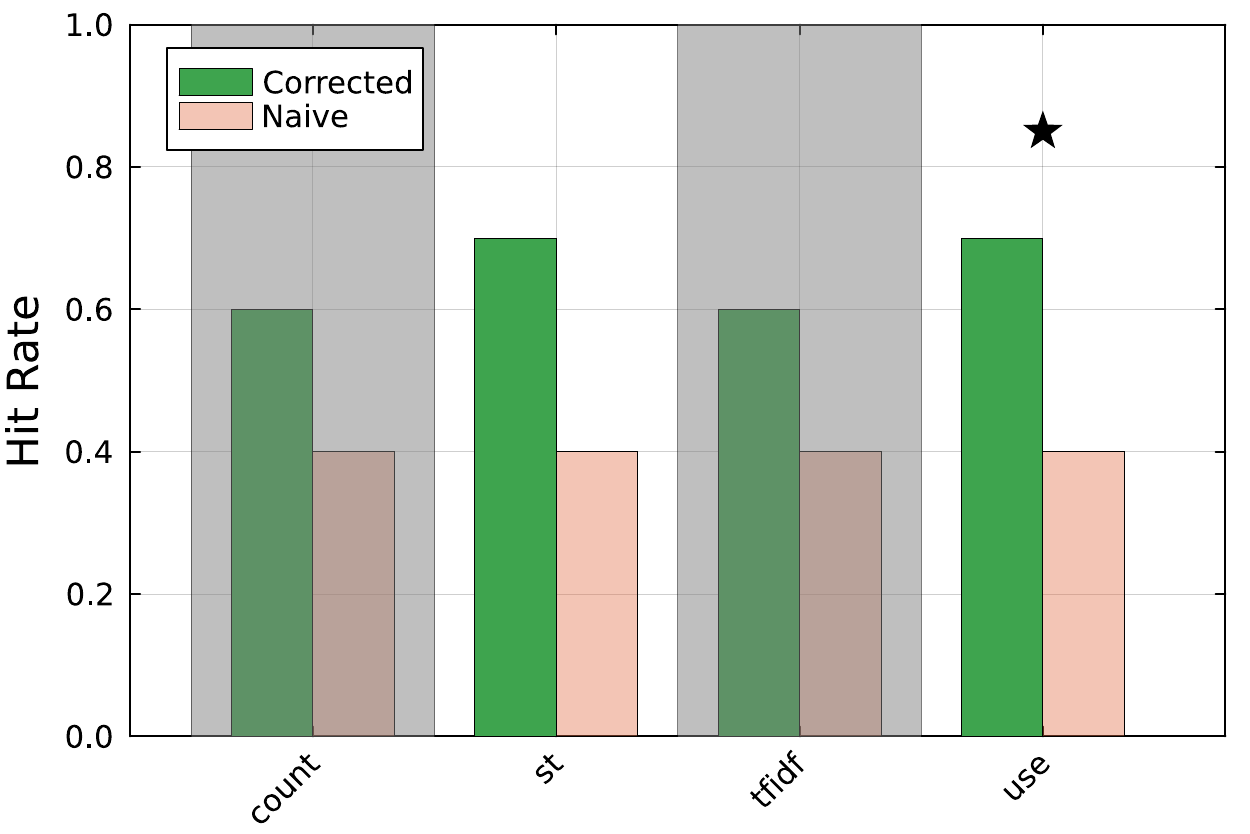}
	\caption{Rates of correct closest-substitute predictions.}
	\label{fig:hit rate experiment}

    \vskip 0.5em
    
    \parbox[]{\textwidth}{\small \emph{Notes: }Darker green bars show the fraction of books for which the bias-corrected estimator correctly identifies the closest substitute. Lighter orange bars show the corresponding fractions for the naive estimator. Each pair of bars corresponds to an ML model used to extract the proxies from the book reviews. The specifications ruled out by the $LM_2$ diagnostic are grayed out. A star flags the specification with the smallest  $LM_1$ diagnostic.}
\end{figure}

\begin{figure}[t]
\centering
 \scalebox{0.8}{%
\begin{tikzpicture}[
    rowlabel/.style={
        align=center
    },
    booktitle/.style={
        font=\rmfamily\scriptsize,
        align=center,
        text width=3.1cm
    },
    orangefill/.style={
    fill=JuliaOrange!18
    },
    greenfill/.style={
    fill=JuliaGreen!18
    }
]

\def\labelW{3.2}
\def\cellW{3.6}
\def\rowH{3.25}

\def\nCols{3}
\def\nRows{4}

\pgfmathsetmacro{\gridW}{\nCols*\cellW}
\pgfmathsetmacro{\totalW}{\labelW+\gridW}
\pgfmathsetmacro{\totalH}{\nRows*\rowH}

\node[rowlabel] at ({\labelW/2},{-\rowH/2})
    {Focal Book};

\node[rowlabel] at ({\labelW/2},{-1.5*\rowH})
    {Empirical};

\node[rowlabel] at ({\labelW/2},{-2.5*\rowH})
    {Naive\\(Review USE)};

\node[rowlabel] at ({\labelW/2},{-3.5*\rowH})
    {Corrected\\(Review USE)};

%
\newcommand{\BookCell}[4]{%
    \node at
        ({\labelW + (#1+0.5)*\cellW},
         {-(#2)*\rowH - 0.95})
        {\includegraphics[height=1.45cm]{#3}};

    \node[booktitle] at
        ({\labelW + (#1+0.5)*\cellW},
         {-(#2)*\rowH - 2.4})
        {#4};
}

\foreach \c in {0,1,2} {
    \fill[orangefill]
        ({\labelW+\c*\cellW},{-2*\rowH})
        rectangle
        ({\labelW+(\c+1)*\cellW},{-3*\rowH});
}

\foreach \c in {0,1,2} {
    \fill[greenfill]
        ({\labelW+\c*\cellW},{-3*\rowH})
        rectangle
        ({\labelW+(\c+1)*\cellW},{-4*\rowH});
}

\draw[line width=0.5pt]
    (\labelW,0) -- (\labelW,-\totalH);

\draw[line width=0.5pt]
    (0,-\rowH) -- (\totalW,-\rowH);

\BookCell{0}{0}{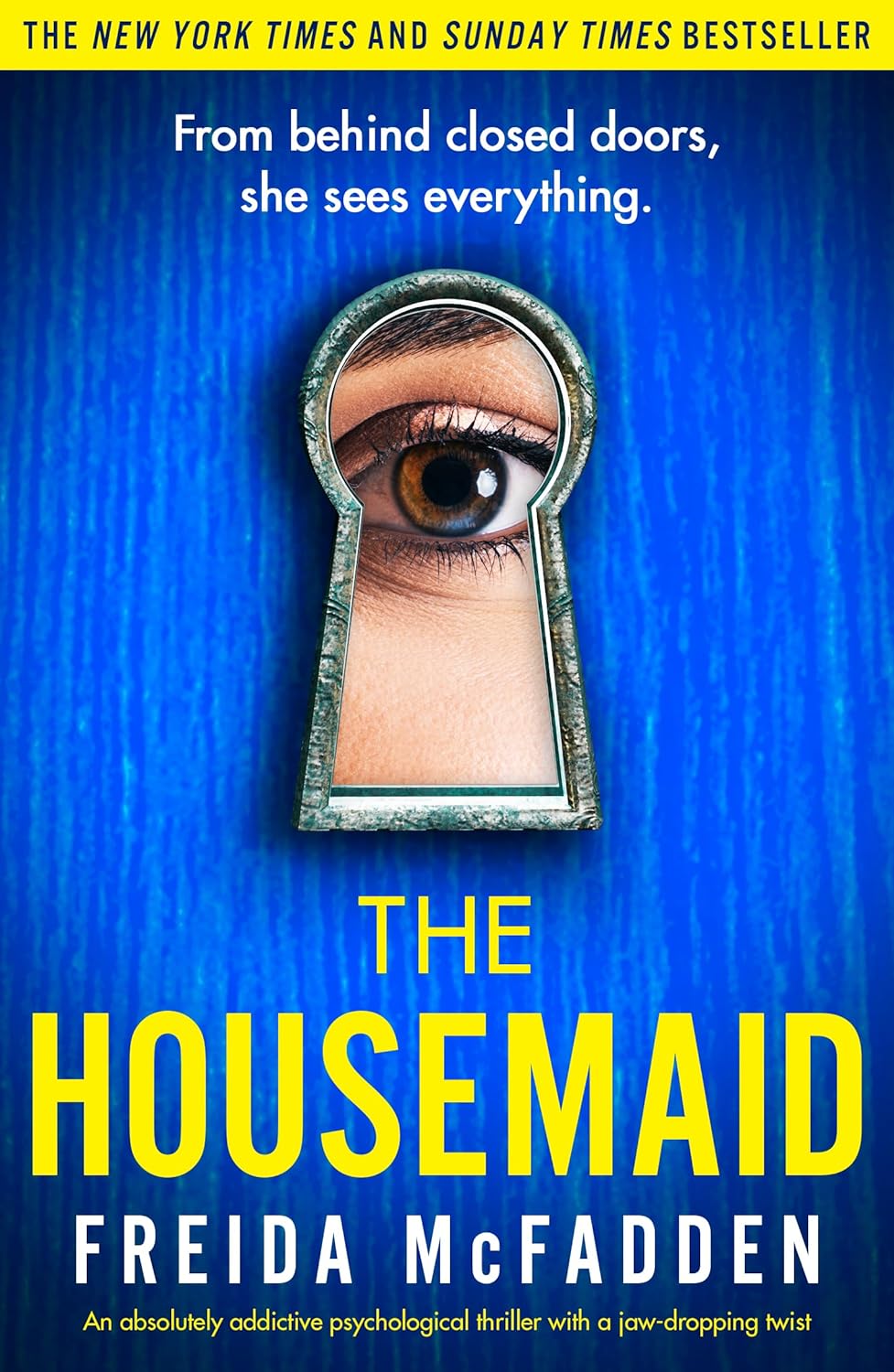}
    {The Housemaid}

\BookCell{1}{0}{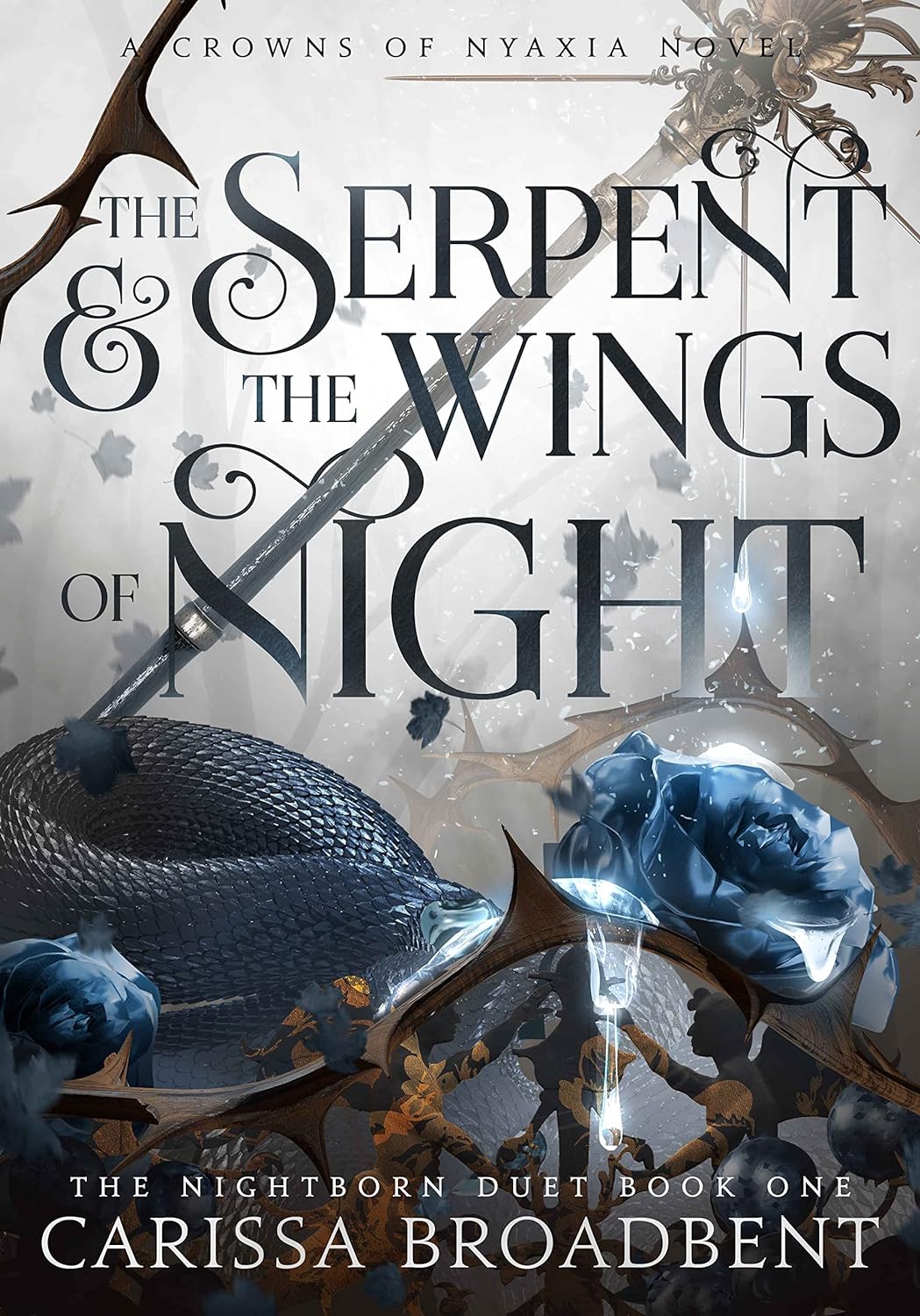}
    {The Serpent \&\\The Wings of\\Night}

\BookCell{2}{0}{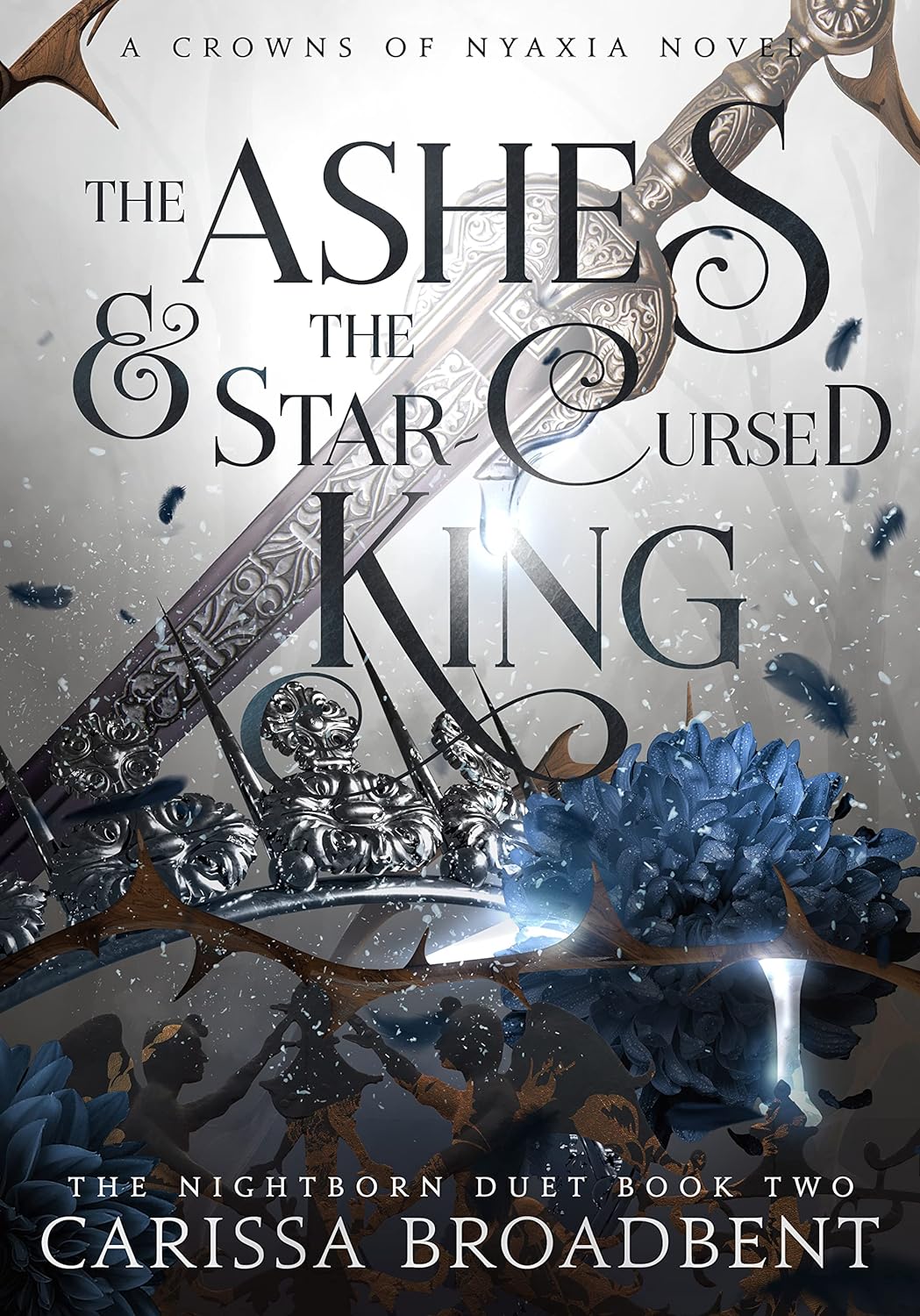}
    {The Ashes \&\\The Star\\Cursed King}

\BookCell{0}{1}{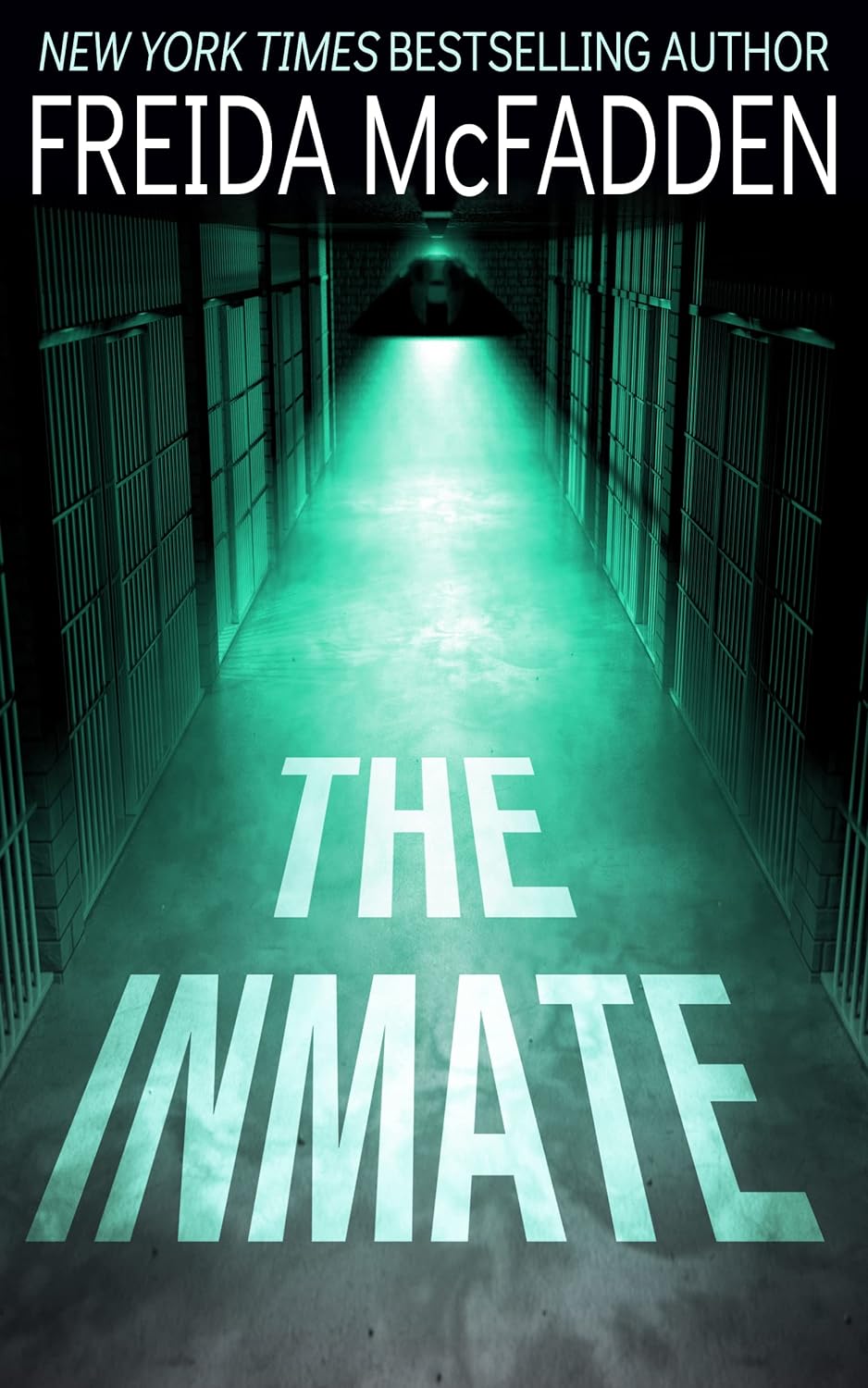}
    {The Inmate}

\BookCell{1}{1}{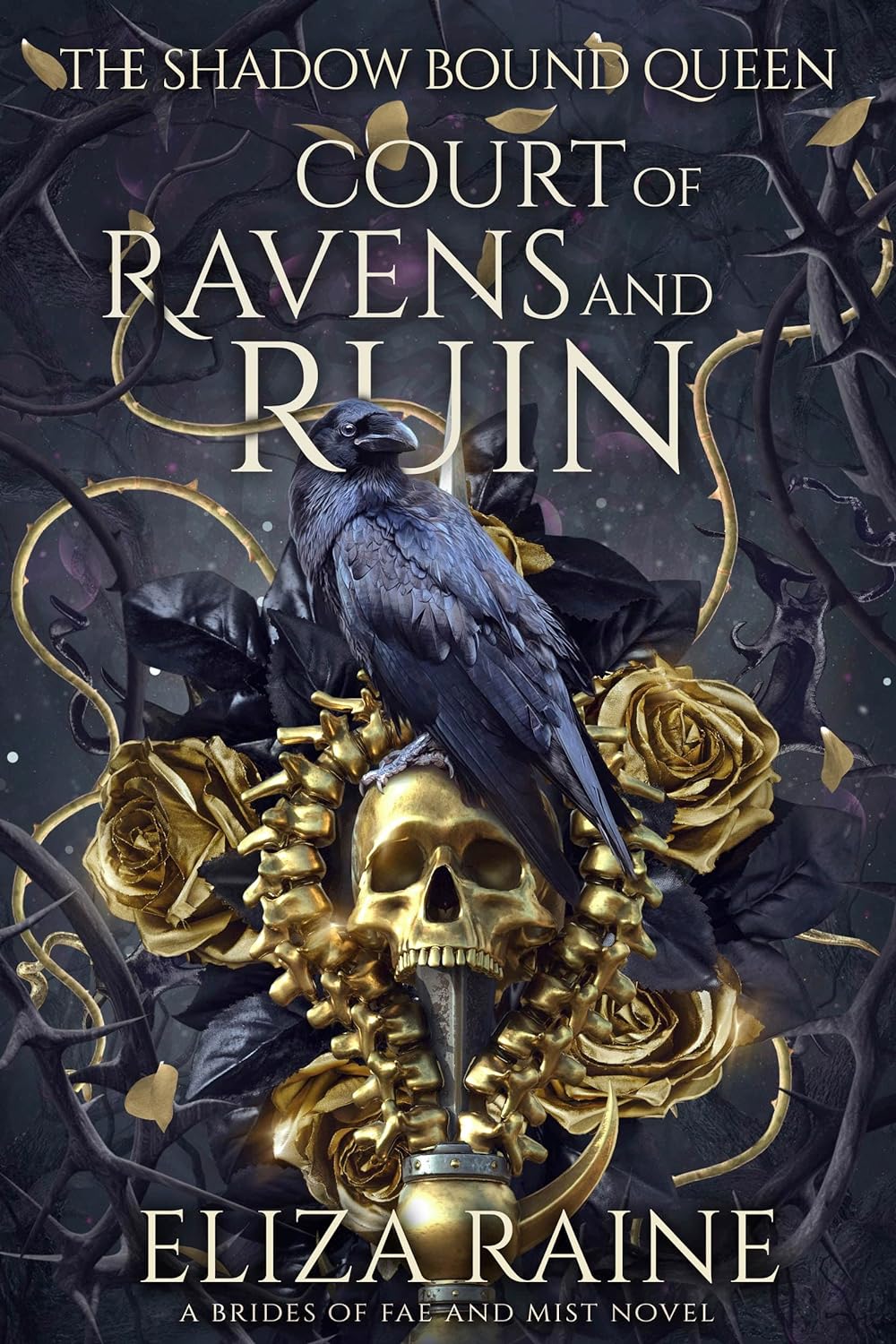}
    {Court of\\Ravens and\\Ruin}

\BookCell{2}{1}{B09WRJJKXC.jpg}
    {The Serpent \&\\The Wings of\\Night}

\BookCell{0}{2}{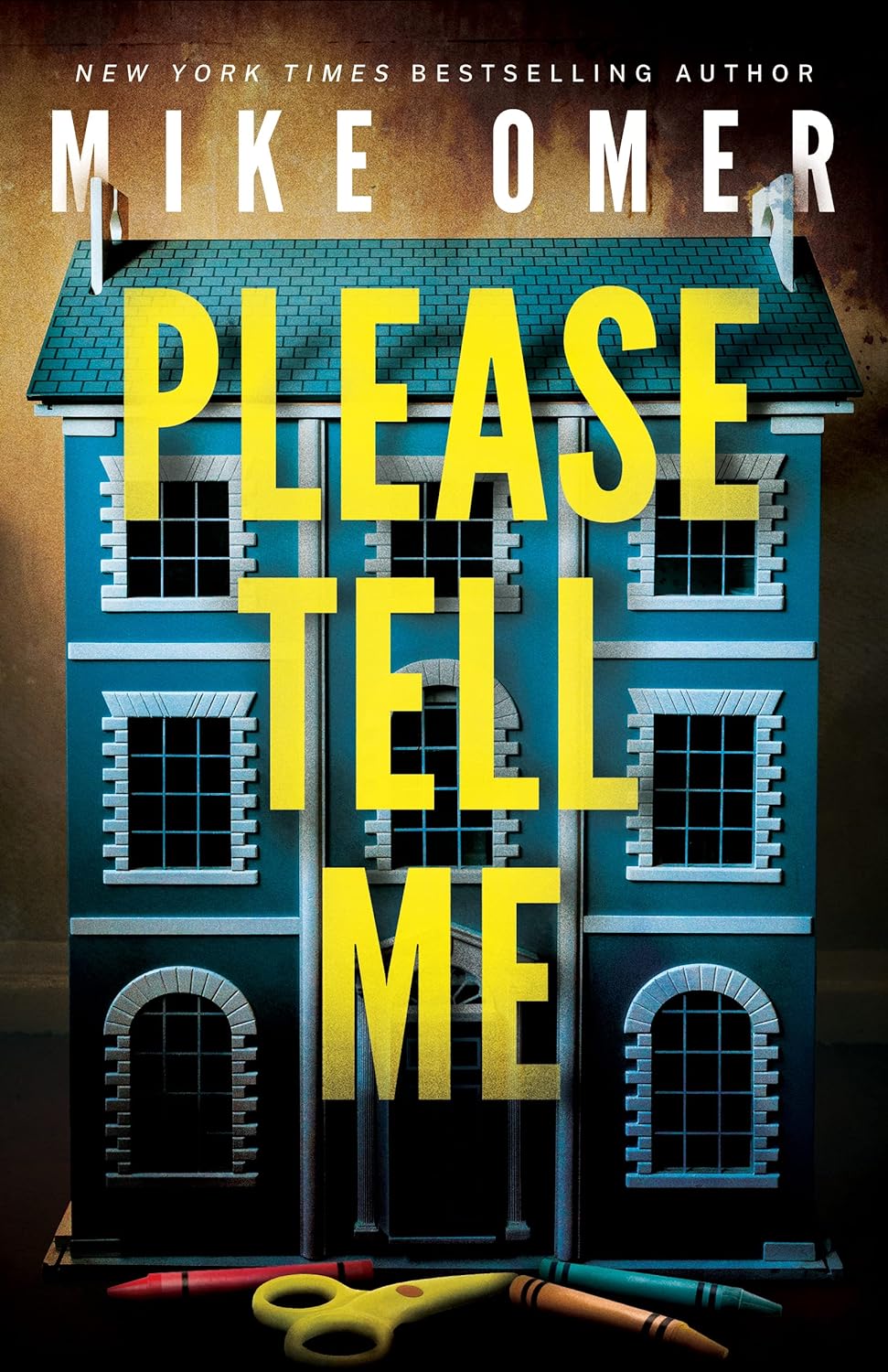}
    {Please Tell\\Me}

\BookCell{1}{2}{B0BWNW9G8X.jpg}
    {Please Tell\\Me}

\BookCell{2}{2}{B0BWNW9G8X.jpg}
    {Please Tell\\Me}

\BookCell{0}{3}{B0B1XQTZ1R.jpg}
    {The Inmate}

\BookCell{1}{3}{B0BFFJRB2S.jpg}
    {Court of\\Ravens and\\Ruin}

\BookCell{2}{3}{B09WRJJKXC.jpg}
    {The Serpent \&\\The Wings of\\Night}

\end{tikzpicture}
}
\vskip 0.5em
\caption{Closest substitute predictions for the three books where the naive and corrected models disagree.}
\label{fig:disagree}
	\parbox[]{\textwidth}{\small \emph{Notes:} The three columns correspond to the focal books for which the naive and bias-corrected estimators (based on the \emph{Review USE} specification) make different closest-substitute predictions. The row labeled ``Empirical'' shows the closest substitute to the focal book revealed by the second-choice data.}
\end{figure}

\subsection{Why Does the Correction Help?}
\label{sec:iia}

To shed more light on the mechanism underlying the improved performance of the corrected estimator, we take the preferred \emph{USE} proxies and focus on the three books for which the bias correction identifies the closest substitute whereas the naive estimator does not. For each of these cases, Figure \ref{fig:disagree} shows the closest substitutes predicted by the naive and corrected estimators, along with the ground truth revealed by the second-choice data.

The naive estimator incorrectly predicts the same book, \emph{Please Tell Me}, as the closest substitute to all three focal books. Importantly, \emph{Please Tell Me} is the most popular book in the first-choice data. This is a typical Independence of Irrelevant Alternatives (IIA) pattern: the model mechanically predicts substitution to the remaining books in proportion to their market shares, irrespective of their similarity to the focal book \citep{mcfadden1974conditional}. Why does this happen? As we saw in the simulations, when the proxies are noisy measures of the true dimensions of differentiation, the mixed logit model will tend to discard the proxies\footnote{Because the model includes product fixed effects, the only channel through which proxies can affect choice is via the covariance of the random component of utilities. This channel is shut down when the variances of the random coefficients on the proxies collapse to zero.} and revert toward the IIA-like substitution patterns of a logit model without random coefficients.
By contrast, the bias correction breaks this pattern: it redirects substitution toward more similar books and correctly identifies the closest substitute.

This illustration offers what we think is a valuable message for practitioners: using imperfect proxies can defeat the very purpose of estimating a mixed logit model by pushing its predictions toward IIA. Correcting the bias induced by proxy error alleviates this and, in both this application and our simulations, leads to a substantial improvement in the model's counterfactual predictions.

\section{Theory}
\label{sec:theory}

In this section, we present the main theoretical results. All proofs are in Appendix~\ref{app:proofs}.

\subsection{Case 1: Endogenous Prices}
\label{sec:theory blp}

As is standard, we treat the aggregate variables $(p_t,s_t,\bar x_t, \xi_t,z_t)$ as independent across markets, and the micro-level variables $(d_{it},y_{it},\bar y_{it})$ as independent within markets conditional on market-level variables. Let $\Gamma \subseteq \mathbb R^{\dim(\gamma)}$ denote the parameter space for $\gamma$. We implicitly assume that $\Gamma$ is convex and open. This is true for the $\gamma(\theta, e)$ given in Examples~\ref{ex:blp} and \ref{ex:micro blp}, for which $\Gamma$ is the product of copies of $\mathbb R$ and $(0,\infty)$ and is therefore convex and open.\footnote{\label{fn:cholesky} For instance, the operation $l$ stacking the lower-triangular entries of the Cholesky factor maps the manifold of symmetric positive definite matrices into the product of copies of $\mathbb R$ and $(0,\infty)$. A similar result holds for the reduced-rank Cholesky decomposition $l_r$ \citep{neuman2023restricted}. Whenever $e$ has full column rank $r$, one can always relabel the indices so that the first $r$ rows of $e$ are linearly independent, so that $l_r(e \Sigma_e e')$ is well defined. If this is not the case, $e$ contains collinear embeddings and the covariance matrix of utilities can be represented with a lower $r$.} In Section~\ref{sec:blp}, the counterfactual function $k_t$, $\hat \xi_t$ from (\ref{eq:xi}), and micro moments $m_t$ depended on $(\theta, e)$ only through the value of $\gamma(\theta, e)$. We can therefore view $k_t$, $\hat \xi_t$, and $m_t$ as random functions defined on $\Gamma$.

\subsubsection{Theory for Bias Correction}

We first give results for the case of combined market-level data and microdata. We let $\chi_t = (p_t, \xi_t, \bar x_t, e)$ and let $\mathcal M$ denote the $\sigma$-algebra generated by $\chi_1,\ldots,\chi_\tau$. Let $V = \mathrm{Var} ( Z_t \hat \xi_t(\gamma_0) )$ be the $\dim(z) \times \dim(z)$ covariance matrix of the aggregate moments at the true parameters, and let $V_t = \mathrm{Var} \left( m_{it} | \chi_t \right)$ be the $\dim(m) \times \dim(m)$ covariance matrix of the micro moments for market $t$ conditional on $\chi_t$. Let 
\begin{equation} \label{eq:H def}
 H = G'V^{-1}G + \sum_{t=1}^\tau M_t' (r_t V_t)^{-1} M_t,
\end{equation}
where $r_1,\ldots,r_\tau$ are defined in Assumption~\ref{a.2} below, 
and let $h = \E[  \dot k_t(\gamma_0) ] - G' V^{-1} K$, where $K = \E [ k_t(\gamma_0) Z_t \hat \xi_t(\gamma_0) ]$ is $\dim(z) \times 1$, $G = \E[Z_t \dot \xi_t(\gamma_0)]$ is $\dim(z) \times \dim(\gamma)$, and $M_t = \dot m_t(\gamma_0)$ is $\dim(m) \times \dim(\gamma)$. Here $V$ and $h$ are deterministic whereas $V_1, \ldots, V_\tau$ and $H$ are $\mathcal M$-measurable random matrices. Let $N$ be a neighborhood of $\gamma_0$. 

\medskip

\begin{assumption}\label{a.2}
	Let the following hold:
	\begin{enumerate}[label={(\roman*)}, nosep]
		\item \label{a.2.1} $k_t(\cdot)$, $m_t(\cdot)$, and $Z_t\hat \xi_t(\cdot)$ are twice continuously differentiable in $\gamma$ on $N$ (almost surely), and elements of the functions and their first and second derivatives are uniformly (for $\gamma \in N$) bounded by a random variable $D_t$ with finite fourth moment;
		\item \label{a.2.2} $\E[\|m_{it}\|^{2+\delta}|\mathcal M] \leq C$ (almost surely) for some $0 < \delta, C < \infty$;
		\item \label{a.2.3} $V$ is positive definite and $\lambda_{\min}(V_1),\ldots,\lambda_{\min}(V_\tau),\lambda_{\min}(H) \geq \epsilon$ (almost surely) for some $\epsilon > 0$;
		\item \label{a.2.4} $T/N_t \to r_t \in (0,\infty)$ for $1 \leq t \leq \tau$.
	\end{enumerate}
\end{assumption}

\bigskip

Parts~\ref{a.2.1}-\ref{a.2.3} of Assumption~\ref{a.2} are standard smoothness, moment, and rank conditions, respectively. Assumption~\ref{a.2}\ref{a.2.4} requires that the aggregate data and microdata have comparable sample sizes. This is designed to accommodate common empirical scenarios where $T$ is in the high tens or hundreds and $N_t$ is in the hundreds or thousands for each market (e.g., \cite{petrin2002quantifying} and \cite{grieco2024evolution}).

The next result shows that the bias-corrected estimator $\hat \kappa_{bc}$ from~\eqref{eq: kappa_bc blp linear} is asymptotically centered at the true counterfactual $\kappa_0 = \E[k_t(\gamma_0)]$ and its asymptotic variance is independent of $\hat \gamma$, $\hat \theta$, and $\tilde e$. Because the effect of market-level variables in markets for which there is microdata persists in the limit, we use the notion of stable convergence. We say a sequence of random variables $Z_T$ converges in distribution to $Z$ ($\mathcal M$-stably) if $\lim_{T \to \infty} \Pr( Z_T \leq z, A) = \Pr(Z \leq z, A)$ for all continuity points $z$ of the distribution of $Z$ and all $\mathcal M$-measurable events $A$. Convergence in distribution is a special case corresponding to replacing $\mathcal M$ with the trivial $\sigma$-algebra $\{\emptyset, \Omega\}$.

\medskip

\begin{proposition}\label{prop:kappa.endogenous}
	Let Assumption~\ref{a.2} hold and $\hat \gamma = \gamma_0 + o_p(T^{-1/4})$. Then $\sqrt T(\hat \kappa_{bc} - \kappa_0)$ converges in distribution ($\mathcal M$-stably) as $T \to \infty$ to a mixed Gaussian random variable with mean zero and $\mathcal M$-measurable variance
	\begin{equation} \label{eq:var.endogenous}
		V_{bc} = \mathrm{Var}\left(k_t(\gamma_0) \right) + h'H^{-1} h - K' V^{-1} K.
	\end{equation} 
\end{proposition}

\medskip

The (random) asymptotic variance $V_{bc}$ can easily be estimated using $\hat V_{bc}$  in equation~(\ref{eq:standard errors blp}).  Standard errors $\mathrm{s.e.}(\hat \kappa_{bc}) = \sqrt{\hat V_{bc}/T}$ are consistent under the conditions of Proposition~\ref{prop:kappa.endogenous}. Inference can be performed based on $t$-statistics $(\hat \kappa_{bc} - \kappa_0)/\mathrm{s.e.}(\hat \kappa_{bc})$ using the standard $N(0,1)$ critical values.

\medskip

\begin{remark}\label{rmk:vicinity}
Proposition~\ref{prop:kappa.endogenous} requires $\hat \gamma = \gamma_0 + o_p(T^{-1/4})$, which is a completely standard condition in two-step estimation \cite[e.g.,][Assumption~5.1(ii)]{newey1994asymptotic}. Given the nonlinear nature of multi-product demand models, a neighborhood condition on $\hat \gamma$ seems unavoidable. Moreover, as discussed below, our $LM_1$ diagnostic can be used to validate this condition in empirical settings. 

It is also important to emphasize that this condition holds under the implicit assumptions of the usual workflow which treats $\tilde e$ as the true $e$. In that case, standard regularity conditions imply that $\hat \theta$ is $\sqrt{T}$-consistent and $\theta \mapsto \gamma(\theta, \tilde e)$ is continuously differentiable at $\theta_0$. This, in turn, implies that $\hat \gamma = \gamma_0 + O_p(T^{-1/2})$, a faster rate on the remainder than we require. Thus, our approach gains robustness to proxy error and its theoretical guarantees hold under weaker conditions than are already assumed by the usual workflow.
\end{remark}

We next provide a sense in which $\hat \kappa_{bc}$ is efficient. The proof of Proposition~\ref{prop:kappa.endogenous} shows that $\hat \kappa_{bc}$ belongs to the class $\mathcal K$ of estimators $\hat \kappa$ of $\kappa_0$ that satisfy
\begin{multline} \label{eq:kappa linearize}
	\sqrt T(\hat \kappa - \kappa_0) = \frac{1}{\sqrt T} \sum_{t=1}^T \left( k_t(\gamma_0) - \kappa_0 - c' Z_t \hat \xi_t(\gamma_0) \right) \\ + \sqrt T \sum_{t = 1}^\tau d_t' \left( \bar m_t - m_t(\gamma_0) \right) + o_p(1),
\end{multline}
where $c$ and $d_1,\ldots,d_\tau$ are any $\mathcal M$-measurable random vectors that satisfy 
\begin{equation}\label{eq:orthog.blp.micro}
	\E[  \dot k_t(\gamma_0) ] - G' c - \sum_{t=1}^\tau M_t' d_t = 0,
\end{equation}
(almost surely). Similar arguments to the proof of Proposition \ref{prop:kappa.endogenous} show that any $\hat \kappa$ obtained by plugging-in any  $\hat \gamma = \gamma_0 + o_p(T^{-1/4})$ into (\ref{eq: kappa_bc blp linear}) for some arbitrary weights $\hat c$ and $\hat d_1,\ldots,\hat d_\tau$ converging to $c$ and $d_1,\ldots,d_\tau$ belongs to this class. Condition \eqref{eq:orthog.blp.micro} typically ensures that such an estimator satisfies \eqref{eq:kappa linearize} uniformly for $\gamma$ local to $\gamma_0$. In general, there are many different weights $\hat c$ and $\hat d_1,\ldots,\hat d_\tau$ whose probability limits $c$ and $d_1,\ldots,d_\tau$ will correspond to different (random) asymptotic variances. The following result shows that $\hat \kappa_{bc}$ has the smallest asymptotic variance among this class of estimators of $\kappa_0$.

\medskip

\begin{proposition}\label{prop:efficiency}
Let Assumption~\ref{a.2} hold and $\hat \gamma = \gamma_0 + o_p(T^{-1/4})$. Then $\hat \kappa_{bc}$ has the smallest asymptotic variance among the class of estimators $\mathcal K$ of $\kappa_0$.
\end{proposition}

\medskip

An important practical take-away from Proposition~\ref{prop:efficiency} is that microdata should be used, if available, to (weakly) improve the efficiency of estimators of counterfactuals.

Of course, in many scenarios microdata may not be available. Here we state a simpler version of Proposition~\ref{prop:kappa.endogenous} tailored to this case. Recall that the bias-corrected estimator in this case is given in \eqref{eq: kappa_bc blp linear no micro}, where $\hat c$ is given in \eqref{eq: c^hat blp} with $\hat H = \hat G' \hat V^{-1} \hat G$. To introduce the assumptions, let $V$ and $G$ be as above, and let $H = G' V^{-1} G$.

\medskip

\begin{assumption}\label{a.3}
	Let the following hold:
	\begin{enumerate}[label={(\roman*)}, nosep]
		\item \label{a.3.1} $k_t(\cdot)$ and $Z_t \hat \xi_t(\cdot)$ are twice continuously differentiable in $\gamma$ on $N$ (almost surely), and the functions and all elements of their first and second derivatives are uniformly (for $\gamma \in N$) bounded by a random variable $D_t$ with finite fourth moment;
		\item \label{a.3.2} $V$ and $H$ are positive definite.
	\end{enumerate}
\end{assumption}

\bigskip

The next result is a special case of Proposition~\ref{prop:kappa.endogenous} and is stated without proof.

\medskip

\begin{proposition}\label{prop:kappa.endogenous.simple}
	Let Assumption~\ref{a.3} hold and $\hat \gamma = \gamma_0 + o_p(T^{-1/4})$. Then, 
	\[
	\sqrt T(\hat \kappa_{bc} - \kappa_0) \to_d N(0,V_{bc})
	\]
	with $V_{bc}$ as in \eqref{eq:var.endogenous} with $H = G'V^{-1}G$.
\end{proposition}

\medskip

The asymptotic variance can be easily estimated using the formula $\hat V_{bc}$ in \eqref{eq:standard errors blp no micro data}. Standard errors are then computed as $\sqrt{\hat V_{bc}/T}$. These are consistent under the conditions of Proposition~\ref{prop:kappa.endogenous.simple}.

\subsubsection{Theory for $LM_1$}
\label{sec:goodness of fit blp.theory}

We now show that the diagnostic $LM_1$ in \eqref{eq: LM1 blp} behaves like $\|\sqrt T(\hat \gamma - \gamma_0)\|^2$ as the sample size grows large. 
In what follows, we abbreviate ``with probability approaching one'' to ``wpa1.'' Recall $H$ from \eqref{eq:H def} and let $\lambda_{\min}(H)$ denote its smallest singular value, which is uniformly bounded away from zero by Assumption~\ref{a.2}\ref{a.2.3}.

\medskip

\begin{proposition}\label{prop:lm.fit.gamma.2}
	Let Assumption~\ref{a.2} hold and let $\hat \gamma = \gamma_0 + o_p(1)$. Fix any sequence $C_T \uparrow \infty$ and any $\epsilon > 0$. Then wpa1, we have 
	\[
	\frac{1 + \epsilon}{1 + 2\epsilon} \left( \sqrt{LM_1} - \epsilon C_T \right) \leq \| H^{1/2}(\sqrt T (\hat \gamma - \gamma_0))\| \leq \frac{1 + 2\epsilon}{1 + \epsilon} \left( \sqrt{LM_1} + \epsilon C_T \right).
	\]
	In particular, wpa1 we have that $LM_1 \leq C_T^2$ implies
	\[
	\|\sqrt T (\hat \gamma - \gamma_0)\| \leq \frac{(1+2\epsilon) C_T}{\sqrt{\lambda_{\min}(H)}}.
	\]
	Moreover, if $\hat \gamma = \gamma_0 + o_p(C_T/\sqrt T)$, then wpa1 we have
	\[
	LM_1 \leq (1+\epsilon)^2 C_T^2.
	\]
\end{proposition}

Proposition~\ref{prop:lm.fit.gamma.2} shows $LM_1$ behaves like $\|\sqrt T (\hat \gamma - \gamma_0)\|^2$. With $C_T = o(T^{1/4})$,  wpa1 we have  that $LM_1 \leq C_T^2$ implies $\|\hat \gamma - \gamma_0\| \leq \mathrm{constant} \times C_T/\sqrt T = o(T^{-1/4})$.

The proof of Proposition~\ref{prop:lm.fit.gamma.2} shows that the ``wpa1'' qualifier depends on whether a $\chi^2_{\dim(\gamma)}$ random variable is less than $\epsilon^2 C_T^2$. With $\epsilon = 1$, say, this suggests taking $C_T^2$ to be at least as large as the 95th or 99th percentile of the $\chi^2_{\dim(\gamma)}$ distribution. To check a convergence rate of $\sqrt{(\log T)/T}$, for instance, one could use $C_T^2 = \chi^2_{\dim(\gamma), 0.95} \log T$.

\subsection{Case 2: Individual-Level Price Variation}
\label{sec:theory model 2}

As is standard, we treat the observations $(d_i,p_i,y_i)$ as independent across individuals. We again let $\Gamma \subseteq \mathbb R^{\dim(\gamma)}$ denote the parameter space for $\gamma$, which we shall take to be convex and open. This is true for the $\gamma(\theta, e)$ in Example~\ref{ex:mixed logit}, for which $\Gamma$ is the product of copies of $\mathbb R$ and $(0,\infty)$ and is therefore convex and open (see footnote~\ref{fn:cholesky}). In Section~\ref{sec: exog prices}, the counterfactual function $k_i$ and choice probabilities $\sigma_i$ depended on $(\theta, e)$ only through the value of $\gamma(\theta, e)$. We therefore treat $k_i$ and $\sigma_i$ as random functions on $\Gamma$.

\subsubsection{Theory for Bias Correction}
		
We now derive the theoretical properties of the bias-corrected estimator $\hat \kappa_{bc}$ from (\ref{eq:kappa.bc}). We first outline some standard smoothness, moment, and rank assumptions. In what follows, we use a dot (as above) and double dot to denote first and second derivatives with respect to $\gamma$. Let $H = \E [ \dot \sigma_i(\gamma_0)' V_i(\gamma_0)^{-1} \dot \sigma_i(\gamma_0) ]$ and $N$ be a neighborhood of $\gamma_0$.

\medskip

\begin{assumption}\label{a.1}
    Let the following hold:
    \begin{enumerate}[label={(\roman*)}, nosep]
        \item \label{a.1.1} $k_i(\cdot)$ and $\sigma_i(\cdot)$ are twice continuously differentiable in $\gamma$ on $N$ (almost surely), and all elements of the functions and their first and second derivatives are uniformly (for $\gamma \in N$) bounded by a random variable $D_i$ with finite second moment;
        \item \label{a.1.2} $\dot \sigma_i(\cdot)' V_i(\cdot)^{-1}$ is continuously differentiable (almost surely) and all elements of the function and its derivative are uniformly (for $\gamma \in N$) bounded by a random variable $D_i$ with finite higher-than-second moment;
        \item \label{a.1.3} $H$ is positive definite.
    \end{enumerate}
\end{assumption}
		
\medskip

Let $\kappa_0 = \E[k_i(\gamma_0)]$ denote the true value of the counterfactual. The following result shows that the asymptotic distribution of $\hat \kappa_{bc}$ is centered at $\kappa_0$ and its variance is independent of $\hat \gamma$, $\hat \theta$, and $\tilde e$.

\medskip

\begin{proposition}\label{prop:kappa.exogenous}
    Let Assumption~\ref{a.1} hold and $\hat \gamma = \gamma_0 + o_p(n^{-1/4})$. Then
    \[
    \sqrt n(\hat \kappa_{bc} - \kappa_0) \to_d N(0,V_{bc}),
    \]
    as $n \to \infty$, where 
    \begin{equation} \label{eq:var}
        V_{bc} = \mathrm{Var}\big(k_i(\gamma_0) \big) + \E [ \dot k_i(\gamma_0) ]' H^{-1} \E [ \dot k_i(\gamma_0) ].
    \end{equation} 
\end{proposition}
	
\medskip

The asymptotic variance $V_{bc}$ can be easily estimated using $\hat V_{bc}$ in~\eqref{eq:standard errors}. Standard errors are then computed as $\sqrt{\hat V_{bc}/n}$. These are consistent under the conditions of Proposition~\ref{prop:kappa.exogenous}. We note that, as before, Proposition~\ref{prop:kappa.exogenous} requires that $\hat \gamma$ be in a vicinity of $\gamma_0$, and refer the reader to Remark~\ref{rmk:vicinity} for a discussion.

\subsubsection{Theory for $LM_1$}
\label{sec:goodness.of.fit.theory}

The following result is analogous to Proposition~\ref{prop:lm.fit.gamma.2} and shows $LM_1$ behaves like $\|\sqrt n (\hat \gamma - \gamma_0)\|^2$.

\medskip

\begin{proposition}\label{prop:lm.fit.gamma}
    Let Assumption~\ref{a.1} hold and let $\hat \gamma = \gamma_0 + o_p(1)$. Fix any sequence $C_n \uparrow \infty$ and any $\epsilon > 0$. Then wpa1, we have 
    \[
    \frac{1 + \epsilon}{1 + 2\epsilon} \left( \sqrt{LM_1} - \epsilon C_n \right) \leq \| H^{1/2}(\sqrt n (\hat \gamma - \gamma_0))\| \leq \frac{1 + 2\epsilon}{1 + \epsilon} \left( \sqrt{LM_1} + \epsilon C_n \right).
    \]
    In particular, wpa1 we have that $LM_1 \leq C_n^2$ implies
    \[
    \|\sqrt n (\hat \gamma - \gamma_0)\| \leq \frac{(1+2\epsilon) C_n}{\sqrt{\lambda_{\min}(H)}}.
    \]
    Moreover, if $\hat \gamma = \gamma_0 + o_p(C_n/\sqrt n)$, then wpa1 we have
    \[
    LM_1 \leq (1+\epsilon)^2 C_n^2.
    \]
\end{proposition}

The implications of Proposition~\ref{prop:lm.fit.gamma} are similar to before. For $C_n = o(n^{1/4})$, we have wpa1 that $LM_1 \leq C_n^2$ implies $\|\hat \gamma - \gamma_0\| \leq \mathrm{constant} \times C_n/\sqrt n = o(n^{-1/4})$. The proof of Proposition~\ref{prop:lm.fit.gamma} shows that the ``wpa1'' qualifier depends on whether a $\chi^2_{\dim(\gamma)}$ random variable is less than $\epsilon^2 C_n^2$. To check a convergence rate of $\sqrt{(\log n)/n}$, for instance, one could use something like $C_n^2 = \chi^2_{\dim(\gamma), 0.95} \log n$. 

With some additional structure, we can use $LM_1$ to deduce a similar bound on the proxies $\tilde e$. Let $\theta_*$ and $e_*$ be such that $\gamma(\theta_*,e_*) = \gamma_0$. We do not require that $\theta_*$ and $e_*$ are the true structural parameters and attributes, only that they induce $\gamma_0$. Let $\hat G_\theta = \frac{\partial \gamma(\hat \theta, \tilde e)'}{\partial \theta}$, $\hat G_e = \frac{\partial \gamma(\hat \theta, \tilde e)'}{\partial{\mathrm{vec}(e)}}$, $G_\theta = \frac{\partial \gamma(\theta_*, e_*)'}{\partial \theta}$, and $G_e = \frac{\partial \gamma(\theta_*, e_*)'}{\partial{\mathrm{vec}(e)}}$ (these are well defined under Assumption~\ref{a.1.e} below). Let $M = I - H^{1/2} G_\theta'(G_\theta H G_\theta')^{-1} G_\theta H^{1/2}$ denote the projection onto the orthogonal complement of the column span of $H^{1/2}G_\theta'$.

\medskip

\begin{assumption}\label{a.1.e}
Let the following hold:
\begin{enumerate}[label={(\roman*)}, nosep]
\item \label{a.1.e.1} $\hat \theta \to_p \theta_*$ and $\tilde e \to_p e_*$ with $\gamma_0 = \gamma(\theta_*,e_*)$;
\item \label{a.1.e.2} $\gamma(\theta,e)$ is continuously differentiable in both its arguments at $(\theta_*, e_*)$ and $G_\theta$ has full row rank;
\item \label{a.1.e.3} $\hat \theta$ satisfies the first-order condition $0 = \hat G_\theta \hat S$ and there exists a constant $C$ such that $\|\hat \theta - \theta_*\| \leq C \|\tilde e - e_*\|$ wpa1;
\item \label{a.1.e.4} $MH^{1/2} G_e'$ has full rank.
\end{enumerate}
\end{assumption}

\medskip

Let $\sigma_{\min}(MH^{1/2} G_e')$ denote the smallest singular value of the matrix $MH^{1/2} G_e'$. Note this is positive by Assumption~\ref{a.1.e}\ref{a.1.e.4}.

\medskip

\begin{proposition}\label{prop:lm.fit.e}
    Let Assumptions~\ref{a.1} and~\ref{a.1.e} hold. Fix any sequence $C_n \uparrow \infty$ and any $\epsilon > 0$. Then wpa1, we have
    \[
     \frac{1+\epsilon}{1+3\epsilon} \left(\sqrt{LM_1} - \epsilon C_n\right) \leq \|MH^{1/2} G_e' \sqrt n (\mathrm{vec}(\tilde e - e_*))\| \leq \frac{1+3\epsilon}{1+\epsilon} \left(\sqrt{LM_1} + \epsilon C_n\right).
    \]
    In particular, wpa1 we have that $LM_1 \leq C_n^2$ implies
    \[
     \|\sqrt n (\mathrm{vec}(\tilde e - e_*))\| \leq \frac{(1+3\epsilon) C_n}{\sigma_{\min}(MH^{1/2} G_e')}.
    \]
\end{proposition}

The proof of Proposition~\ref{prop:lm.fit.e} shows that the ``wpa1'' qualifier depends on whether a $\chi^2_{\mathrm{rank}(M)}$ random variable is less than $\epsilon^2 C_n^2$. With $\epsilon = 1$, say, this suggests taking $C_n^2$ to be at least as large as the 95th or 99th percentile of the $\chi^2_{\mathrm{rank}(M)}$ distribution.

\section{Conclusion}
\label{sec:conclusion}

In this paper, we develop a practical toolkit to correct bias and perform valid inference on counterfactuals when the product attributes used in demand estimation may only imperfectly capture the latent attributes that drive substitution. A leading case is when consumer choice is driven by difficult-to-quantify characteristics and unstructured data, such as product images, descriptions, review text, or consumer surveys, are converted into numerical variables using ML methods. As e-commerce continues to expand and such data play an increasingly central role in driving consumer choices, the need to incorporate these sources into demand estimation will only grow. In addition, our methods may be applied as simple post-estimation robustness checks even with standard numeric attributes when they are imperfect proxies. All our methods require minimal additional computation once model parameters are estimated and can be easily integrated in the canonical demand estimation workflow.

{
\let\oldbibliography\thebibliography
\renewcommand{\thebibliography}[1]{\oldbibliography{#1}
\setlength{\itemsep}{0pt}}

\putbib 
}
\end{bibunit}


\newpage 
\appendix
\onehalfspacing

\begin{bibunit}

\section{Proofs}
\label{app:proofs}

\subsection{Proofs for Section~\ref{sec:theory blp}}

\begin{proof}[Proof of Proposition~\ref{prop:kappa.endogenous}]
	We have $\hat \gamma \in N$ wpa1 by the assumed consistency of $\hat \gamma$. By Assumption~\ref{a.2}\ref{a.2.1}, wpa1 we may take a mean value expansion around $\gamma_0$ to obtain
	\[
	\begin{aligned}
		\sqrt T(\hat \kappa_{bc} - \kappa_0)
		& = \left( \frac{1}{\sqrt T} \sum_{t=1}^T \left( k_t(\gamma_0) - \kappa_0 - c' Z_t \hat \xi_t(\gamma_0) \right) + \sqrt T \sum_{t = 1}^\tau d_t' \left( \bar m_t - m_t(\gamma_0) \right) \right) \\
		& \quad + \left(  \frac{1}{\sqrt T} \sum_{t=1}^T  (c - \hat c)' Z_t \hat \xi_t(\gamma_0) + \sqrt{T} \sum_{s = 1}^\tau (\hat d_s - d_s)' \left( \bar m_s - m_s(\gamma_0) \right) \right) \\
		& \quad + \frac{1}{\sqrt T} \sum_{t=1}^T \left( \dot k_t(\tilde \gamma)' - \hat c' Z_t \dot \xi_t(\tilde \gamma) - \sum_{s=1}^\tau \hat d_s' \dot m_s(\tilde \gamma) \right) (\hat \gamma - \gamma_0) \\
		& =: T_{1,T} + T_{2,T} + T_{3,T},
	\end{aligned}
	\]
	where $\tilde \gamma$ is in the segment between $\hat \gamma$ and $\gamma_0$, and 
	\begin{equation}\label{eq:blp weights}
	\begin{aligned}
		c & = V^{-1} ( K + G H^{-1} h ) , &&&
		d_t & = (r_t V_t)^{-1} M_t H^{-1} h, \quad t = 1,\ldots,\tau.
	\end{aligned}
	\end{equation}
    Note that $c$ and $d_1,\ldots,d_\tau$ are well defined by virtue of Assumption~\ref{a.2}\ref{a.2.3}.
	
	For $T_{1,T}$, define the $(\dim(z) + 1) \times 1$ random vector $\zeta_t =( k_t(\gamma_0) - \kappa_0  , Z_t \hat \xi_t(\gamma_0) )$. By Theorem 2 of \cite{hahn2022joint} (noting Assumptions~\ref{a.2}\ref{a.2.1}\ref{a.2.2} and independence within and across markets are sufficient for their integrability and dependence conditions) and Assumption~\ref{a.2}\ref{a.2.4}, for any $\mathcal M$-measurable random vectors $d_1,\ldots,d_\tau$, we have 
	\[
	\left( \begin{array}{c}
		\frac{1}{\sqrt T} \sum_{t=1}^T \zeta_t \\
		\sum_{t = 1}^\tau  d_t'\sqrt T  \left( \bar m_t - m_t(\gamma_0) \right)
	\end{array} \right) 
	\to_d
	\left( \begin{array}{c}
		Z_A \\
		(\sum_{t=1}^\tau r_t d_t' V_t d_t)^{1/2} Z_M
	\end{array}
	\right)
	\]
	$\mathcal M$-stably, where the random vector $Z_A$ and random variable $Z_M$ are jointly normally distributed and independent with mean zero, $\mathrm{Var}(Z_A) = \mathrm{Var}(\zeta_t)$, and $\mathrm{Var}(Z_M) = 1$. Moreover, $(Z_A,Z_M)$ are independent of any $\mathcal M$-measurable random variable. Hence, the asymptotic distribution of $T_{1,T}$ is mixed Gaussian with mean zero and random variance
	\[
	\mathrm{Var}\left(k_t(\gamma_0) \right) + c' V c  - 2 c' K + \sum_{t=1}^\tau r_t d_t' V_t d_t.
	\]
	Substituting the above formulas for $c$ and $d_1,\ldots,d_\tau$ gives the form of the variance in display (\ref{eq:var.endogenous}). It remains to show that $T_{2,T}$ and $T_{3,T}$ are asymptotically negligible.
	
	For term $T_{2,T}$, first recall the expressions for $\hat h$ and $\hat H$ in~\eqref{eq:h blp}. Note that by Assumption~\ref{a.2}\ref{a.2.1}\ref{a.2.2} and consistency of $\hat \gamma$, we can deduce by standard arguments (e.g., Lemma~2.4 of \cite{newey1994large}) that $\hat k \to_p \E[\dot k_t(\gamma_0)]$, $\hat K \to_p K$, $\hat G \to_p G$, and $\hat V \to_p V$. Hence, $\hat h \to_p h$ by Assumption~\ref{a.2}\ref{a.2.3} and Slutsky's theorem. Note that for each $1 \leq t \leq \tau$, the $m_{it}$ are iid conditional on $\mathcal M$. It follows by Lemma~1 of \cite{andrews2005cross} and Assumption~\ref{a.2}\ref{a.2.2}\ref{a.2.4} that $\hat V_t \to_p r_t V_t$. Hence, $\hat V_t^{-1} \to_p (r_t V_t)^{-1}$ for $1 \leq t \leq \tau$ by Assumption~\ref{a.2}\ref{a.2.3}. Finally, Assumption~\ref{a.2}\ref{a.2.1} implies $\hat M_t \to_p M_t$ for $1 \leq t \leq \tau$. Hence, $\hat H \to_p H$ and so $\hat H^{-1} \to_p H^{-1}$, $\hat c \to_p c$, and $\hat d_t \to_p d_t$ for $1 \leq t \leq \tau$ by Assumption~\ref{a.2}\ref{a.2.3} and Slutsky's theorem.
	
	Now write
	\begin{multline*}
		T_{2,T} = (c - \hat c)' \frac{1}{\sqrt T} \sum_{t=1}^T  Z_t \hat \xi_t(\gamma_0) +  \sum_{t = 1}^\tau \frac{\sqrt T}{\sqrt{N_t}} (\hat d_t - d_t)' \sqrt{N_t} \left( \bar m_t - m_t(\gamma_0) \right) \\
		=: T_{2,T,a} + T_{2,T,b}.
	\end{multline*}
	Term $T_{2,T,a} \to_p 0$ because $\hat c \to_p c$ and $ \frac{1}{\sqrt T} \sum_{t=1}^T Z_t \hat \xi_t(\gamma_0) = O_p(1)$ by Assumption~\ref{a.2}\ref{a.2.1}. Similarly, $T_{2,T,b} \to_p 0$ because $T/N_t \to r_t \in (0,\infty)$ by Assumption~\ref{a.2}\ref{a.2.4}, $\hat d_t \to_p d_t$, and $\sqrt{N_t}(\bar m_t - m_t(\gamma_0))$ converges in distribution $\mathcal M$-stably to a mixed normal limit with mean zero and variance $V_t$ (by Assumption~\ref{a.2}\ref{a.2.2}) and is therefore tight by Assumption~\ref{a.2}\ref{a.2.2}.
	
	For term $T_{3,T}$, we first let $m_{tl}(\gamma)$ denote the $l$-th element of $m_t(\gamma)$, and let $\rho_{tl}(\gamma)$ denote the $l$-th element of $Z_t \hat \xi_t(\gamma)$. Similarly, we let $\hat c_l$ and $\hat d_{tl}$ denote the $l$-th elements of $\hat c$ and $\hat d_t$. Then we may write
	\[
	T_{3,T} = \frac{1}{\sqrt T} \sum_{t=1}^T \left( \dot k_t(\tilde \gamma)' - \sum_{l=1}^{\dim(z)} \hat c_l \dot \rho_{tl}(\tilde \gamma)' - \sum_{s=1}^\tau \sum_{l=1}^{\dim(m)} \hat d_{sl} \dot m_{sl}(\tilde \gamma)' \right) (\hat \gamma - \gamma_0) .
	\]
	By construction, $\hat c$ and $\hat d_1,\ldots,\hat d_\tau$ satisfy the in-sample orthogonality condition
	\begin{equation}\label{eq:orthogonal.sample.endogenous}
		\hat k - \hat G' \hat c - \sum_{s=1}^\tau \hat M_s' \hat d_s  =  0,
	\end{equation}
	wpa1. 
	By Assumption~\ref{a.2}\ref{a.2.1} we may take a second mean-value expansion, this time of $\tilde \gamma$ around $\hat \gamma$, to arrive at
	\[
	\begin{aligned}
		T_{3,T} 
		& = \frac{1}{\sqrt T} \sum_{t=1}^T \left( \dot k_t(\hat \gamma)' - \sum_{l=1}^{\dim(z)} \hat c_l \dot \rho_{tl}(\hat \gamma)' - \sum_{s=1}^\tau \sum_{l=1}^{\dim(m)} \hat d_{sl} \dot m_{sl}(\hat \gamma)' \right) (\hat \gamma - \gamma_0) \\
		& \quad + T^{1/4} (\tilde \gamma - \hat \gamma)' \left( \frac{1}{T} \sum_{t=1}^T \left( \ddot k_t(\check \gamma) - \sum_{l=1}^{\dim(z)} \hat c_l \ddot \rho_{tl}(\check \gamma) - \sum_{s=1}^\tau \sum_{l=1}^{\dim(m)} \hat d_{sl} \ddot m_{sl}(\check \gamma) \right) \right) T^{1/4} (\hat \gamma - \gamma_0) \\
		& =: T_{3,T,a} + T_{3,T,b},
	\end{aligned}
	\]
	wpa1, where $\check \gamma$ is in the segment between $\hat \gamma$ and $\tilde \gamma$, $\ddot k_t(\gamma) = \frac{\partial^2 k_t(\gamma)}{\partial \gamma \partial \gamma'}$, $\ddot \rho_{tl}(\gamma) = \frac{\partial^2 \rho_{tl}(\gamma)}{\partial \gamma \partial \gamma'}$, and $\ddot m_{tl}(\gamma) = \frac{\partial^2 m_{tl}(\gamma)}{\partial \gamma \partial \gamma'}$.
	
	We have 
    \[
    T_{3,T,a} = \left( \hat k - \hat G' \hat c - \sum_{s=1}^\tau \hat M_s' \hat d_s \right) \sqrt T(\hat \gamma - \gamma_0) = 0
    \]
    wpa1 by the in-sample orthogonality condition (\ref{eq:orthogonal.sample.endogenous}).
	
	To show $T_{3,T,b} \to_p 0$, in view of the condition $\hat \gamma = \gamma_0 + o_p(T^{-1/4})$, it is enough to show that the central term in parentheses is $O_p(1)$. To this end, standard arguments (e.g., Lemma~2.4 of \cite{newey1994large}) using Assumption~\ref{a.2}\ref{a.2.1} and consistency of $\hat \gamma$ yield $ \frac{1}{T} \sum_{t=1}^T \ddot k_t(\check \gamma) \to_p \E[ \ddot k_t(\gamma_0) ]$ and $\frac 1T \sum_{t=1}^T \ddot \rho_{tl}(\check \gamma) \to_p \E[ \ddot \rho_{tl}(\gamma_0)]$, both of which are finite. It also follows by the fact that $\hat d_t \to_p d_t$ for $1 \leq t \leq \tau$, Assumption~\ref{a.2}\ref{a.2.1}, and consistency of $\hat \gamma$ that $\hat d_{sl} \ddot m_{sl}(\check \gamma) \to_p d_{sl} \ddot m_{sl}(\gamma_0)$ for $1 \leq s \leq \tau$ and $1 \leq l \leq \dim(m)$. Finally, Assumption~\ref{a.2}\ref{a.2.1}-\ref{a.2.3} implies $c$ and $d_1,\ldots,d_\tau$ are tight.
\end{proof}

\begin{proof}[Proof of Proposition~\ref{prop:efficiency}]
    Arguing as in the proof of Proposition~\ref{prop:kappa.endogenous}, the asymptotic distribution of any estimator of the form \eqref{eq:kappa linearize} is mixed Gaussian with mean zero and variance 
    \[
    	\mathrm{Var}\left(k_t(\gamma_0) \right) + c' V c  - 2 c' K + \sum_{t=1}^\tau r_t d_t' V_t d_t.
    \]
    Conditioning on $\mathcal M$, we may minimize this expression with respect to the vectors $c$ and $d_1,\ldots,d_\tau$ subject to (\ref{eq:orthog.blp.micro}) to obtain the weights $c$ and $d_1,\ldots,d_\tau$ in \eqref{eq:blp weights}. Substituting into the above display yields the minimum variance $V_{bc}$ given in \eqref{eq:var.endogenous}.
\end{proof}

Before proving Proposition~\ref{prop:lm.fit.gamma.2}, we first state and prove a lemma.

\begin{lemma}\label{lem:lm.2}
	Let Assumption~\ref{a.2} hold and let $\hat \gamma = \gamma_0 + o_p(1)$. Then
	\[
	\sqrt T \hat S = Z_T + o_p(1) + (H + o_p(1))(\sqrt T (\hat \gamma - \gamma_0)),
	\]
	where $Z_T := G' V^{-1} \frac{1}{\sqrt T} \sum_{t=1}^T Z_t \hat \xi_t(\gamma_0) + \sum_{t=1}^\tau M_t' (r_t V_t)^{-1} \sqrt{T} (m_t(\gamma_0) - \bar m_t)$ converges in distribution ($\mathcal M$-stably) to a mixed Gaussian random variable with mean zero and $\mathcal M$-measurable variance $H$.
\end{lemma}

\begin{proof}[Proof of Lemma~\ref{lem:lm.2}]
	By definition of $\hat S$, we have
	\[
    \begin{aligned}
		\sqrt T \hat S & = \hat G' \hat V^{-1} \frac{1}{\sqrt T} \sum_{t=1}^T Z_t \hat \xi_t(\gamma_0) + \sum_{t=1}^\tau \hat M_t' \hat V_t^{-1} \sqrt{T} (m_t(\gamma_0) - \bar m_t) \\
        & \quad \quad  + \hat G' \hat V^{-1} \left( \frac{1}{\sqrt T} \sum_{t=1}^T Z_t (\hat \xi_t(\hat \gamma) - \hat \xi_t(\gamma_0)) \right) \\
		& \quad \quad + \sum_{t=1}^\tau \hat M_t' \hat V_t^{-1}  \sqrt{T} (m_t(\hat \gamma) - m_t(\gamma_0))
		=: T_{1,T} + T_{2,T} + T_{3,T} + T_{4,T},
	\end{aligned}
    \]
    where $\hat G \to_p G$, $\hat M_t \to_p M_t$ for $1 \leq t \leq \tau$, $\hat V^{-1} \to_p V^{-1}$, and $\hat V_t^{-1} \to_p (r_t V_t)^{-1}$ for $1 \leq t \leq \tau$ (all by the proof of Proposition~\ref{prop:kappa.endogenous}). 
	
	For $T_{1,T}$ and $T_{2,T}$, we  have by the proof of Proposition~\ref{prop:kappa.endogenous} that $\frac{1}{\sqrt T} \sum_{t=1}^T Z_t \hat \xi_t(\gamma_0)$ and $\sqrt{N_t} (m_t(\gamma_0) - \bar m_t)$, $1 \leq t \leq \tau$, are all $O_p(1)$. It follows by Assumption~\ref{a.2}\ref{a.2.4} that $T_{1,T} + T_{2,T} = Z_T + o_p(1)$. Hence, by similar arguments to the proof of Proposition~\ref{prop:kappa.endogenous} we may invoke Theorem 2 of \cite{hahn2022joint} to conclude that $Z_T$ converges $\mathcal M$-stably to a mixed Gaussian limit with mean zero and variance $H$.
	
	For $T_{3,T}$ and $T_{4,T}$, a mean-value expansion in $\hat \gamma$ around $\gamma_0$ yields 
	\[
	\begin{aligned}
		T_{3,T} & = \hat G' \hat V^{-1} \left( \frac 1T \sum_{t=1}^T Z_t \dot \xi_t(\tilde \gamma) \right)\sqrt T (\hat \gamma - \gamma_0), & 
		T_{4,T} & = \sum_{t=1}^\tau \hat M_t' \hat V_t^{-1}  \dot m_t(\tilde \gamma)\sqrt{T} (\hat \gamma - \gamma_0),
	\end{aligned}
	\]
	for $\tilde \gamma$ in the segment between $\hat \gamma$ and $\gamma_0$. It follows by Assumption~\ref{a.2}\ref{a.2.1} and standard arguments that $\frac 1T \sum_{t=1}^T Z_t \dot \xi_t(\tilde \gamma) \to_p  G$ and $\dot m_t(\tilde \gamma) \to_p M_t$, $1 \leq t \leq \tau$. Hence, $T_{3,T} + T_{4,T} = (H + o_p(1))\sqrt T(\hat \gamma - \gamma_0)$.
\end{proof}

\begin{proof}[Proof of Proposition~\ref{prop:lm.fit.gamma.2}]
    The proof of Proposition~\ref{prop:kappa.endogenous} shows that $\hat H \to_p H$. Combined with Lemma~\ref{lem:lm.2} and the triangle inequality, we have
    \begin{multline*}
        \| (H^{1/2} + o_p(1))(\sqrt T (\hat \gamma - \gamma_0))\| + \| (H^{-1/2} + o_p(1))(Z_T + o_p(1))\| \\
        \geq \sqrt{LM_1} 
        \geq \| (H^{1/2} + o_p(1))(\sqrt T (\hat \gamma - \gamma_0))\| - \| (H^{-1/2} + o_p(1))(Z_T + o_p(1))\|.
    \end{multline*}
    As $\| (H^{-1/2} + o_p(1))(Z_T + o_p(1))\|^2 \to_d \chi^2_{\dim(\gamma)}$ by the proof of Lemma~\ref{lem:lm.2}, we have
    \[
    \| (H^{-1/2} + o_p(1))(Z_T + o_p(1))\| \leq \epsilon C_T
    \]
    wpa1. Moreover, we have that 
    \[
    \frac{1 + 2\epsilon}{1 + \epsilon} \| H^{1/2}(\sqrt T (\hat \gamma - \gamma_0))\| \geq \| (H^{1/2} + o_p(1))(\sqrt T (\hat \gamma - \gamma_0))\| \geq \frac{1 + \epsilon}{1 + 2\epsilon} \| H^{1/2}(\sqrt T (\hat \gamma - \gamma_0))\|
    \]
    wpa1. The first result follows by combining the above three displays and rearranging. The second and third results are implications of the first.
\end{proof}

\subsection{Proofs for Section~\ref{sec:theory model 2}}

\begin{proof}[Proof of Proposition~\ref{prop:kappa.exogenous}]
    Let $c_i' = ( \E[ \dot k_i(\gamma_0) ] )' H^{-1} \dot \sigma_i(\gamma_0)' V_i(\gamma_0)^{-1}$ denote the population counterpart of $\hat c_i'$. We have $\hat \gamma \in N$ wpa1 by consistency of $\hat \gamma$. Hence, wpa1, we may take a mean value expansion around $\gamma_0$ to obtain
    \[
     \begin{aligned}
     \sqrt n(\hat \kappa_{bc} - \kappa_0)
     & = \frac{1}{\sqrt n} \sum_{i=1}^n \left( k_i(\gamma_0) - \kappa_0 + c_i'(d_i - \sigma_i(\gamma_0)) \right) \\
     & \quad + \frac{1}{\sqrt n} \sum_{i=1}^n (\hat c_i - c_i)' (d_i - \sigma_i(\gamma_0)) \\
     & \quad + \frac{1}{\sqrt n} \sum_{i=1}^n \left( \dot k_i(\tilde \gamma)' - \hat c_i' \dot \sigma_i(\tilde \gamma) \right) (\hat \gamma - \gamma_0) \\
     & =: T_{1,n} + T_{2,n} + T_{3,n},
     \end{aligned}
    \]
    where $\tilde \gamma$ is in the segment between $\hat \gamma$ and $\gamma_0$. This expansion is valid by Assumption~\ref{a.1}\ref{a.1.1}. Term $T_{1,n}$ is asymptotically $N(0,V_{bc})$. It remains to show that $T_{2,n}$ and $T_{3,n}$ are both asymptotically negligible.
    
    For $T_{2,n}$, first define the $1 \times \dim(\gamma)$ vectors
    \[
     \begin{aligned}
     \hat a & =  \bar k ' \hat H^{-1} , &&& 
     a & = \E[ \dot k_i(\gamma_0) ] ' H^{-1},
     \end{aligned}
    \]
    where $\bar k = \frac 1n \sum_{i=1}^n \dot k_i(\hat \gamma)$. 
    Also define the $\dim(\gamma) \times J$ random element $b_i(\gamma) = \dot \sigma_i(\gamma)' V_i(\gamma)^{-1}$, and let $e_i = d_i - \sigma_i(\gamma_0)$, which is $J \times 1$. Then we may write
    \[
     \begin{aligned}
     T_{2,n} 
     & = \hat a \left( \frac{1}{\sqrt n} \sum_{i=1}^n (b_i(\hat \gamma) - b_i(\gamma_0)) e_i \right) + (\hat a - a) \left( \frac{1}{\sqrt n} \sum_{i=1}^n b_i(\gamma_0) e_i \right) \\
     & =: T_{2,n,a} + T_{2,n,b}.
     \end{aligned}
    \]
    
    By Assumptions~\ref{a.1}\ref{a.1.1}\ref{a.1.2} and consistency of $\hat \gamma$, it follows by standard arguments \cite[e.g.,][Lemma~2.4]{newey1994large} that $\bar k \to_p \E[\dot k_i(\gamma_0)]$ and $\hat H \to_p H$. Hence, $\hat a \to_p a$ by Assumption~\ref{a.1}\ref{a.1.3}. 
    
    To show $T_{2,n,a} \to_p 0$, first note that $\hat a = O_p(1)$ and $\E[b_i(\gamma)e_i] = 0$. Consider the empirical process $\nu_n(\gamma) = \frac{1}{\sqrt n} \sum_{i=1}^n b_i(\gamma) e_i$ defined for $\gamma \in N$. For $\gamma_1, \gamma_2 \in N$, we have by a mean-value expansion that $b_i(\gamma_1) e_i - b_i(\gamma_2) e_i = (\gamma_1 - \gamma_2)' \dot b_i(\tilde \gamma)e_i$ for $\tilde \gamma$ in the segment between $\gamma_1$ and $\gamma_2$ (with possibly different values for each element), where $\dot b_i(\gamma)e_i = \frac{\partial}{\partial \gamma} (\dot b_i(\tilde \gamma)e_i)'$. This expansion is valid in view of Assumption~\ref{a.1}\ref{a.1.2}. The elements of $e_i$ are bounded by $\pm 1$ and the elements of $\dot b_i(\gamma)$ are uniformly (for $\gamma \in N$) bounded by some random variable with finite second moment, again by Assumption~\ref{a.1}\ref{a.1.2}. Hence, $\|b_i(\gamma_1) e_i - b_i(\gamma_2) e_i\| \leq B_i \|\gamma_1 - \gamma_2\|$ for $\gamma_1, \gamma_2 \in N$, for some random variable $B_i$ with finite second moment. Thus, $\{b_i(\gamma) e_i : \gamma \in N\}$ is a type-II class of \cite{andrews1994empirical}. It follows by Theorems~1 and~2 of \cite{andrews1994empirical} (using  Assumption~\ref{a.1}\ref{a.1.2} to verify the moment condition on the envelope function) that $\nu_n(\cdot)$ is stochastically equicontinuous. Also note by the Lipschitz condition the pseudometric corresponding to this process is dominated by the Euclidean metric. Hence,  by consistency of $\hat \gamma$ we have $\frac{1}{\sqrt n} \sum_{i=1}^n (b_i(\hat \gamma) - b_i(\gamma_0)) e_i \to_p 0$. 
    
    To show $T_{2,n,b} \to_p 0$, first note $\hat a \to_p a$. Moreover, $\E[\|b_i(\gamma_0) e_i \|^2] < \infty$ by Assumption~\ref{a.1}\ref{a.1.2} and $\E[b_i(\gamma_0)e_i] = 0$, so $\frac{1}{\sqrt n} \sum_{i=1}^n b_i(\gamma_0) e_i  = O_p(1)$ by Chebyshev's inequality. 
    
    For term $T_{3,n}$, we first write
    \[
     T_{3,n} = \frac{1}{\sqrt n} \sum_{i=1}^n \left( \dot k_i(\tilde \gamma)' - \sum_{j=1}^J \hat c_{ij} \dot \sigma_{ij}(\tilde \gamma)' \right) (\hat \gamma - \gamma_0),
    \]
    where $\hat c_i = (\hat c_{i1}, \ldots, \hat c_{iJ})$, and $\dot \sigma_{ij}(\gamma) = \frac{\partial \sigma_{ij}(\gamma)}{\partial \gamma}$.  Note by construction that $\hat c_i$ satisfies the in-sample orthogonality condition
    \begin{equation}\label{eq:orthogonal.sample}
     \frac{1}{n} \sum_{i=1}^n \left( \dot k_i(\hat \gamma) - \dot \sigma_i(\hat \gamma)' \hat c_i \right) = 0.
    \end{equation}
    A second mean-value expansion, this time of $\tilde \gamma$ around $\hat \gamma$, yields
    \[
     \begin{aligned}
     T_{3,n} 
     & = \frac{1}{\sqrt n} \sum_{i=1}^n \left( \dot k_i(\hat \gamma)' - \hat c_{i}' \dot \sigma_{i}(\hat \gamma) \right) (\hat \gamma - \gamma_0) \\
     & \quad + n^{1/4} (\tilde \gamma - \hat \gamma)' \left( \frac{1}{n} \sum_{i=1}^n \left( \ddot k_i(\check \gamma) - \sum_{j=1}^J \hat c_{ij} \ddot \sigma_{ij}(\check \gamma) \right) \right) n^{1/4} (\hat \gamma - \gamma_0) \\
     & =: T_{3,n,a} + T_{3,n,b},
     \end{aligned}
    \]
    where $\check \gamma$ is in the segment between $\hat \gamma$ and $\tilde \gamma$, $\ddot k_i(\gamma) = \frac{\partial^2 k_i(\gamma)}{\partial \gamma \partial \gamma'}$, and $\ddot \sigma_{ij}(\gamma) = \frac{\partial^2 \sigma_{ij}(\gamma)}{\partial \gamma \partial \gamma'}$. 
    We have $T_{3,n,a} = 0$ by the in-sample orthogonality condition (\ref{eq:orthogonal.sample}).
    
    To show $T_{3,n,b} \to_p 0$, in view of the condition $\hat \gamma = \gamma_0 + o_p(n^{-1/4})$, it is enough to show that the central term in parentheses is $O_p(1)$. To this end, standard arguments \cite[e.g.,][Lemma~2.4]{newey1994large} using Assumption~\ref{a.1}\ref{a.1.1} and consistency of $\hat \gamma$ yield $ \frac{1}{n} \sum_{i=1}^n \ddot k_i(\check \gamma) \to_p \E[ \ddot k_i(\gamma_0) ]$, which is finite. We may similarly deduce by the fact that $\hat a \to_p a$ and Assumption~\ref{a.1}\ref{a.1.1}-\ref{a.1.3} that $\frac{1}{n} \sum_{i=1}^n \hat c_{ij} \ddot \sigma_{ij}(\check \gamma) \to_p \E[c_{ij} \ddot \sigma_{ij}(\gamma_0)]$, which is finite, for $j = 1,\ldots,J$.
\end{proof}

\begin{proof}[Proof of Proposition~\ref{prop:lm.fit.gamma}]
	Analogous to the proof of Proposition~\ref{prop:lm.fit.gamma.2}, using Lemma~\ref{lem:lm.1} below in place of Lemma~\ref{lem:lm.2}.
\end{proof}

\begin{lemma}\label{lem:lm.1}
    Let Assumption~\ref{a.1} hold and let $\hat \gamma = \gamma_0 + o_p(1)$. Then
    \[
    \sqrt n \hat S = Z_n + o_p(1) - (H + o_p(1))(\sqrt n (\hat \gamma - \gamma_0)),
    \]
    where $Z_n := \frac{1}{\sqrt n} \sum_{i=1}^n \dot \sigma_i(\gamma_0)' V_i(\gamma_0)^{-1} (d_i - \sigma_i(\gamma_0)) \to_d N(0,H)$.
\end{lemma}

\begin{proof}[Proof of Lemma~\ref{lem:lm.1}]
    First note that since $\sum_{j=0}^J \dot \sigma_{ij}(\gamma) = 0$ for all $\gamma \in \Gamma$, we may write
    \[
    \begin{aligned}
        \sqrt n \hat S & = \frac{1}{\sqrt n} \sum_{i=1}^n \sum_{j=0}^J \frac{d_{ij} - \sigma_{ij}(\gamma_0)}{\sigma_{ij}(\hat \gamma)} \dot \sigma_{ij}(\hat \gamma) + \frac{1}{\sqrt n} \sum_{i=1}^n \sum_{j=0}^J \frac{\sigma_{ij}(\gamma_0) - \sigma_{ij}(\hat \gamma)}{\sigma_{ij}(\hat \gamma)} \dot \sigma_{ij}(\hat \gamma) \\
        & =: T_{1,n} + T_{2,n}.
    \end{aligned}
    \]
    
    For $T_{1,n}$, we may rewrite this term using the notation from the proof of Proposition~\ref{prop:kappa.exogenous} as
    \[
    \begin{aligned}
        T_{1,n} & = \frac{1}{\sqrt n} \sum_{i=1}^n b_i(\gamma_0) e_i + \frac{1}{\sqrt n} \sum_{i=1}^n \left( b_i(\hat \gamma) - b_i(\gamma_0)\right) e_i \\
        & =: T_{1,n,a} + T_{1,n,b}
    \end{aligned}
    \]
    where $b_i(\gamma) = \dot \sigma_i(\gamma)' V_i(\gamma)^{-1}$ and $e_i = d_i - \sigma_i(\gamma_0)$. The summands in $T_{1,n,a}$ have mean zero and variance $H$, which is finite and non-singular by Assumption~\ref{a.1}\ref{a.1.3}. Hence, $T_{1,n,a} \to_d N(0,H)$. The proof of Proposition~\ref{prop:kappa.exogenous} shows that the empirical process $\nu_n(\gamma) = \frac{1}{\sqrt n} \sum_{i=1}^n b_i(\gamma) e_i$ defined for $\gamma \in N$, a suitable neighborhood of $\gamma_0$, is stochastically equicontinuous under Assumption~\ref{a.1}\ref{a.1.2}. Hence, $T_{1,n,b} \to_p 0$.
    
    For $T_{2,n}$, a mean-value expansion in $\gamma_0$ around $\hat \gamma$ yields
    \[
    T_{2,n} = \left(\frac 1n \sum_{i=1}^n \sum_{j=0}^J \frac{\dot \sigma_{ij}(\hat \gamma) \dot \sigma_{ij}(\tilde \gamma)'}{\sigma_{ij}(\hat \gamma)}  \right)\sqrt n (\gamma_0 - \hat \gamma),
    \]
    for $\tilde \gamma$ in the segment between $\gamma_0$ and $\hat \gamma$ (with possibly different values for each element). This expansion is valid in view of Assumption~\ref{a.1}\ref{a.1.1}. For the term in parentheses, note
    \[
    \frac 1n \sum_{i=1}^n \sum_{j=0}^J \frac{\dot \sigma_{ij}(\hat \gamma) \dot \sigma_{ij}(\tilde \gamma)'}{\sigma_{ij}(\hat \gamma)} = \frac 1n \sum_{i=1}^n \dot \sigma_{i}(\hat \gamma)' V_i(\hat \gamma)^{-1} \dot \sigma_{i}(\tilde \gamma). 
    \]
    Standard arguments (e.g., Lemma~2.4 of \cite{newey1994large}) then yield that $\frac 1n \sum_{i=1}^n \dot \sigma_{i}(\hat \gamma)' V_i(\hat \gamma)^{-1} \dot \sigma_{i}(\tilde \gamma) \to_p  H$ under Assumption~\ref{a.1}\ref{a.1.1}\ref{a.1.2}.
\end{proof}

\begin{proof}[Proof of Proposition~\ref{prop:lm.fit.e}]
    First note that Assumption~\ref{a.1.e}\ref{a.1.e.1}\ref{a.1.e.2} implies $\hat \gamma \to_p \gamma_0$. Moreover, $\hat H \to_p H$ by the proof of Proposition~\ref{prop:kappa.exogenous}, $H$ is positive definite by Assumption~\ref{a.1}\ref{a.1.3}, and $\hat G_\theta \to_p G_\theta$ by Assumption~\ref{a.1.e}\ref{a.1.e.1}\ref{a.1.e.2}. It follows by Assumption~\ref{a.1.e}\ref{a.1.e.2} that $\hat M = I - \hat H^{1/2} \hat G_\theta'(\hat G_\theta \hat H \hat G_\theta')^{-1} \hat G_\theta \hat H^{1/2}$ exists wpa1 and $\hat M \to_p M$. Now by the first-order condition $0=\hat G_\theta\hat S$ in Assumption~\ref{a.1.e}\ref{a.1.e.3} and Lemma~\ref{lem:lm.1},
    \[
     \sqrt n \hat H^{-1/2} \hat S = \sqrt n \hat M \hat H^{-1/2} \hat S = \hat M \hat H^{-1/2}( Z_n + o_p(1)) - \hat M \hat H^{-1/2}(H + o_p(1))(\sqrt n (\hat \gamma - \gamma_0)),
    \]
    where $\hat M = I - \hat H^{1/2} \hat G_\theta'(\hat G_\theta \hat H \hat G_\theta')^{-1} \hat G_\theta \hat H^{1/2}$. Hence,
    \[
     \sqrt n \hat H^{-1/2} \hat S = (M + o_p(1))(H^{-1/2} Z_n + o_p(1)) - (M H^{1/2} + o_p(1))(\sqrt n (\hat \gamma - \gamma_0)).
    \]
    A mean-value expansion of $\hat \gamma$ around $(\theta_*,e_*)$ yields
    \[
     \sqrt n( \hat \gamma - \gamma_0 ) = (G_\theta' + o_p(1)) \sqrt n (\hat \theta - \theta_*) + ( G_e' + o_p(1)) \sqrt n (\mathrm{vec}(\tilde e - e_*)).
    \]
    Since $M H^{1/2} G_\theta' = 0$, we have by Assumption~\ref{a.1.e}\ref{a.1.e.3}\ref{a.1.e.4} that 
    \[
     \|(M H^{1/2} + o_p(1))(G_\theta' + o_p(1)) \sqrt n (\hat \theta - \theta_*)\| \leq \frac{\epsilon}{1+3\epsilon} \|MH^{1/2} G_e' \sqrt n (\mathrm{vec}(\tilde e - e_*))\|
    \]
    wpa1. Moreover,
    \begin{multline*}
     \frac{1+2\epsilon}{1+3\epsilon} \| MH^{1/2} G_e'\sqrt n (\mathrm{vec}(\tilde e - e_*)) \| \\
     \leq \| (MH^{1/2} + o_p(1))( G_e' + o_p(1))\sqrt n (\mathrm{vec}(\tilde e - e_*))\| \\
     \leq \frac{1+2\epsilon}{1+\epsilon} \| MH^{1/2} G_e'\sqrt n (\mathrm{vec}(\tilde e - e_*)) \|
    \end{multline*}
    wpa1. We also have $\|M H^{-1/2} Z_n\|^2 \to_d \chi^2_{\mathrm{rank}(M)}$ by Lemma~\ref{lem:lm.1}, which implies that the inequality $\| (M + o_p(1))(H^{-1/2} Z_n + o_p(1))\| \leq \epsilon C_n$ holds wpa1. Hence, wpa1,
    \begin{multline*}
     \frac{1+3\epsilon}{1+\epsilon}\|MH^{1/2} G_e' \sqrt n (\mathrm{vec}(\tilde e - e_*))\| + \epsilon C_n \\
     \geq \sqrt{LM_1} \geq \frac{1+\epsilon}{1+3\epsilon}\|MH^{1/2} G_e' \sqrt n (\mathrm{vec}(\tilde e - e_*))\| - \epsilon C_n .
    \end{multline*}
    The first result follows by rearranging. The second is an immediate implication.
\end{proof}

\section{Diagnostic 2}

We now outline our $LM_2$ diagnostics for the models in Sections~\ref{sec:blp} and~\ref{sec: exog prices}. Theory for these diagnostics can be developed by straightforward adaptation of arguments from \cite{andrewsOptimalTestsWhen1994} and \cite{hansen1996inference}.

\subsection{Case 1: Endogenous Prices}
\label{app: LM2 case 1}

To introduce the diagnostic, we extend $\gamma$ to $\zeta = (\gamma,\psi)$. With slight abuse of notation, we now write $\hat \xi_{jt}$ and $m_t$ on this extended space as functions $\hat \xi_{jt}(\zeta;\eta)$ and $m_t(\zeta; \eta)$ of $\zeta$ and $\eta$, with the understanding that $\hat \xi_{jt}(\gamma) = \hat \xi_{jt}((\gamma, 0);\eta)$ and $m_{t}(\gamma) = m_{t}((\gamma, 0);\eta)$. Let
\begin{align}
\label{eq: lambda^hat}
 \hat \lambda(\eta) & = \left( \frac{1}{T} \sum_{t=1}^T Z_t w_t(\hat \gamma, \eta) \right)' \hat V^{-1} ( I - \hat G (\hat G' \hat V^{-1} \hat G )^{-1} \hat G' \hat V^{-1} ) \sqrt T \bar g, \\
 \hat \lambda_t(\eta) & = u_t(\hat \gamma, \eta)' \hat V_t^{-1}( I - \hat M_t ( \hat M_t' \hat V_t^{-1} \hat M_t )^{-1} \hat M_t' \hat V_t^{-1} ) \sqrt T( m_t(\hat \gamma) - \bar m_t), \label{eq: lambda^hat_t}
\end{align}
where $w_t(\gamma, \eta) = (w_{jt}(\gamma,\eta)')_{j \in \mathcal J_t}$ is $|\mathcal J_t| \times \dim(\psi)$ and $u_t(\gamma, \eta)$ is $\dim(m) \times \dim(\psi)$, with
\[
 \begin{aligned}
 w_{jt}(\gamma, \eta) & = \lim_{\psi \to 0} \frac{\partial \hat \xi_{jt}((\gamma, \psi); \eta)}{\partial \psi}, \quad j \in \mathcal{J}_t, & & & 
 u_t(\gamma, \eta)' & = \lim_{\psi \to 0} \frac{\partial m_t((\gamma,\psi);\eta)'}{\partial \psi}.
 \end{aligned}
\]
We take limits to deal with parameters that are at the boundary when $\psi = 0$, such as the variance of the random coefficients on $\eta$. For the intuition, the left-most terms in (\ref{eq: lambda^hat}) and (\ref{eq: lambda^hat_t}) are the Jacobian of the moments with respect to $\psi$. These can depend to first order on $\hat \gamma$. The second parts of these expressions eliminate this dependence with a similar correction to $\hat \kappa_{bc}$. To ensure that this elimination does not make the test statistic degenerate, in practice we form the moments and gradients in these statistics using an enlarged set of instruments formed by augmenting $Z_t$ with an additional $\dim(\psi) \times |\mathcal J_t|$ set of ``optimal instruments'' for $\psi$. These are given by $\tilde Z_t(\eta) = \hat w_t(\hat \gamma, \hat \eta)'$, where $\hat w_t$ is formed as above but using predicted prices $\hat p_t$ and corresponding predicted market shares $\hat s_t$ as in Remark~\ref{rmk:optimal}. The dependence of $\tilde Z_t$ on $\eta$ does not affect the properties of the statistic and we suppress this notation in what follows, subsuming $\tilde Z_t(\eta)$ in $Z_t$.
We estimate the variance of $\hat \lambda(\eta)$ and $\hat \lambda_t(\eta)$ using
\begin{align}
\label{eq: Lambda^hat}
 \hat \Lambda(\eta) & = \left( \frac{1}{T} \sum_{t=1}^T Z_t w_t(\hat \gamma, \eta) \right)' \hat V^{-1} ( I - \hat G (\hat G' \hat V^{-1} \hat G )^{-1} \hat G' \hat V^{-1} ) \left( \frac{1}{T} \sum_{t=1}^T Z_t w_t(\hat \gamma, \eta) \right) , \\
 \hat \Lambda_t(\eta) & = u_t(\hat \gamma, \eta)' \hat V_t^{-1}( I - \hat M_t ( \hat M_t' \hat V_t^{-1} \hat M_t )^{-1} \hat M_t' \hat V_t^{-1} )  u_t(\hat \gamma, \eta).\notag
\end{align}
Finally, define
\[
 \hat W(\eta) = \left(\hat \lambda(\eta) + \sum_{t=1}^\tau \hat \lambda_t(\eta)\right)' \left( \hat \Lambda(\eta) + \sum_{t=1}^\tau \hat \Lambda_t(\eta) \right)^{-1} \left(\hat \lambda(\eta) + \sum_{t=1}^\tau \hat \lambda_t(\eta)\right).
\]
Without microdata, $\hat W(\eta)$ simplifies to $\hat W(\eta) = \hat \lambda(\eta)' \hat \Lambda(\eta)^{-1} \hat\lambda(\eta)$.
The statistic $\hat W(\eta)$ can be shown to behave like a $\chi^2_{\dim(\psi)}$ random variable under the null (i.e., when the true $\psi = 0$), but it depends on the nuisance parameter $\eta$. Thus, we define our second diagnostic as
\begin{equation}
\label{eq: LM2 blp}
 LM_2 = \sup_{\eta \in S^J : \eta \perp C(\tilde e)} \hat W(\eta),
\end{equation}
where we take the supremum over $\eta$ in the unit sphere $S^J$ (since the scale of $\eta$ is not important) that are orthogonal to the column span of $\tilde e$. 

Following, e.g., \cite{hansen1996inference}, critical values can be computed by simulation. Draw $(\varpi_t^*)_{t=1}^T, (\varpi_{i1}^*)_{i=1}^{N_1},\ldots, (\varpi_{i\tau}^*)_{i=1}^{N_\tau}$ iid from a $N(0,1)$ distribution, and set
\[
 \hat W^*(\eta) = \left(\hat \lambda^*(\eta) + \sum_{t=1}^\tau \hat \lambda_t^*(\eta)\right)' \left( \hat \Lambda(\eta) + \sum_{t=1}^\tau \hat \Lambda_t(\eta) \right)^{-1} \left(\hat \lambda^*(\eta) + \sum_{t=1}^\tau \hat \lambda_t^*(\eta)\right),
\]
where
\[
 \begin{aligned}
 \hat \lambda^*(\eta) & =  \left( \frac{1}{T} \sum_{t=1}^T Z_t w_t(\hat \gamma, \eta) \right)' \hat V^{-1} ( I - \hat G (\hat G' \hat V^{-1} \hat G )^{-1} \hat G' \hat V^{-1} ) \left( \frac{1}{\sqrt T} \sum_{t=1}^T \varpi_t^* Z_t \hat \xi_t(\hat \gamma) \right) , \\
 \hat \lambda_t^*(\eta) & = u_t(\hat \gamma, \eta)' \hat V_t^{-1}( I - \hat M_t ( \hat M_t' \hat V_t^{-1} \hat M_t )^{-1} \hat M_t' \hat V_t^{-1} ) \left( \frac{\sqrt T}{N_t} \sum_{i=1}^{N_t} \varpi_{it}^*( m_t(\hat \gamma) - m_{it}) \right).
 \end{aligned}
\]
For each collection of $N(0,1)$ draws, compute
\[
 LM_2^* = \sup_{\eta \in S^J : \eta \perp C(\tilde e)} \hat W^*(\eta).
\]
Let $\hat \xi_{0.95}^*$ denote the 95th percentile of $LM_2^*$ across a large number of independent draws. This quantity is easy to compute, as only the right-most terms in the expressions for $\hat \lambda^*$ and $\hat \lambda^*_t$ need to be recomputed for different draws and these terms do not depend on $\eta$. We reject the null that the dimension of $\tilde e$ is adequate if $LM_2 > \hat \xi_{0.95}^*$.

\begin{continuance}{ex:blp}
Let the random coefficient on $\eta$ be $\delta_i = \psi_1 + \sqrt{\psi_2} Z_i$, where $Z_i \sim F_0$ has mean zero and unit variance, $\psi_1 \in \mathbb R$, and $\psi_2 \in [0,\infty)$.\footnote{We parametrize $\psi$ in terms of the variance of the random coefficient to avoid boundary issues, following the recommendation of \cite{ketzAsymptoticSizeDistortions2019}. The resulting test is similar to the test for neglected heterogeneity of \cite{chesherTestingNeglectedHeterogeneity1984}.} For simplicity, suppose $\bar x_t$ is empty and $|\mathcal J_t| = J$. Define the functions $\sigma_{jt}(\,\cdot\,; (\gamma, \psi), \eta)$ and $\varsigma_{jt}(\,\cdot\,;\gamma)$  from $\mathbb R^{J}$ to $\mathbb R$ by
\[
 \begin{aligned}
 \sigma_{jt}(\xi_t; (\gamma, \psi), \eta) & = \int \varsigma_{jt}(\xi_t + (\psi_1 + \sqrt{\psi_2} z) \eta ; \gamma) \, d F_0(z),  \\
 \varsigma_{jt}(u; \gamma) & = \int \frac{e^{-\alpha' p_{jt} + \beta' e_j + u_j}}{1 + \sum_{k \in \mathcal J_t} e^{-\alpha' p_{kt} + \beta' e_k + u_k} } dF(\alpha, \beta),  
 \end{aligned}
\]
where we have suppressed dependence on $p_t$. Here $\hat \xi_{jt}(\zeta; \eta)$ solves $s_{jt} = \sigma_{jt}(\hat \xi; \zeta, \eta)$ for $j \in \mathcal J_t$. Evidently, $\hat \xi_{jt}(\gamma) = \hat \xi_{jt}((\gamma, 0);\eta)$. 
Using $\dot \varsigma_{jt}(u;\gamma)$ and $\ddot \varsigma_{jt}(u;\gamma)$ to denote the first and second derivatives of $\varsigma_{jt}(u;\gamma)$ with respect to $u$, we have
\[
 \lim_{\psi \to 0} \left( \left. \frac{\partial \sigma_{jt}(\xi_t;(\gamma, \psi), \eta)}{\partial \psi} \right|_{\xi_t = \hat \xi_t((\gamma,\psi); \eta)}\right) = \left( \begin{array}{c}
  \eta' \dot \varsigma_{jt}(\hat \xi_{t}(\gamma); \gamma) \\
 \frac 12 \eta' \ddot \varsigma_{jt}(\hat \xi_{t}(\gamma); \gamma) \eta
 \end{array} \right).
\]
It follows by the implicit function theorem that $w_t(\gamma, \eta)$ is $J \times 2$ and is given by
\[
 w_t(\gamma, \eta) = - \begin{pmatrix}
  \eta &&
 \left( \frac{\partial \varsigma_t(\hat \xi_t(\gamma); \gamma)}{\partial \xi'} \right)^{-1} \left( \frac 12  \eta' \ddot \varsigma_{jt}(\hat \xi_{t}(\gamma); \gamma) \eta
 \right)_{j \in \mathcal J_t} \end{pmatrix},
\]
where $\varsigma_t(u; \gamma) = (\varsigma_{jt}(u; \gamma))_{j \in \mathcal J_t}$. Plugging this into \eqref{eq: lambda^hat} and \eqref{eq: Lambda^hat}, one can compute $\hat W(\eta)$ and thus the $LM_2$ diagnostic.
\end{continuance}

\subsection{Case 2: Individual-Level Price Variation}
\label{app: LM2 case 2}

We augment $\tilde e$ with a vector $\eta$ representing additional attributes not included in $\tilde e$ and augment $\theta$ with an additional component $\psi \in \Psi$ representing coefficients on $\eta$. Correspondingly, we extend $\gamma$ to $\zeta = (\gamma, \psi)$. With slight abuse of notation, we now write choice probabilities on this extended space as $\sigma_{ij}(\zeta; \eta)$, with the understanding that $\sigma_{ij}(\gamma) = \sigma_{ij}((\gamma, 0);\eta)$.

Consider an LM test of the null hypothesis that $\psi = 0$. Such a test could be based on the score
\begin{equation}\label{eq:lm.rank.prelim}
 \frac{1}{n} \sum_{i=1}^n \sum_{j=0}^J \frac{d_{ij}}{\sigma_{ij}(\hat \gamma)} w_{ij}(\hat \gamma,\eta),
\end{equation}
where
\[
 w_{ij}(\gamma, \eta) = \lim_{\psi \to 0} \frac{\partial \sigma_{ij}((\gamma, \psi);\eta)}{\partial \psi}.
\]

\medskip

\begin{continuance}{ex:mixed logit}
Let the random coefficient on $\eta$ be $\delta_i = \sqrt{\psi} Z_i$, where $\psi \in [0,\infty)$ and $Z_i \sim F_0$ has mean zero and unit variance. For simplicity, suppose $\Pi = 0$. Define the functions $\sigma_{ij}(\,\cdot\,; (\gamma, \psi), \eta)$ and $\varsigma_{ij}(\,\cdot\,; (\gamma, \psi), \eta)$ from $\mathbb R^J$ to $\mathbb R$ by 
\[ 
 \begin{aligned}
 \sigma_{ij}((\gamma, \psi); \eta) & = \int \varsigma_{ij}(\xi + \sqrt{\psi} \eta z; \gamma) \, d F_0(z), \\
 \varsigma_{ij}(u; \gamma) & = \int \frac{e^{\alpha' p_{ij} + \beta' e_j + u_j}}{1 + \sum_{k=1}^J e^{\alpha' p_{ik} + \beta' e_k + u_k} } dF(\alpha, \beta),
 \end{aligned}
\]
where we have suppressed dependence on $p_{ij}$ and $\bar x$. 
Using $\dot \varsigma_{ij}$ and $\ddot \varsigma_{ij}$ to denote first and second derivatives of $\varsigma_{ij}$ with respect to its first argument, we have
\[
 \frac{\partial \sigma_{ij}((\gamma, \psi);\eta)}{\partial \psi} = \frac{1}{2 \sqrt \psi} \int z \eta' \dot \varsigma_{ij}(\xi + \sqrt \psi \eta z; \gamma) \, dF_0(z).
\]
Thus, $w_{ij}(\gamma, \eta) =  \frac 12 \eta' \ddot \varsigma_{ij}(\xi; \gamma) \eta$. The term $w_{ij}(\hat \gamma, \eta) =  \frac 12 \eta' \ddot \varsigma_{ij}(\hat \xi; \hat \gamma) \eta$ is plugged into the $LM_2$ statistic below.
\end{continuance}

\bigskip

The statistic (\ref{eq:lm.rank.prelim}) can still depend to first order on $\hat \gamma$. To eliminate this dependence, we perform a similar correction to $\hat \kappa_{bc}$. Let $g_i(\gamma;\eta) = w_i(\gamma, \eta)' V_i(\gamma)^{-1} (d_i - \sigma_i(\gamma))$, where $w_i(\gamma; \eta) = (w_{ij}(\gamma; \eta)')_{j=1}^J$ is $J \times \dim(\psi)$. Then define the $\dim(\psi) \times J$ matrix
\[
 \hat c_i(\eta)' = \left(w_i(\hat \gamma, \eta) + \dot \sigma_i(\hat \gamma) \hat H^{-1}  \left( \frac 1n \sum_{l=1}^n \dot g_l(\hat \gamma;\eta)' \right)\right)' V_i(\hat \gamma)^{-1},
\]
where $\dot g_i(\gamma;\eta)' = \frac{\partial g_i(\gamma; \eta)}{\partial \gamma}$ is $\dim(\gamma) \times \dim(\psi)$, and let
\[
 \begin{aligned}
 \hat \lambda(\eta) & = \frac{1}{\sqrt n } \sum_{i=1}^n \hat c_i(\eta)'(d_i - \sigma_i(\hat \gamma)), &&& 
 \hat \Lambda(\eta) & =  \frac 1n \sum_{i=1}^n  \hat c_i(\eta)' V_i(\hat \gamma) \hat c_i(\eta) .
 \end{aligned}
\]
Finally, let
\[
 \hat W(\eta) = \hat \lambda(\eta)'\hat \Lambda(\eta)^{-1} \hat \lambda(\eta).
\]
It can be shown that the statistic $\hat W(\eta)$ behaves like a $\chi^2_{\dim(\psi)}$ random variable under the null (i.e., when the true $\psi = 0$), but it depends on the nuisance parameter $\eta$. Thus, we again define our second diagnostic as 
\begin{equation}\label{eq:lm.rank}
 LM_2 = \sup_{\eta \in S^J : \eta \perp C(\tilde e)} \hat W(\eta).
\end{equation}
Critical values can be computed by simulation: for iid $N(0,1)$ random variables $(\varpi_i^*)_{i=1}^n$, compute
\[
 LM_2^* = \sup_{\eta \in S^J : \eta \perp C(\tilde e)} \hat W^*(\eta),
\]
where $\hat W^*(\eta) = \hat \lambda^*(\eta)'\hat \Lambda(\eta)^{-1} \hat \lambda^*(\eta)$, with $\hat \lambda^*(\eta) = \frac{1}{\sqrt n } \sum_{i=1}^n \varpi_i^* \hat c_i(\eta)'(d_i - \sigma_i(\hat \gamma))$. Note $\hat \lambda^*(\eta)$ factors into the product of terms involving $\eta$, which only need to be computed once, and $(\varpi_i^*)_{i=1}^n$, making it trivial to compute. Let $\hat \xi_{0.95}^*$ denote the 95th percentile of $LM_2^*$ across a large number of independent sequences $(\varpi_i^*)_{i=1}^n$. The null that the dimension of $\tilde e$ is adequate can be rejected if $LM_2 > \hat \xi^*_{0.95}$.

\section{Fixed Effects in Case 1}\label{app:fes}
Here we briefly discuss the asymptotic properties of the bias corrected estimator, standard error formulas, and diagnostics, when utilities depend additively on fixed effects. Recall that product and/or market fixed effects are handled by de-meaning $Z_t$ and/or $\hat \xi_t$ at the product and/or market level when forming moments, as in \cite{somainiAlgorithmEstimateTwoWay2016} and \cite{conlon2020best}.

\emph{Market Fixed Effects.} De-meaning at the market level introduces correlation among the instruments within each market and among the $\hat \xi_t$ within each market. This does not affect the theoretical derivations, which already allow for arbitrary correlation among the elements of $Z_t$ and $\hat \xi_t$ within markets. In particular, the standard error formula~\eqref{eq:standard errors blp} still applies as it clusters at the market level. The bootstrap for generating critical values for $LM_2$ does not need to be modified either, as it is also based on clustering at the market level.

\emph{Product Fixed Effects.} De-meaning at the product level introduces correlation across markets, but this is asymptotically negligible. It is straightforward to show that the theoretical results from Section~\ref{sec:theory blp} carry over to this case. An alternative (equivalent) approach is to simply include dummy variables for each product, in which case the theory carries over directly once the parameter space is suitably expanded.

\emph{Product and Market Fixed Effects.} In this case de-meaning introduces correlation both across products within markets and across markets. The correlation across products within markets is accommodated by the fact that we cluster at the market level, whereas the correlation across markets is asymptotically negligible.

\section{Additional Details for Section~\ref{sec:simulations case 1}}
\label{app:simulations}

In this section we give additional details on the simulation design used to generate results in Section~\ref{sec:simulations case 1} and present additional results not included in the main text.

\subsection{Simulation Design}
\label{app:simulations design}

We first fit a mixed logit model to the data using PyBLP. We use the third model in the PyBLP tutorial. The model has product fixed effects, and random coefficients on price, sugar content, mushy, and a constant for the inside products, with a diagonal covariance matrix. In addition to the 20 instruments provided in the PyBLP data, we estimate the model using two additional differentiation IVs, following \cite{gandhiMeasuringSubstitutionPatterns2019}. For product $j$ in market $t$, we form
\begin{equation}\label{eq:diff ivs}
 \sum_{j' \neq j} \mathbb{I}[|d_{j,j'}| < \hat \sigma_j] (\hat p_{jt} - \hat p_{j't}), 
 \quad \quad 
 \sum_{j' \neq j} \mathbb{I}[|d_{j,j'}| < \hat \sigma_j] (\hat p_{jt} - \hat p_{j't})^2, 
\end{equation}
where $d_{j,j'}$ is the difference between the sugar content of focal product $j$ and other product $j'$, $\hat \sigma_j$ is the standard deviation of $|d_{j,j'}|$ across products $j'$, and $\hat p_{jt}$ are the fitted values from a regression of $p_{jt}$ on $z_t$ across markets (i.e., regression coefficients are estimated at the product level). 

The estimated coefficients are $\bar{\alpha} = -30.1587$, $\sigma_{\alpha} = 0.1431$, $\sigma_{\mbox{\tiny inside}} = 0.0303$,  $\sigma_{\mbox{\tiny sugar}} =0.0365$, and $\sigma_{e} = 0.5272$. The parameter $\sigma_{e}$ is estimated imprecisely, with a standard error of $1.51$. To make task of correcting bias more challenging, we increase $\sigma_{e}$ to $2.0$, which is well inside the 95\% confidence interval for $\sigma_e$ based on the PyBLP estimates. We use these parameter values to generate choices given prices and exogenous variables in the simulations.

Let $\mu_{z,j}$ denote the sample mean of the 20 instruments $z_{jt}^{py}$ provided in the PyBLP data across markets, and $\Sigma_z$ and $\Sigma_{\xi}$ denote the sample variance of $z_{jt}^{py} - \mu_{z,j}$ and $\hat \xi_{jt}$, respectively, across products and markets, where $\hat \xi_{jt}$ are the estimated residuals at the data-generating parameters. For each simulated dataset, we first draw $\varepsilon_{jt}^*$ iid $N(0,\Sigma_z)$ and $f_{0t}^*$ and $f_{1t}^*$ iid $N(0,1)$, and set $z_{jt}^{py*} = \mu_{z,j} + \sqrt{1-\rho_z} \, \varepsilon_{jt}^* + \sqrt{\rho_z \,\mathrm{diag}(\Sigma_z)} f_{e_jt}^*$ with $\rho_z = 0.8$. This ensures that the mean and variances of the 20 instruments $z_{jt}^{py*}$ match those in the data, but with a factor structure that makes the instruments more informative about the nonlinear parameters. We also draw $\xi_{jt}^*$ iid $N(0,\Sigma_{\xi})$. To generate prices, we first run a pooled regression of prices $p_{jt}$ on $z_{jt}^{py}$, $\hat \xi_{jt}$, and product fixed effects. We then set $p_{jt}^* = \hat p_{jt}^* + u_{jt}^*$, where $\hat p_{jt}^*$ denotes the predicted values from this regression at $z_{jt}^{py} = z_{jt}^{py*}$ and $\hat \xi_{jt} = \xi_{jt}^*$, and $u_{jt}^*$ is an iid $N(0,\sigma_p^2)$ error term, where $\sigma_p^2$ is the sample variance of the residuals from the pooled price regression. Given the simulated $p_{t}^* = (p_{jt}^*)_{j=1}^J$ and $\xi_{t}^* = (\xi_{jt}^*)_{j=1}^J$, we then simulate market shares at the above parameter values (with product fixed effects also matching parameter estimates). Finally, for each $j$ and $t$ we form two more differentiation IVs as in~\eqref{eq:diff ivs}. This yields a total of 22 instruments $z_{jt}$ per product, which we use to estimate the parameters and implement our bias corrections.

\subsection{Additional Results with (Potentially) Mismeasured IVs}
\label{app: sims e^tilde IVs}

Here we present simulation results using two additional differentiation IVs based on the (potentially) mismeasured product attributes. In addition to the 22 instruments described above, we use two more based on the proxies:
\[
  \sum_{j' \neq j} \mathbb{I}[d_{j,j'}^e = 0] (\hat p_{jt} - \hat p_{j't}), 
 \quad \quad 
 \sum_{j' \neq j} \mathbb{I}[d_{j,j'}^e = 0] (\hat p_{jt} - \hat p_{j't})^2, 
\]
where $d_{j,j'}^e = \mathbb{I}[\tilde e_j \leq 0.5] - \mathbb{I}[\tilde e_{j'} \leq 0.5]$. When there is no proxy error ($\rho = 0$), these capture differentiation in price for products that are close in ``mushiness.'' However, for $\rho > 0$ these capture differentiation in price for products that are close according to the imperfect proxy $\tilde e_j$. Bias and RMSEs are presented in Figure~\ref{fig:bias rmse sim case 1 e}, and closely match those in Figure~\ref{fig:bias rmse sim case 1} (which excludes the last two differentiation IVs based on proximity in $\tilde e_j$). Coverage of 95\% confidence intervals for the counterfactuals based on the bias-corrected estimator and the variance formula~\eqref{eq:standard errors blp no micro data} are presented in Table~\ref{tab:coverage case 1 e}. These are close to the values reported in Table~\ref{tab:coverage case 1} in the main text.

\begin{figure}[t]
	\centering
	\begin{subfigure}[t]{0.48\textwidth}
		\centering
		\includegraphics[width=\linewidth]{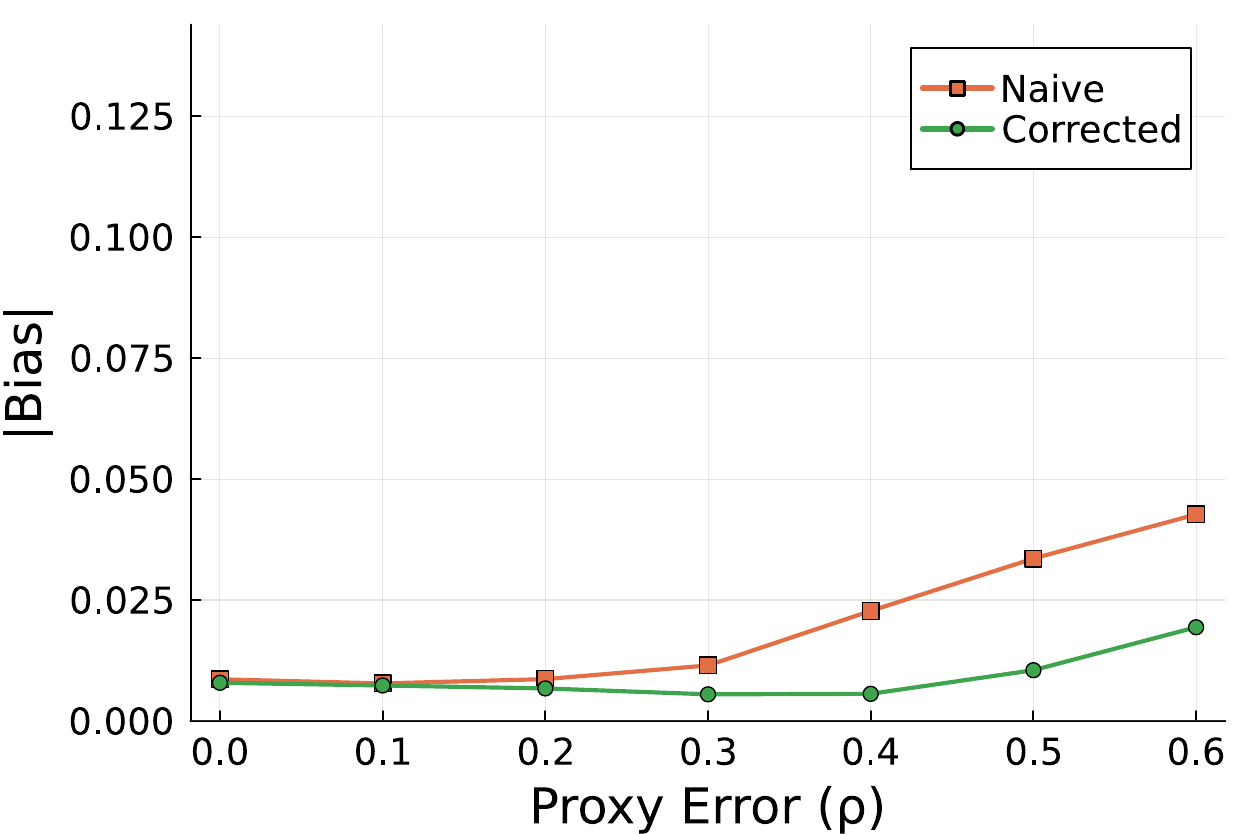}
		\caption{Bias for the ``all'' counterfactual}
	\end{subfigure}
	\hfill 
	\begin{subfigure}[t]{0.48\textwidth}
	\centering
	\includegraphics[width=\linewidth]{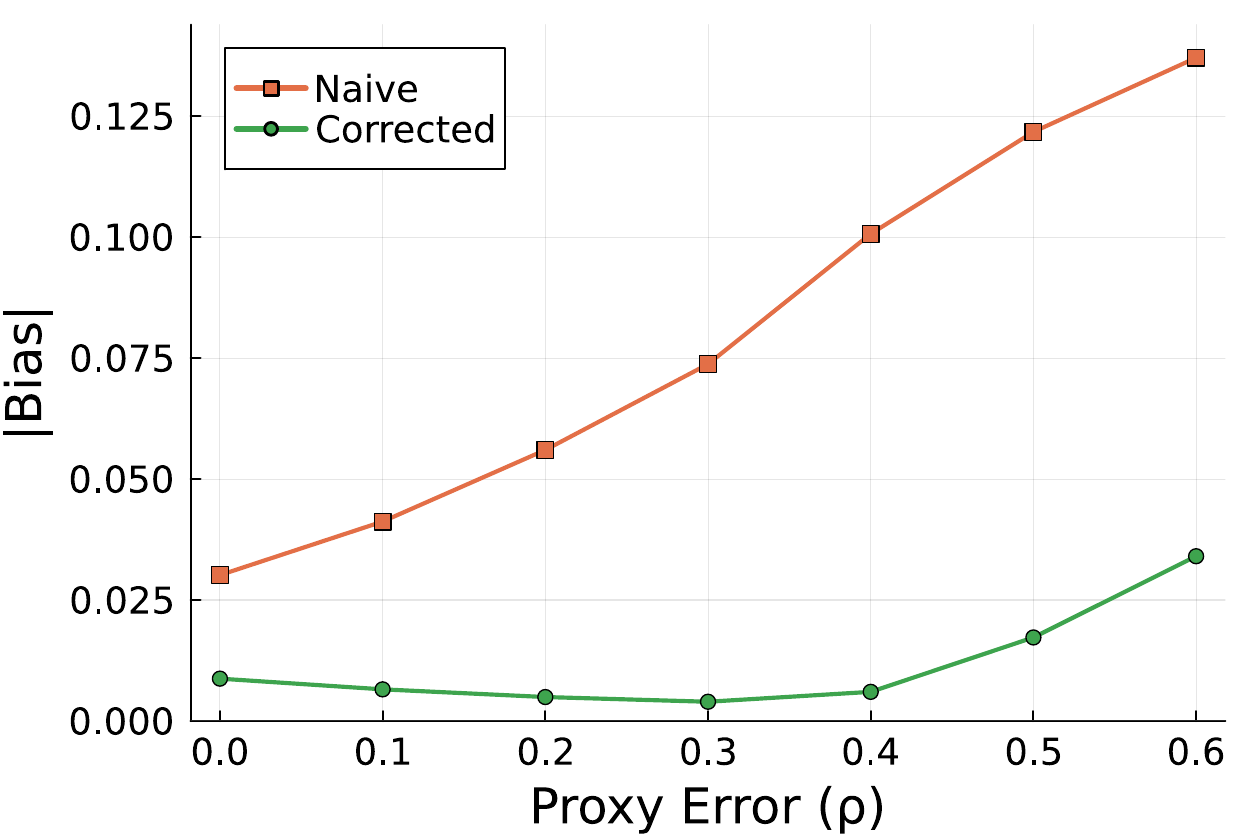}
	\caption{Bias for the ``mushy'' counterfactual}
\end{subfigure}
	
	\par\vspace{1em} 
	
		\begin{subfigure}[t]{0.48\textwidth}
		\centering
		\includegraphics[width=\linewidth]{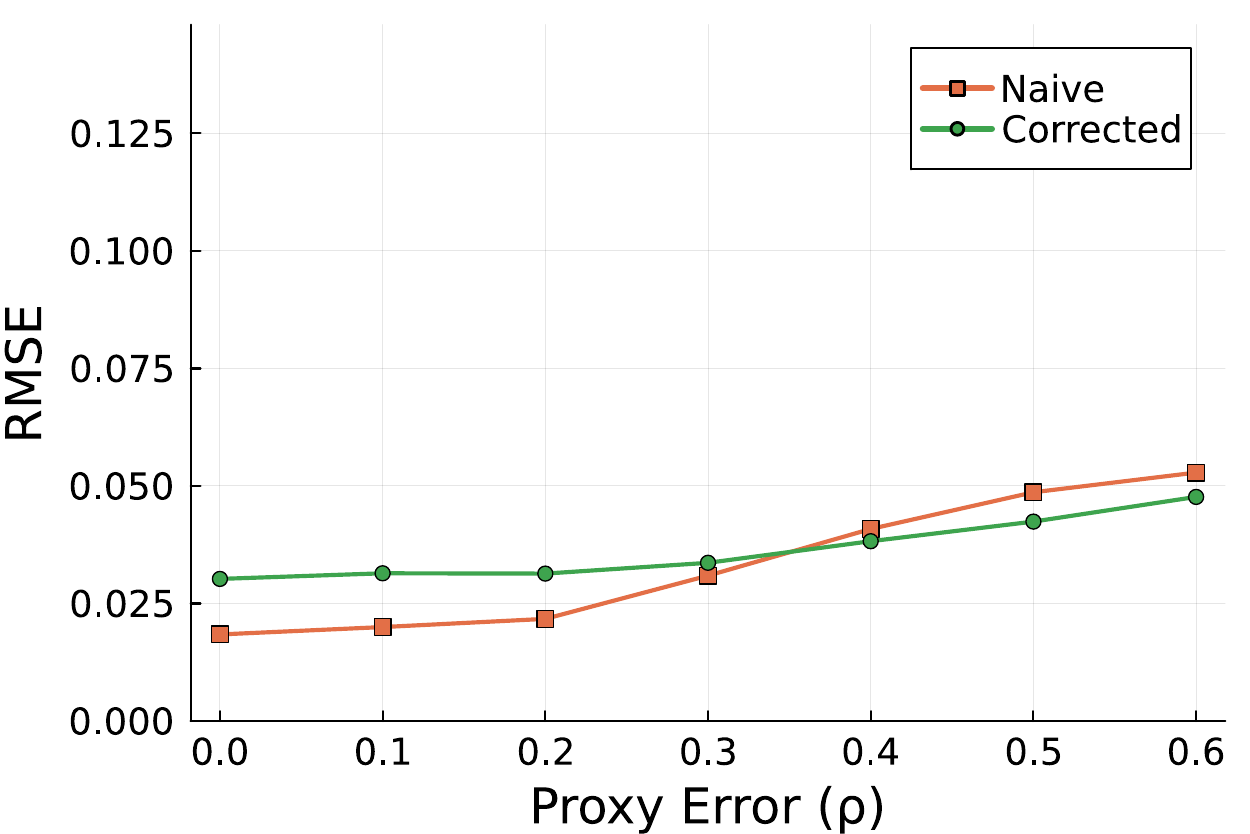}
		\caption{RMSE for the ``all'' counterfactual}
	\end{subfigure}
	\hfill 
	\begin{subfigure}[t]{0.48\textwidth}
	\centering
	\includegraphics[width=\linewidth]{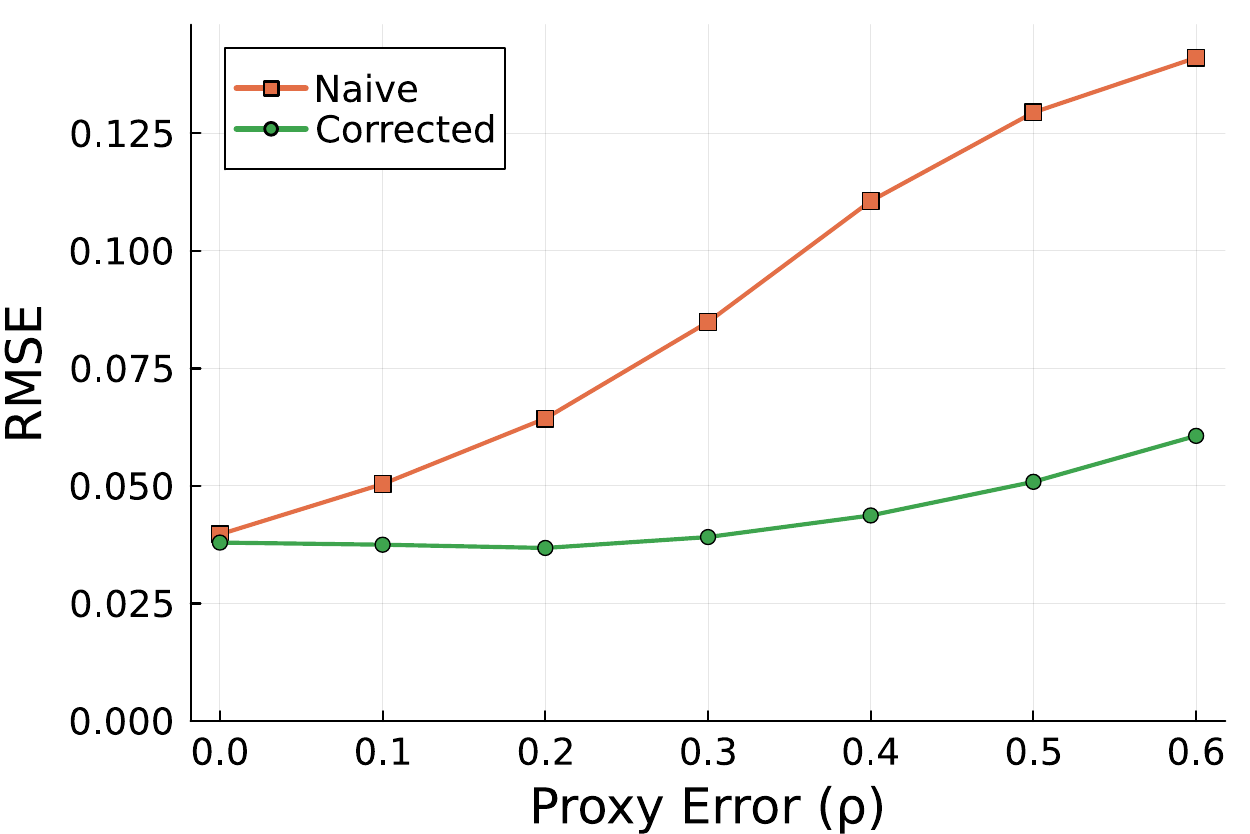}
	\caption{RMSE for the ``mushy'' counterfactual}
\end{subfigure}

	\caption{Bias and RMSE as a function of proxy error, \cite{nevo2001measuring} design}
	\label{fig:bias rmse sim case 1 e}

\vskip 0.5em

\begin{minipage}{\textwidth}
    \small \textit{Note:} Figures show the absolute bias and RMSE across simulations of the naive estimator $\hat \kappa$ in~\eqref{eq: kappa naive blp} and bias-corrected estimator $\hat \kappa_{bc}$ in~\eqref{eq: kappa_bc blp linear no micro} for the ``all'' and ``mushy'' counterfactuals.
\end{minipage}
    
\end{figure}

\begin{table}[H]
    \centering
    \begin{tabular}{ccc} \hline \hline 
    Proxy error $\rho$ & \phantom{111} Coverage: ``all''\phantom{111} & Coverage: ``mushy'' \\  \hline \hline 
    0.00 & 0.958 & 0.962 \\
    0.10 & 0.966 & 0.971 \\
    0.20 & 0.966 & 0.972 \\
    0.30 & 0.959 & 0.967 \\
    0.40 & 0.953 & 0.940 \\
    0.50 & 0.923 & 0.870 \\
    0.60 & 0.899 & 0.769 \\ \hline \hline
    \end{tabular}
    \caption{Coverage of 95\% confidence intervals for counterfactuals, \cite{nevo2001measuring} design}
    \label{tab:coverage case 1 e}

    \vskip 0.5em

  \begin{minipage}{\textwidth}
    \small \textit{Note:} Fraction of simulations in which confidence intervals for counterfactuals contain the true value. Confidence intervals are computed as $\hat \kappa_{bc} \pm 1.96 (\hat V_{bc}/T)^{1/2}$ with $\hat \kappa_{bc}$ in~\eqref{eq: kappa_bc blp linear no micro} and $\hat V_{bc}$ as in Remark~\ref{rmk:std errors}.
  \end{minipage}
  
\end{table}

\section{Additional Results for Section~\ref{sec:application}}
\label{app:empirical}

Figure \ref{fig:hit rate all} compares the hit rates with and without bias correction for additional unstructured data sources and ML models. For book titles and descriptions, we use the same ML models used for book reviews (discussed in Section \ref{sec:application}). For images, we use standard pre-trained classification models (see \citetalias{compiani2025demand} for more details).

For 14 out of 16 proxies (around 88\%), the
bias correction weakly improves performance and the magnitude of the improvement
is large in several cases.
This confirms that the bias correction is helpful in improving the estimator's performance at counterfactual predictions across various unstructured data sources.

\begin{figure}[p]
	\centering
	\includegraphics[width=0.85\linewidth]{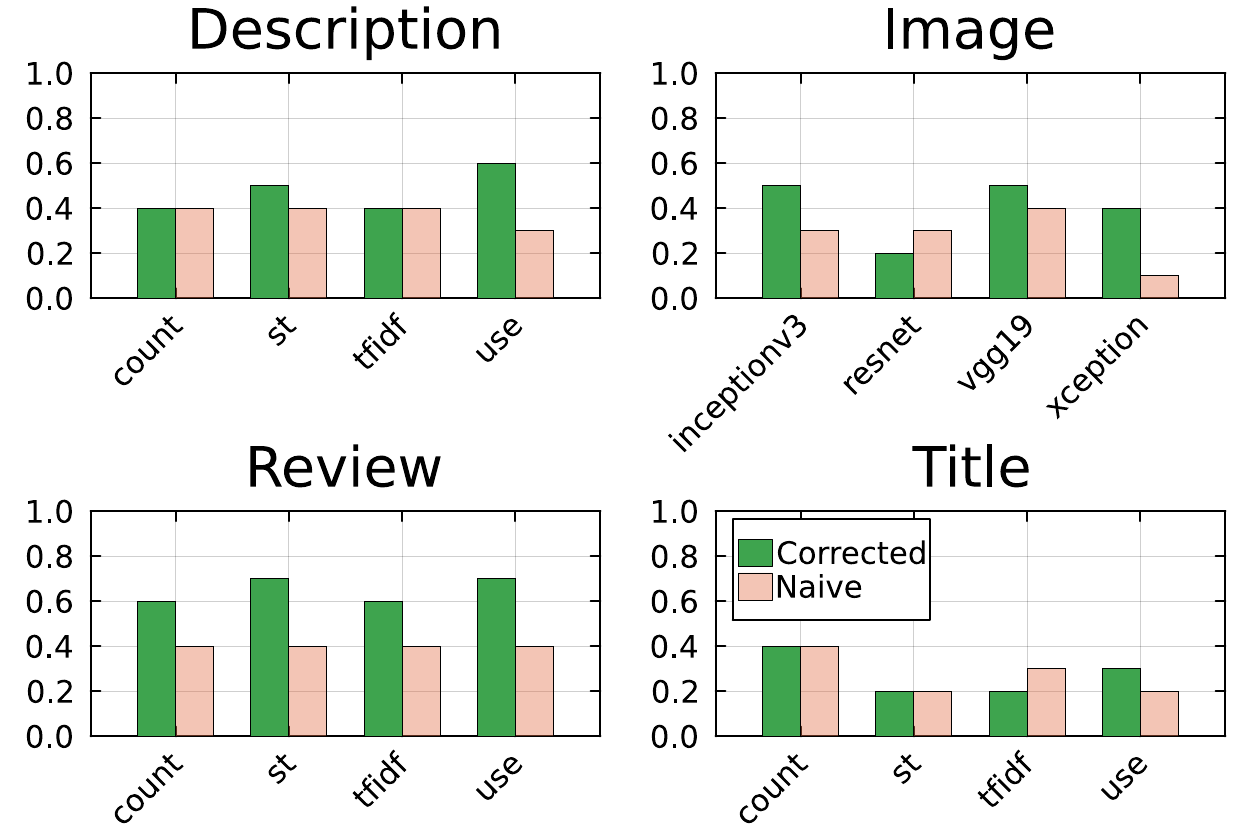}
	\caption{Rates of correct closest-substitute predictions.\label{fig:hit rate all}}
    \parbox[]{\textwidth}{\small \emph{Note: }Darker green bars show the fraction of books for which the bias-corrected estimator correctly identifies the closest substitute. Lighter orange bars show the corresponding fractions for the naive estimator. Each of the four panels corresponds to an unstructured data source. Within each panel, each pair of bars corresponds to an ML model used to extract the proxies from that data source.}
\end{figure}

\section{Reparameterization for Micro BLP}
\label{app:micro BLP}

Consider Example~\ref{ex:micro blp}. We partition $\beta_i = (\beta_{\bar x, i}, \beta_{e,i})$ and $\Pi = [\Pi_{\bar x} \; \Pi_e \; \Pi_p]$, and write the utilities  as:
	\[
	u_{ijt} = \beta'_{\bar x, i} \bar x_{jt}  + \beta'_{e, i} e_{j} -\alpha_i p_{jt} + y_{it}' [\Pi_{\bar x} \; \Pi_{p}] (\bar x_{jt}, p_{jt}) +  y_{it}' \Pi_e e_j + \pi ' \bar y_{ijt} + \xi_{jt} + \varepsilon_{ijt} , \quad j \in \mathcal{J}_t .
	\]
	Suppose $\alpha_i \sim N(\bar \alpha, \sigma^2_\alpha)$, $\beta_{\bar x,i} \sim N(\bar \beta_{\bar x}, \Sigma_{\bar x})$, $\beta_{e,i} \sim N(\bar \beta_e, \Sigma_{e})$, and $\alpha_i$, $\beta_{\bar x,i}$ and $\beta_{e,i}$ are independent. Then, 
	\[
	\theta = \left(\bar \alpha, \sigma_\alpha, \bar \beta_{\bar x}, \bar \beta_e, \pi, v(\Pi_{\bar x}), v(\Pi_p), v(\Pi_e), l(\Sigma_{\bar x}), l(\Sigma_e) \right),
\] 
where  we use the same notation $v$ and $l$ as for Examples~\ref{ex:blp} and~\ref{ex:mixed logit}.\footnote{As with Example~\ref{ex:blp}, if $\Sigma_{\bar x}$ and/or $\Sigma_e$ are diagonal, then we replace $l(\Sigma_{\bar x})$ and/or $l(\Sigma_e)$ with vectors containing their diagonal entries.} Note that  $e_j$ only enters via $\beta'_{e,i} e_j$ and $y_{it}' \Pi_e e_j$. 
As before, collecting $\beta'_{e,i} e_j$ across products, we have $e \beta_{e, i} \sim N(e \bar \beta_e, e \Sigma_e e')$, where $e \Sigma_e e'$ has rank $r \leq J$ because $e$ is $J \times r$. Hence,
	\[
	\gamma(\theta,e) = \left(\bar \alpha, \sigma_\alpha, \bar \beta_{\bar x}, e \bar \beta_e, \pi, v(\Pi_{\bar x}), v(\Pi_p), v(e \Pi_e'), l(\Sigma_{\bar x}), l_r(e \Sigma_{e}e')\right),
	\]
with $l_r$ as in Examples~\ref{ex:blp} and~\ref{ex:mixed logit}.

{\let\oldbibliography\thebibliography
\renewcommand{\thebibliography}[1]{\oldbibliography{#1}
\setlength{\itemsep}{0pt}}
\putbib 
}
\end{bibunit}

\end{document}